%% file: main.tex
\documentclass{article}

\usepackage[utf8]{inputenc}
\usepackage[left=3cm, right=3cm, top=2cm, bottom=2cm]{geometry}

\usepackage[numbers,sort]{natbib}
\usepackage[noend,boxed]{algorithm2e}
\usepackage{mdframed}
\usepackage{graphicx}
\usepackage{amsmath}
\usepackage{amssymb}
\usepackage{amsthm}
\usepackage{authblk}
\usepackage{comment}
\usepackage[usenames, dvipsnames]{xcolor}
\usepackage{tikz, pgfplots}
\usetikzlibrary{decorations.pathreplacing}
\definecolor{DarkRed}{rgb}{0.65,0,0}
\definecolor{Red}{rgb}{1,0,0}
\usepackage[linktocpage=true,
pagebackref=true,colorlinks,
linkcolor=DarkRed,citecolor=ForestGreen,
bookmarks,bookmarksopen,bookmarksnumbered]
{hyperref}
\usepackage{hyperref}
\usepackage{xspace}

\newtheorem{definition}{Definition}
\newtheorem{theorem}{Theorem}
\newtheorem{corollary}[theorem]{Corollary}
\newtheorem{lemma}[theorem]{Lemma}
\newtheorem{hyp}[theorem]{Hypothesis}

\newtheorem{conjecture}[theorem]{Conjecture}
\newtheorem{claim}[theorem]{Claim}

\newcommand{\eps}{\varepsilon}

\newcommand{\dist}{\text{dist}}

\newcommand{\rr}{\mathbb{R}}

\newcommand{\mst}{\texttt{MST}}

\newcommand{\Thetaish}{\widetilde{\Theta}}
\newcommand{\betastar}{\beta^*}
\newcommand{\abeta}{\protect\overrightarrow{\beta}}

\newcommand{\DMC}{\texttt{DMC}\xspace}
\newcommand{\MMF}{\texttt{MMF}\xspace}
\newcommand{\DSC}{\texttt{DSC}\xspace}
\newcommand{\DSF}{\texttt{DSF}\xspace}

\newcommand{\spo}{\texttt{SPO}}
\newcommand{\ppo}{\texttt{PO}}

\newcommand{\Oish}{\widetilde{O}}
\newcommand{\Omegaish}{\widetilde{\Omega}}
\newcommand{\rs}{\texttt{rs}}

\newcommand{\rpp}{\texttt{RP}}

\newcommand{\dpp}{\texttt{DP}}
\newcommand{\apdpp}{\texttt{ADP}}

\newcommand{\sss}{\texttt{SS}}
\newcommand{\gamgam}{\gamma\gamma}
\newcommand{\fcg}{\texttt{MCG}}
\newcommand{\scg}{\texttt{SCG}}
\newcommand{\dsfg}{\texttt{DSFG}}

\newcommand{\ms}{\texttt{MS}}
\newcommand{\lms}{\texttt{LMS}}
\newcommand{\ehh}{\texttt{EH}}

\newcommand{\Althofer}{Alth\"{o}fer}

\newcommand{\Erdos}{Erd\"{o}s}
\newcommand{\Szekely}{Sz\'{e}kely}
\newcommand{\Szemeredi}{Szemer\'{e}di}
\newcommand{\Turan}{Tur\'{a}n}

\usepackage{booktabs}
\usepackage{longtable}
\usepackage{multirow}

\usepackage{placeins}

\usepackage{thmtools}
\declaretheorem[numberlike=theorem]{observation}

\newcommand{\mcut}{\textsc{MCut}}
\newcommand{\fmcut}{\widehat{\textsc{MCut}}}
\newcommand{\mmflow}{\textsc{MMFlow}}
\newcommand{\vmcut}{\textsc{VMCut}}
\newcommand{\vfmcut}{\widehat{\textsc{VMCut}}}

\usepackage{environ}
\makeatletter
\newsavebox{\measure@tikzpicture}
\NewEnviron{scaletikzpicturetowidth}[1]{%
  \def\tikz@width{#1}%
  \def\tikzscale{1}\begin{lrbox}{\measure@tikzpicture}%
  \BODY
  \end{lrbox}%
  \pgfmathparse{#1/\wd\measure@tikzpicture}%
  \edef\tikzscale{\pgfmathresult}%
  \BODY
}
\makeatother

\title{Bridge Girth: A Unifying Notion in Network Design\footnote{This work was supported by NSF:AF 2153680.}}

\author[1]{Greg Bodwin}
\author[1]{Gary Hoppenworth}
\author[2]{Ohad Trabelsi}
\affil[1]{University of Michigan EECS. \texttt{\{bodwin,garytho\}@umich.edu}}
\affil[2]{Toyota Technological Institute at Chicago. \texttt{ohadt@ttic.edu}\thanks{Work partially done at University of Michigan, and partially supported by the NSF grant CCF-1815316 and the NWO VICI grant 639.023.812.}}
\date{}

\begin{document}

\maketitle

\thispagestyle{empty}

\begin{abstract}
A classic 1993 paper by Alth{\" o}fer et al.\ proved a tight reduction from spanners, emulators, and distance oracles to the extremal function $\gamma$ of high-girth graphs.
This paper initiated a large body of work in network design, in which problems are attacked by \emph{reduction} to $\gamma$ or the analogous extremal function for other girth concepts.
In this paper, we introduce and study a new girth concept that we call the \emph{bridge girth of path systems}, and we show that it can be used to significantly expand and improve this web of connections between girth problems and network design.
We prove two kinds of results:

\begin{itemize}

\item We write the maximum possible size of an $n$-node, $p$-path system with bridge girth $>k$ as $\beta(n, p, k)$, and we write a certain variant for ``ordered'' path systems as $\beta^*(n, p, k)$.
We identify several arguments in the literature that implicitly show upper or lower bounds on $\beta, \beta^*$, and we provide some polynomial improvements to these bounds.
In particular, we construct a tight lower bound for $\beta(n, p, 2)$, and we polynomially improve the upper bounds for $\beta(n, p, 4)$ and $\beta^*(n, p, \infty)$.

\item We show that many state-of-the-art results in network design can be recovered or improved via black-box reductions to $\beta$ or $\beta^*$.
Examples include bounds for distance/reachability preservers, exact hopsets, shortcut sets, the flow-cut gaps for directed multicut and sparsest cut, an integrality gap for directed Steiner forest. 
\end{itemize}

We believe that the concept of bridge girth can lead to a stronger and more organized map of the research area.
Towards this, we leave many open problems related to both bridge girth reductions and extremal bounds on the size of path systems with high bridge girth.

\end{abstract}
\clearpage
\setcounter{tocdepth}{2}

\tableofcontents

\thispagestyle{empty}
\setcounter{page}{0}

\input{intro}


\input{pathprelims}

\input{bounds}

\input{reductions}

\input{integralitygaps}

\section*{Acknowledgments}

We are grateful to Omer Reingold, Vivek Madan, Matthew Fahrbach, and Idan Shabat for helpful technical discussions.

\bibliographystyle{plain}
\bibliography{refs}

\appendix

\input{girthtour}
\input{cleaning}

\input{integralitygaps_app}

\input{implicitbounds}

\end{document}

%% file: intro.tex

\clearpage
\section{Introduction}

A common goal in theoretical computer science is to compress a graph into a small-space representation while approximately preserving structural information related to shortest paths, distances, or reachability.
Examples include spanners \cite{PU89jacm, PU89sicomp, ADDJS93, ACIM99, BS07, EP04, TZ06, Chechik13soda, Pettie09, Woodruff06, Woodruff10, BKMP10, AB17jacm, ABP17, Knudsen17, Knudsen14, AlDhalaan21}, emulators \cite{DHZ00, EP04, TZ05}, distance oracles \cite{TZ01, EP16, Chechik14, ENW16, WulffNilsen12, RTZ05}, distance and reachability preservers \cite{BCE05, CE06, Bodwin21, Bodwin19, AB18, BCR16, CC20, CCC22}, hopsets \cite{EN20, BP20, HP19, EN16, EN17, MPVX15, Cohen00}, shortcut sets \cite{Hesse03, Thorup92, UY91, Fineman19, LJS19, HP18, KS21, KP21, KP22}, etc.; see survey \cite{ABSHJKS20} for more.
We shall broadly refer to this research area as \emph{network design}.

A successful strategy has been to reduce network design problems to \emph{girth problems} in extremal combinatorics.
Generally speaking, a girth problem asks for the maximum possible size of a combinatorial system that avoids short ``cycles'' of some kind.
The contribution of this paper is to introduce a new girth problem, based on a particular kind of cycle in \emph{path systems} that we call ``bridges.''
We then use our new girth problem to organize and improve the understanding of several well-studied problems in network design.
This paper contains two kinds of results:
\begin{enumerate}
\item We polynomially improve upper and lower bounds on the maximum possible size of path systems of high bridge girth (over bounds implicit in the previous literature), and 

\item We show reductions from various problems in network design to our new bridge girth problem, and use them to recover or improve state-of-the-art upper and/or lower bounds.
\end{enumerate}


\begin{figure} [ht]
\begin{scaletikzpicturetowidth}{\textwidth}
\begin{tikzpicture}[scale=\tikzscale]

\node (girth) at (12, 0) [draw, thick, align=center] {
\large Girth Problem $\gamma(n, k)$\\
\footnotesize mult.\ spanners, emulators,\\
\footnotesize distance oracles (upper \& lower bds) \cite{ADDJS93}\\
\footnotesize vertex fault tolerant mult.\ spanners and vertex\\
\footnotesize dist.\ sensitivity oracles (upper \& lower bds) \cite{BP19, BDPV18}};

\node (wgirth) at (12, -4) [draw, thick, align=center] {
\large Weighted Girth Problem $\lambda(n, k)$\\
\footnotesize light mult.\ spanners (upper \& lower bds) \cite{ENS14}};

\node (bpt) at (0, 0) [draw, thick, align=center] {
\large Bipartite Girth Problem $\gamgam(n,p, k)$\\
\footnotesize edge fault tolerant mult.\ spanners and \\
\footnotesize edge dist.\ sensitivity oracles (lower bds) \cite{BDR22}\\
\footnotesize comm.\ compl.\ of mult.\ spanners (lower bds) \cite{FWY20}};

\node (set) at (0, 4) [draw, thick, align=center] {
\large Set Girth Problem $\Sigma(n, p, k)$};

\node (rs) at (12, 4) [draw, thick, align=center] {
\large Ruzsa-\Szemeredi{} Problem $\rs(n)$\\
\footnotesize undir.\ unweighted dist.\ preservers (upper bds) \cite{Bodwin21}\\
\footnotesize undir.\ unweighted\ dist.\ labeling schemes (upper bds) \cite{KUV19}};

\node (ordbridge) at (12, 9) [draw, thick, align=center, blue] {
\large Ordered Bridge Girth Problem $\beta^*(n, p, k)$\\
\footnotesize dir.\ weighted dist.\ preservers and\\
\footnotesize shortest path oracles (lower bds)\\
\footnotesize dir.\ weighted exact hopsets (lower bds)\\
\footnotesize online reachability preservers (upper \& lower bds)};

\node (bridge) at (0, 9) [draw, thick, align=center, blue] {
\large Bridge Girth Problem $\beta(n, p, k)$\\
\footnotesize reachability preservers and\\
\footnotesize path oracles (upper \& lower bds)\\
\footnotesize dir.\ weighted dist.\ preservers and\\
\footnotesize shortest path oracles (upper bds)\\
\footnotesize shortcut sets (lower bds) \\
\footnotesize dir.\ flow-cut gap (lower bds)\\
\footnotesize sparsest cut gap (lower bds)\\
\footnotesize DSF integrality gap (lower bds)};

\draw [thick, ->] (bpt) -- (girth) node [above, midway, align=center]{
\footnotesize case $p=n$\\
\footnotesize (folklore)};

\draw [thick, <->] (bpt) -- (set) node [left, midway, align=center] {\footnotesize equivalent\\
\footnotesize $\gamgam(n, p, 2k) = \Sigma(n, p, k)$\\
\footnotesize (folklore)};

\draw [thick, ->] (bpt) -- (rs) node [above, midway, sloped, align=center]{
\footnotesize case $k=6$,\\
\footnotesize $p$ large \cite{de1991maximum}};

\draw [thick, ->] (set) -- (rs) node [above, midway, align=center] {
\footnotesize case $k=3$,\\
\footnotesize $p$ large \cite{de1991maximum}};

\draw [blue, thick, ->] (bridge) -- (rs) node [above, midway, sloped, align=center] {
\footnotesize case $k=3$,\\
\footnotesize $p$ large \cite{de1991maximum} (see Thm \ref{thm:b3rs})};

\draw [blue, thick, ->] (set) -- (bridge) node [left, midway, align=center] {
\footnotesize directed version\\
\footnotesize $\Sigma \le \beta$};

\draw [blue, thick, ->] (bridge) -- (ordbridge) node [above, midway, align=center] {
\footnotesize ordered version,\\
\footnotesize $\beta \le \beta^*$};

\draw [thick, <->] (girth) -- (wgirth) node [midway, right, align=center] {
\footnotesize approx equivalent \cite{LS22}\\
\footnotesize conjectured fully equivalent \cite{ENS14}};

\end{tikzpicture}
\end{scaletikzpicturetowidth}
\caption{\label{fig:introgirth} Relationships among some girth problems in the literature, and the problems they capture.  Bridge girth, ordered bridge girth, and the associated reductions -- all in blue -- are new in this paper.  See Appendix \ref{app:girthtour} for more detail on the prior work reflected in this chart.
}
\end{figure}
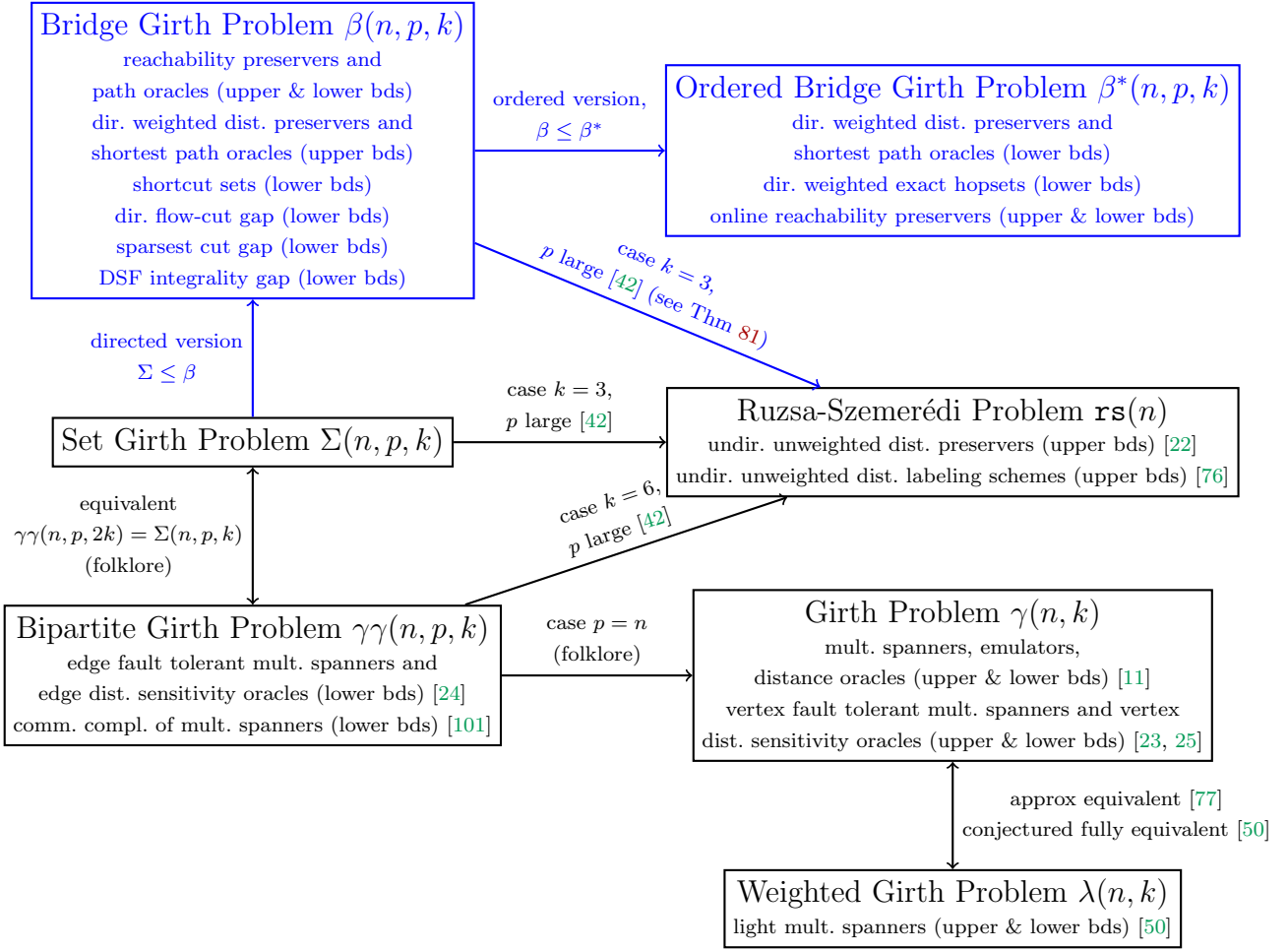

\subsection{Bridge Girth and the Landscape of Girth Reductions}

\subsubsection{Previously-Studied Girth Concepts}

In Appendix \ref{app:girthtour}, we provide a detailed tour through all the objects mentioned in Figure \ref{fig:introgirth}.
Here we give a much faster overview, to show how our new concept of bridge girth relates to previously-studied girth concepts in the literature.

The strategy of reducing network design problems to girth problems was pioneered in a classic 1993 paper by Alth{\"o}fer, Das, Dobkin, Joseph, and Soares \cite{ADDJS93}.
They provided an exactly-tight reduction between the extremal functions of \emph{spanners} and \emph{high-girth graphs}.
We will rephrase their main result a bit, to give an example of the reduction-based perspective used in this paper.
\begin{definition} [Multiplicative Spanners]
A multiplicative $k$-spanner of a graph $G$ is a subgraph $H$ satisfying $\dist_H(s, t) \le k \cdot \dist_G(s, t)$ for all nodes $s, t$.
The function $\ms(n, k)$ is the least integer such that every undirected weighted $n$-node graph has a $k$-spanner on $\le \ms(n, k)$ edges.
\end{definition}

\begin{definition} [Graph Girth]
The girth of a graph $G$ is the least number of edges in a cycle in $G$ (or $\infty$ if $G$ is a forest).
The function $\gamma(n, k)$ is the maximum possible number of edges in an $n$-node graph of girth $>k$.
\end{definition}

\begin{theorem} [\cite{ADDJS93}] \label{thm:msintro}
$\gamma(n, k) = \ms(n, k+1)$.
\end{theorem}

It remains a major open question in extremal combinatorics to determine the asymptotic value of $\gamma$.
Regardless, by tight reduction, Theorem \ref{thm:msintro} is considered by the community to close the question of the existential size of multiplicative spanners.
It has been highly influential in network design, spawning a long line of work similarly reducing spanner or spanner-like problems to $\gamma$ \cite{RZ11, BDPV18, BP19, DR20, ENS14, FS16, CW16, LS22}.

More recently, some variants and extensions of the function $\gamma$ have emerged as similarly fundamental in network design.
One example is an elegant paper by Elkin, Neiman, and Solomon, which developed a notion of ``weighted girth,'' and showed that the corresponding extremal function is equivalent to the tradeoff between stretch and \emph{lightness} for spanners \cite{ENS14} (see Section \ref{sec:wtdgirth} for details).
Another example is a line of this work \cite{FWY20, BDR22, de1991maximum} that has developed reductions to $\gamgam$, a generalization of $\gamma$ to bipartite graphs:
\begin{definition} [Bipartite Graph Girth]
The function $\gamgam(n, p, k)$ is the maximum possible number of edges in a bipartite graph with $n, p$ nodes on each side of its bipartition and girth $>k$.
\end{definition}

There is a folklore reduction showing that $\gamma(n, k) = \Theta(\gamgam(n, n, k))$ (see Theorem \ref{thm:ggtog}), and so any reductions to $\gamma(n, k)$ can be equivalently phrased as a reduction to $\gamgam(n, n, k)$.
The importance of $\gamgam$ was further shown in an important paper by de Caen and \Szekely{} \cite{de1991maximum}, which proves an equivalence between a special case of $\gamgam$ and the \emph{Ruzsa-\Szemeredi{} function} (see Appendix \ref{app:rs}), another extremal function that captures various problems in network design \cite{KUV19, Bodwin21}.
The function $\gamgam$ is sometimes expressed in the equivalent language of \emph{set system girth}:
\begin{definition} [Set Systems]
A set system is a pair $S = (V, \mathcal{T})$, where $V$ is a ground set of nodes and $\mathcal{T}$ is a multiset of subsets of $V$.
The size of $S$ is given by
$$\|S\| := \sum \limits_{T \in \mathcal{T}} |T|.$$
\end{definition}

\begin{definition} [Set System Girth]
A $k$-cycle in a set system $S = (V, \mathcal{T})$ is a circularly-ordered list of distinct nodes $v_0, v_1, \dots, v_k=v_0 \in V$ and distinct sets $T_0, T_1, \dots, T_k=T_0 \in \mathcal{T}$ for which we have $v_i, v_{i+1} \in T_i$ for all $i$.
The girth of a set system is the smallest integer $k$ for which the system has a $k$-cycle.
The maximum possible size of a set system with $n$ nodes, $s$ sets, and girth $>k$ is written $\Sigma(n, s, k)$.
\end{definition}

For example, a set system in which each set has size $2$ can be viewed as an undirected graph.\footnote{A set system is equivalent to a (not necessarily uniform) hypergraph.  We call these set systems rather than hypergraphs (1) to emphasize the way in which path systems can be viewed as a directed variant, and (2) because there are several competing notions of hypergraph size/girth in the literature, but these terms are unambiguous for set systems.}
Set systems are in natural bijection with their bipartite \emph{incidence graphs}, and this implies the folklore equivalence $\Sigma(n, p, k) = \gamgam(n, p, 2k)$ (see Theorem \ref{thm:ggtoset}).
Thus, we can consider the set system girth problem as merely a rephrasing of the bipartite girth problem.


\subsubsection{Bridge Girth and Ordered Bridge Girth}

A directed version of a set system is a \emph{path system}, in which we have node sequences instead of sets:

\begin{definition} [Path Systems]
A path system is a pair $S = (V, \Pi)$ where $V$ is a ground set of nodes and $\Pi$ is a multiset of vertex sequences called paths.
Each path may contain at most one instance of each node.
The size of a path system is written\footnote{Note that $|\pi|$ counts the number of \emph{nodes} in $\pi$, and so it differs by $1$ from the length of $\pi$ when viewed as a path through a graph.}
$$\|S\| := \sum \limits_{\pi \in \Pi} |\pi|.$$
\end{definition}

For example, a path system in which all paths have length $2$ is essentially a directed graph.
Our new girth concept is based on the following notion of a ``cycle'' in a path system:

\begin{definition} [$b$-Bridges]
In a path system $S = (V, \Pi)$, a $b$-bridge is a set of $b$ distinct nodes $v_1, \dots, v_b \in V$ and $b$ distinct paths $\pi_1, \dots, \pi_b$ such that (1) for all $1 \le i \le b-1$ we have $v_i, v_{i+1} \in \pi_i$ with $v_i$ preceding $v_{i+1}$, and (2) we have $v_1, v_b \in \pi_b$ with $v_1$ preceding $v_b$.
The path $\pi_b$ is called the river, and the other paths $\pi_1, \dots, \pi_{b-1}$ are called arcs.
\end{definition}

Note that the nodes $v_i, v_{i+1}$ are not necessarily consecutive on their arc $\pi_i$; there might be many nodes between these, and it still counts as a bridge.
Informally, a $b$-bridge resembles a directed $b$-cycle with one of the path directions reversed; the reversed path is called the river, and the non-reversed paths are called the arcs.
See Figure \ref{fig:4bridge} for an example.

\begin{figure}[htbp]
  \begin{center}
    \includegraphics[width=1.5in]{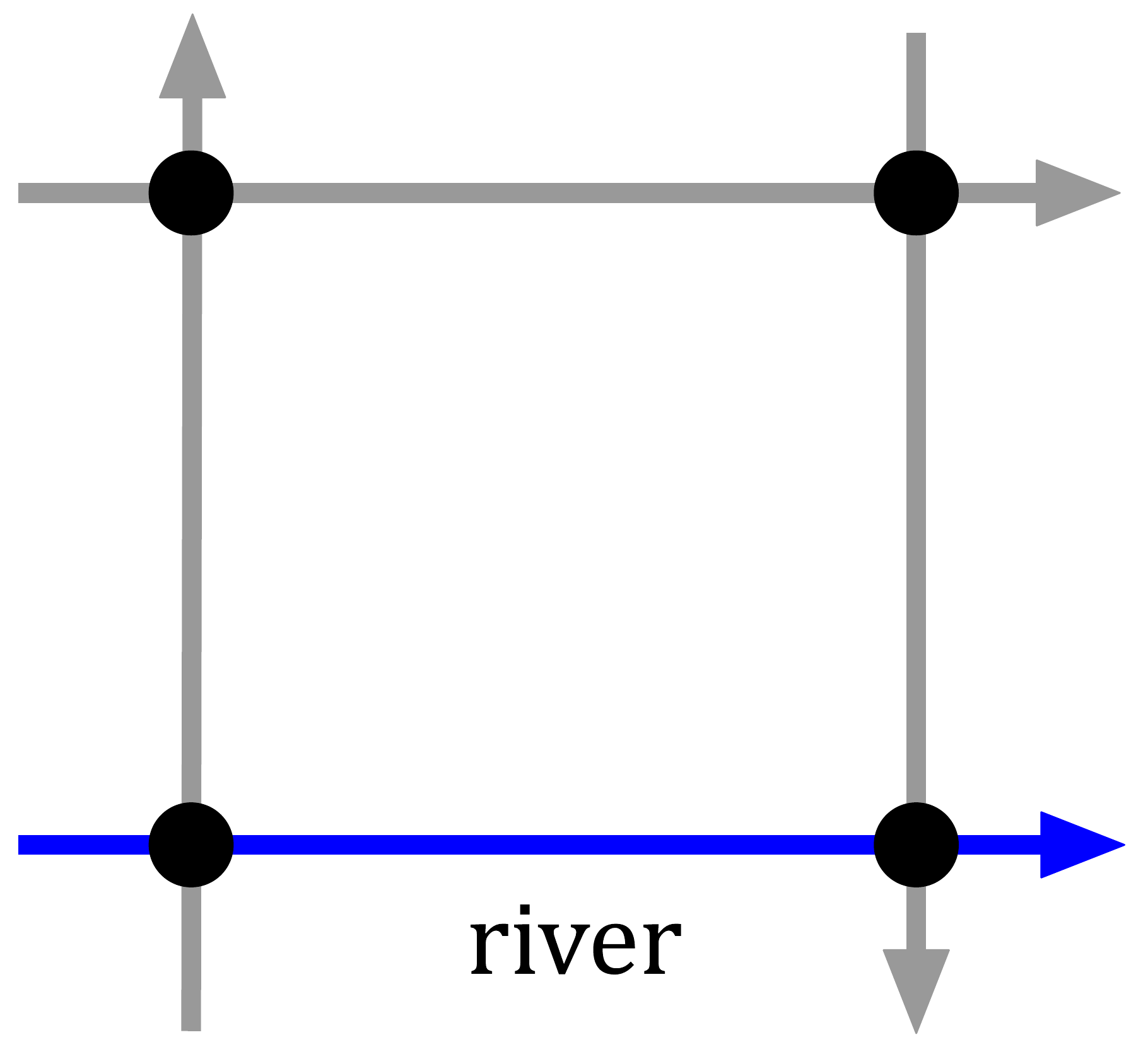}
    \end{center}
    \caption{\label{fig:4bridge} A $4$-bridge}
\end{figure}

\begin{definition} [Bridge Girth]
The bridge girth of a path system $S$ is the smallest integer $b$ for which $S$ contains a $b$-bridge.
The function $\beta(n, p, k)$ is the maximum possible size of a path system with $n$ nodes, $p$ paths, and bridge girth $>k$.
We allow $k=\infty$, meaning that the system has no bridges of any size.
\end{definition}

We will also consider a related notion of bridge girth, based on path systems with an ordering on their paths.
\begin{definition} [Ordered Path Systems and Ordered Bridges]
An ordered path system is a path system $S = (V, \Pi)$ equipped with a total ordering of its paths $\Pi$.
An ordered bridge in an ordered path system is a bridge in which the river comes after all the arcs in the ordering (and the arcs may occur in any order relative to each other).
\end{definition}

\begin{definition} [Ordered Bridge Girth]
The ordered bridge girth of an ordered path system $S$ is the smallest integer $b$ for which $S$ has an ordered $b$-bridge (with the river last in the ordering).
The function $\betastar(n, p, k)$ is the maximum possible size of an ordered path system with $n$ nodes, $p$ paths, and ordered bridge girth $>k$.
\end{definition}

It is immediate from the definitions that $\betastar(n, p, k) \ge \beta(n, p, k)$, since an ordered path system of ordered bridge girth $>k$ is a strictly less constrained object than an (unordered) path system of (unordered) bridge girth $>k$.

We discuss our definitions before proceeding to their applications.
Notice that a bridge is a directed version of a set cycle, in the sense that a $b$-bridge becomes a set $b$-cycle if we forget the order of each path and interpret it as a set.
However, there are many other patterns besides bridges that correspond to set cycles in the same way.
Perhaps the most natural alternative is a directed cycle, defined like a bridge but with the paths/nodes circularly ordered instead of having a river with reversed direction (see Figure \ref{Figs:bridgepic}).
Why focus on bridges rather than directed cycles?

\begin{figure}[htbp]
  \begin{center}
    \includegraphics[width=4.5in]{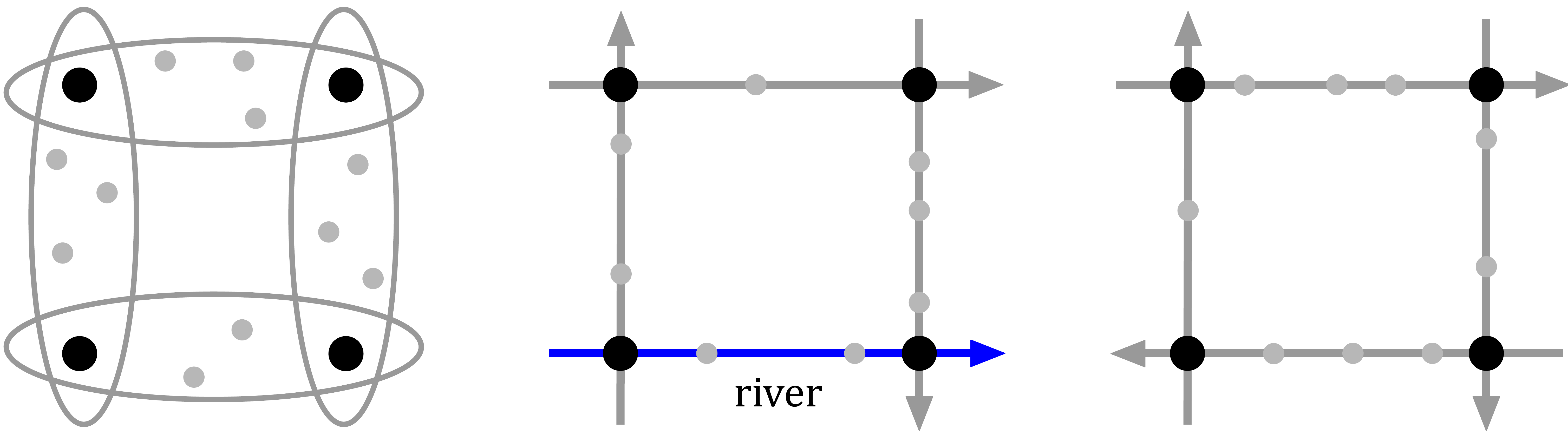}
    \end{center}
    \caption{A set system $4$-cycle, a $4$-bridge, and directed $4$-cycle (left to right).}
    \label{Figs:bridgepic}
\end{figure}

In part, this paper is a lengthy answer to this question.
Our \emph{point} is that, out of the many possible directed versions of set cycles, bridges are the ones that realize applications in network design.
Table \ref{tbl:combo} summarizes our findings to this effect, listing objects where state-of-the-art bounds can be improved, or recovered in a black-box manner, via bridge girth reductions.

\begin{table}[t]
\begin{tabular}{cllc}
\hline
\textbf{Object} & \textbf{Bound} & \textbf{Implied By} & \textbf{First Proved}\\
\hline
\multirow{2}{10em}{\centering Reachability Preservers} & $\rpp \ge \Omega(n^{\frac{2}{d+1}} p^{\frac{d-1}{d}}), d$ any pos int & Thm \ref{thm:rpreduction}, Cor \ref{cor:binftylb} & \cite{AB18, CE06}\\
& $\rpp \le O(n^{3/4} p^{1/2} + n^{5/8} p^{11/16} + n)$ & Thms \ref{thm:rpreduction}, \ref{thm:bnp4upper} & \textbf{New}\\
\hline
\multirow{3}{7em}{\centering Online Reachability Preservers} & $\rpp^* \ge \Omega(n^{2/3} p^{2/3} + n)$  & Thms \ref{thm:rpstarreduction}, \ref{thm:bstarlower} & \textbf{New}\\
 & $\rpp^* \le O(\min\{n^{1/2} p, np^{1/2}\} + n)$  & Thms \ref{thm:rpstarreduction}, \ref{thm:bstarupper} & \textbf{New}\\
 & $\rpp^* = \Thetaish(n^{4/3})$ conditionally & Hyp \ref{hyp:ordered_gap}, Thms \ref{thm:rpstarreduction}, \ref{thm:hypbstarlb} & \textbf{New}\\
\hline
\multirow{2}{*}{Path Oracles} & $\ppo \ge \Omega(n^{\frac{2}{d+1}} p^{\frac{d-1}{d}}), d$ any pos int & Thm \ref{thm:pathoracle}, Cor \ref{cor:binftylb} & \textbf{New}\\
& $\ppo \le \Oish(n^{3/4} p^{1/2} + n^{5/8} p^{11/16} + n)$ & Thms \ref{thm:pathoracle}, \ref{thm:bnp4upper} & \textbf{New}\\
\hline
\multirow{2}{*}{Shortcut Sets} & \multirow{2}{*}{$\sss \ge \Omega(n^{\frac{2}{d+1}} p^{\frac{-1}{d}}), d$ any pos int} & \multirow{2}{*}{Thm \ref{thm:sslb}, Cor \ref{cor:binftylb}} & \cite{HP18} ($p=n$)\\
& & & \textbf{New} (other $p$)\\
\hline
\hline
\multirow{2}{8em}{\centering Distance Preservers} & $\dpp \ge \Omega(n^{2/3}p^{2/3} + n)$ & Thms \ref{thm:distpreschain}, \ref{thm:bstarlower} & \cite{CE06}\\
& $\dpp \le O(\min\left\{n^{2/3} p, np^{1/2} \right\} + n)$ & Thms \ref{thm:distpreschain}, \ref{thm:cebp2}, \ref{thm:bodbnp2} & \cite{Bodwin21}\\
\hline
\multirow{2}{8em}{\centering Shortest Path Oracles} & $\spo \ge \Omega(n^{2/3}p^{2/3} + n)$ & Thms \ref{thm:shortestpathoracle}, \ref{thm:bstarlower} & \textbf{New}\\
& $\spo \le \Oish(\min\left\{n^{2/3} p, np^{1/2} \right\} + n)$ & Thms \ref{thm:shortestpathoracle}, \ref{thm:cebp2}, \ref{thm:bodbnp2} & \cite{CE06, Bodwin21}\\
\hline
\multirow{2}{*}{Exact Hopsets} & \multirow{2}{*}{$\ehh \ge \Omega\left( n^{2/3} p^{-1/3} \right)$} & \multirow{2}{*}{Thms \ref{thm:exacthopsets}, \ref{thm:bstarlower}} & \cite{KP22} ($p=n$)\\
& & & \textbf{New} (other $p$)\\
\hline
\hline
\multirow{2}{7em}{\centering Directed Flow-Cut Gap} &  \multirow{2}{*}{$\fcg \ge \widetilde{\Omega}\left(n^{1/7}\right)
$} & \multirow{2}{*}{Thm \ref{thm:fcg}, Cor \ref{cor:binftylb}} & \multirow{2}{*}{\cite{chuzhoy2009polynomial}} \\
& & & \\
\hline
Sparsest Cut Gap & $\scg \ge \widetilde{\Omega}\left(n^{1/7}\right)$ & Thm \ref{thm:fcg_sc}, Cor \ref{cor:binftylb} & \cite{chuzhoy2009polynomial} \\
\hline
\multirow{2}{10em}{\centering Directed Steiner Forest Gap} & \multirow{2}{11em}{$\dsfg \geq \Omega (n^{1/2 - o(1)})$} & \multirow{2}{*}{Thm \ref{thm:dsf_gap}, Cor \ref{cor:binftylb}} & \multirow{2}{*}{\textbf{New}}\\
& & & \\
\hline

\end{tabular}
\caption{Quantitative bounds for the problems in network design considered in this paper, implied by combining our bridge girth reductions with our bounds on $\beta, \beta^*$.  For brevity the parameters $(n, p)$ on the functions are omitted. See also Figure~\ref{Figs:Bridge}. \label{tbl:combo}}
\end{table}

But a more succinct answer is that directed cycles do not have an interesting extremal function.
One can construct a path system of size $\|S\| = np$ without directed cycles, by taking $p$ identical paths that all proceed through a sequence of the $n$ nodes in the same order.
This system would have no directed cycles (but it would have many $2$-bridges).
In contrast, we shall see shortly that the extremal function $\beta$ for path systems of high bridge girth is highly nontrivial.

\subsection{Reductions from Network Design to Bridge Girth}

\begin{table} [t]
\begin{center}
\begin{tabular}{lll}
\hline
\textbf{Object} & \textbf{Reduction} & \textbf{Theorem}\\
\hline
\textbf{Reachability Preservers} & $\rpp(n, p) = \Theta\left(\beta(n, p, \infty)\right)$ & Thm \ref{thm:rpreduction}\\
Online Reachability Preservers & $\rpp^*(n, p) = \Theta\left( \betastar(n, p, \infty) \right)$ & Thm \ref{thm:rpstarreduction}\\
Path Oracles & $\ppo(n, p) = \Thetaish(\beta(n, p, \infty))$ & Thm \ref{thm:pathoracle}\\
Shortcut Sets & $\sss(n, p) = \Omega\left( \frac{\beta(n, p, \infty)}{p} \right)$ & Thm \ref{thm:sslb}\\
\hline
\textbf{Distance Preservers} & $\Omega(\betastar(n, p, \infty)) \le \dpp(n, p) \le \betastar(n, p, 2)$ & Thm \ref{thm:distpreschain}\\
Shortest Path Oracles & $\Omega\left(\betastar(n, p, \infty)\right) \le \spo(n, p) \le \Oish\left(\betastar(n, p, 2)\right)$ & Thm \ref{thm:shortestpathoracle}\\
Exact Hopsets & $\texttt{EH}(n, p) = \Omega\left( \frac{\betastar(n, p, \infty)}{p} \right)$ & Thm \ref{thm:exacthopsets}\\
\hline
\textbf{Directed Flow-Cut Gap} &  $\fcg(\beta(n, n, \infty)) =  \widetilde{\Omega}\left(\frac{\beta(n, n, \infty)}{n}\right)
$  & Thm \ref{thm:fcg} \\
Sparsest Cut Gap & $\scg(\beta(n, n, \infty)) =  \widetilde{\Omega}\left(\frac{\beta(n, n, \infty)}{n}\right)
$ & Thm \ref{thm:fcg_sc} \\
Directed Steiner Forest Gap & $\dsfg(n,p)=\Omega\left(\frac{\beta(n, p, \infty)}{n^{3/2}}\right)$ & Thm \ref{thm:dsf_gap}\\
\hline
\end{tabular}
\caption{\label{tbl:reductions} Our results on the relationships between $\beta, \beta^*$, and objects in the literature on succinct network design.}
\end{center}
\end{table}

The main conceptual contribution of this paper is a series of reductions from problems in network design to the functions $\beta$ or $\beta^*$.
Table \ref{tbl:reductions} lists our results of this type, separated into three main technical threads.
For every row of this table, one can recover or improve the current state-of-the-art bounds for the object in question by plugging in bounds for $\beta$ or $\beta^*$ (see Table \ref{tbl:combo} for details).
Thus, (1) $\beta, \beta^*$ have arguably been under the surface in prior work on all of these problems, and (2) further improved bounds for $\beta, \beta^*$ could have widespread, black-box consequences for the area.


While new ideas are often needed to prove the bounds in Table \ref{tbl:reductions}, this is overall the less technical part of our paper; our main technical contributions lie in improved bounds for $\beta, \beta^*$, discussed next.

\subsubsection{Technical Overview: Preservers \label{sec:preservers}}

We begin with reachability preservers.
\begin{definition} [Reachability Preservers \cite{AB18}]
Let $G = (V, E)$ be a directed graph and let $P \subseteq V \times V$ be a set of demand pairs.
A reachability preserver is a subgraph $H \subseteq G$ in which, for all $(s, t) \in P$ such that there exists an $s \leadsto t$ path in $G$, there also exists an $s \leadsto t$ path in $H$.

We define $\rpp(n, p)$ as the smallest integer such that every $n$-node graph and set of $|P|=p$ demand pairs has a reachability preserver on $\le \rpp(n, p)$ edges.
\end{definition}

Extremal bounds for reachability preservers have been studied recently \cite{AB18, CCC22, CC20, BCR16}, but they had long been studied algorithmically in the context of the Directed Steiner Forest problem, which asks to compute a reachability preserver of minimum total weight of a given input instance $G, P$.
This problem is NP-hard, but the state-of-the-art approximation algorithms use \emph{extremal} bounds for reachability preservers as an ingredient \cite{CDKL17, AB18, GLQ21}.
We prove:
\begin{theorem}
$\rpp(n, p) = \Theta(\beta(n, p, \infty))$.
\end{theorem}
That is, the extremal bounds for reachability preservers are entirely captured by the value of $\beta$.
This reduction is perhaps our most consequential one: with our improved upper bounds on $\beta(n, p, 4)$ discussed later, this implies a polynomial improvement in the extremal bounds for reachability preservers.
\begin{corollary}
Every $n$-node graph and set of $p$ demand pairs has a reachability preserver on $O(n^{3/4} p^{1/2} + n^{5/8} p^{11/16} + n)$ edges.
\end{corollary}

The previous upper bound was $O(n^{2/3} p^{2/3} + n)$ \cite{AB18}.
The lower bound $\rpp \ge \Omega(\beta(n, p, \infty))$ is straightforward and perhaps implicit in \cite{AB18}, but the upper bound $\rpp(n, p) \le O(\beta(n, p, \infty)$ takes more work.
A natural proof attempt might be to take a hard input instance $G, P$ for reachability preservers requiring $\rpp(n, p)$ edges, carefully choose a path for each demand pair, interpret these choices as a path system, and hope that the resulting path system has bridge girth $\infty$ and therefore size $\le \beta(n, p, \infty)$.
Unfortunately, this attempt fails: for some inputs $G, P$, it is not possible to choose paths that yield a path system of bridge girth $\infty$.
Specifically, this may not work on input instances $G, P$ that have several possible paths for each demand pair, or where these paths are not edge-disjoint (the overlapping parts of paths count as $2$-bridges).

Our solution is perhaps conceptually unusual: we do not attempt to handle these troublesome input instances $G, P$ at all.
Instead, we prove an \emph{independence lemma}, showing that there exist highly structured hard input instances realizing $\rpp(n, p)$.
This structure allows us to map these particular structured worst-case instances to systems of bridge girth $\infty$, which is enough for an extremal reduction between $\rpp(n, p)$ and $\beta(n, p, \infty)$.
We prove an analogous independence lemma for online reachability preservers\footnote{Our model of online reachability preservers is a slight variant of the one introduced recently by Grigorescu, Lin, and Quanrud \cite{GLQ21}.}.
These independence lemmas are also the missing ingredient towards an incompressibility theorem for reachability preservers: we show that no \emph{data structure} (not necessarily a subgraph) can encode paths among demand pairs with better space efficiency than a reachability preservers, which yields our reduction for path oracles.

A \emph{distance preserver} is a subgraph that preserves distance among demand pairs, not just reachability.
Distance preservers were introduced by Coppersmith and Elkin \cite{CE06}, and extremal bounds for distance preservers were studied in \cite{CE06, Bodwin21, BCE05, BV21, CGMW18, CDKL17}.
\begin{definition} [Distance Preservers \cite{CE06}] \label{def:dps}
Let $G=(V, E, w)$ be a directed weighted graph and let $P \subseteq V \times V$ be a set of demand pairs.
A distance preserver is a subgraph $H \subseteq G$ in which, for all $(s, t) \in P$, we have $\dist_H(s, t) = \dist_G(s, t)$.

We define $\dpp(n, p)$ as the least integer such that every $n$-node graph and set of $|P|=p$ demand pairs has a distance preserver on $\le \dpp(n, p)$ edges.
\end{definition}

We prove:
\begin{theorem}
$\Omega(\beta^*(n, p, \infty)) \le \dpp(n, p) \le \beta^*(n, p, 2)$.
\end{theorem}

This time, our main conceptual contribution is the lower bound $\dpp(n, p) \ge \Omega(\beta^*(n, p, \infty))$, based on realizing ordered path systems as unique shortest paths in a graph.
The upper bound $\dpp(n, p) \le \beta^*(n, p, 2)$ is arguably implicit in \cite{CE06}, and is based on a well-known connection to \emph{consistent path systems}: that is, it follows from the simple observation that no two unique shortest paths in a graph may intersect, split apart, and then intersect again later.\footnote{Since $\beta^*(n, p, 2) = \beta(n, p, 2)$, and so we could have just as well written $\beta(n, p, 2)$ for the upper bound on $\dpp$.  We chose $\beta^*(n, p, 2)$ because it suggests an open question: since $\dpp$ is sandwiched between $\beta^*(n, p, \infty)$ and $\beta^*(n, p, 2)$, can it be placed more precisely in the $\beta^*$ hierarchy?}
However, we note that the following section contains a tight lower bound on $\beta^*(n, p, 2) = \beta(n, p, 2)$.
This implies a major technical limitation to further progress on distance preservers: the power of consistency has been pushed to its limit, and so if we are to improve the state-of-the-art upper bounds for distance preservers (see Table \ref{tbl:combo}), we \emph{must} rely on more intricate structural properties of shortest paths.

We also prove an analogous independence lemma for distance preservers, which implies an analogous incompressibility theorem: no data structure can record shortest paths among demand pairs with significantly better space efficiency than a distance preserver.
This yields our reduction for shortest path oracles.

\subsubsection{Technical Overview: Flow-Cut and Integrality Gaps}

Finally, we consider integrality gaps for three problems in network design: Directed Multicut (\DMC), Directed Sparsest Cut (\DSC), and Directed Steiner Forest (\DSF).
%
The standard integrality gaps for \DMC and \DSC are often interpreted as \textit{flow-cut gaps}, as we explain next.

In the \DMC problem, we are given a graph $G = (V, E)$ and a set of demand pairs $P \subseteq V \times V$, and the objective is to find a minimum-size subset of $E$ whose removal separates all pairs of nodes in $P$.
The Maximum Multicommodity Flow (\MMF) problem asks for the maximum total flow that can be simultaneously pushed between the demand pairs (under unit edge capacities).
\MMF is the LP dual of the fractional relaxation $\widehat{\DMC}$ of $\DMC$.
Thus, for any input $G, P$, we have
$$\MMF(G, P) = \widehat{\DMC}(G, P) \le \DMC(G, P).$$
The famous min-cut max-flow theorem states that we have equality when $|P|=1$, but we do \emph{not} have equality in general.
It is interesting to study the maximum possible ratio between these terms, as an approximate version of the min-cut/max-flow theorem, and which often has applications in approximation algorithms.
For undirected graphs, the seminal work of Leighton and Rao showed that the maximum possible ratio is $\Theta(\log p)$ \cite{LR99, GVY96}.
We will be interested in the corresponding quantity for directed graphs, called the \emph{directed flow-cut gap}:
\begin{definition} [$\fcg$]
The function $\fcg(n)$ is the least integer $k$ such that, for every $n$-node directed graph $G$ and set of demand pairs $P$ (of any size), we have
$\DMC(G, P) \le k \cdot \MMF(G, P).$
\end{definition}


This function $\fcg$ has been studied in \cite{AAC07, CKR01, Gupta03, KKN05, LR99, SSZ04}.
On the lower bounds side, an important paper by Chuzhoy and Khanna was the first to show that the bound is polynomial, with a lower bound of $\fcg(n) = \Omegaish(n^{1/7})$ \cite{chuzhoy2009polynomial}.
On the upper bounds side, the current bound is $\fcg(n) = \Oish(n^{11/23})$ \cite{AAC07}.
We show the following reduction, which recovers the lower bound from \cite{chuzhoy2009polynomial} (see Table \ref{tbl:combo}):
\begin{theorem}\label{thm:fcg:weak}
$\fcg(\beta(n, n, \infty)) = \widetilde{\Omega}\left(\frac{\beta(n, n, \infty)}{n}\right)$.
\end{theorem}

At a technical level, this theorem closely follows the construction of Chuzhoy and Khanna \cite{chuzhoy2009polynomial}.
Their construction uses a particular path system construction as an internal ingredient, which may be interpreted as a $\beta(n, n, \infty)$ lower bound system that also has many additional convenient properties (e.g., it is layered and highly symmetric).
Our contribution is a generalization of their analysis, to show that these convenient properties are not really necessary, and one can plug in any system achieving $\beta(n, n, \infty)$ as a black box.
Similarly, we obtain a reduction for the \emph{sparsest cut problem}, again based on \cite{chuzhoy2009polynomial}.

Finally, in the Directed Steiner Forest problem (\DSF), we are given a weighted directed graph $G = (V, E, w)$ and a set of demand pairs $P \subseteq V \times V$, and the goal is to find  a  minimum weight subgraph $H \subseteq G$ that contains a directed $s \leadsto t$ path for all $(s, t) \in P$.
\DSF is NP-hard, but it has a natural integer programming formulation (see Section \ref{sec:dsfgap}), and studying the integrality gap of its LP relaxation is a natural step towards designing approximation algorithms.

This integrality gap and related approximation algorithms have been studied in \cite{li2022polynomial, BermanBMRY13, FKN12}.
Letting $\dsfg(n, p)$ be the integrality gap for instances with $n$ nodes and $p$ demand pairs, the previous best bound was that there exists $p$ for which $\dsfg(n, p) = \Omega(n^{0.0418})$ \cite{li2022polynomial}, implied by work on Directed Steiner Tree.
We show:
\begin{theorem}
\label{thm:dsfg}
    $\dsfg(n, p) = \Omega\left( \frac{\beta(n, p, \infty)}{n^{3/2}} \right)$ for  $p \leq n^{2-o(1)}$.
\end{theorem}

In particular, plugging in $p = n^{2 - o(1)}$, we improve the integrality gap to $\Omega(n^{1/2 - o(1)})$.
This theorem partially addresses an open question in \cite{AMS13}, where Alon, Moitra, and Sudakov asked whether bounds on the Ruzsa-\Szemeredi{} function, which is equivalent to bounds on $\beta(n, p, \infty)$ in the setting of large $p$ (see Theorem \ref{thm:b3rs}), could be useful towards proving integrality gaps for Directed Steiner \emph{Tree}.

\subsection{New Extremal Bounds for Bridge Girth}

The main technical contributions of this paper are some new upper and lower bounds for $\beta$ and $\beta^*$, polynomially improving over bounds implicit in the previous literature.
These improved bounds imply new results for various problems in network design; see Table \ref{tbl:combo} in the next section.
Table \ref{tbl:betabounds} gives a quick reference to our new bounds on $\beta, \beta^*$, as well as bounds implicit in prior work, and Figure \ref{Figs:tikzpicture2} plots the state-of-the-art bounds on $\beta, \beta^*$ following our paper.

\begin{table} [t]
\begin{center}
\begin{tabular}{lll}
\hline
\textbf{$k$} &  \textbf{Bound} & \textbf{Justification}\\
\hline
\multirow{2}{*}{$2$} & \multirow{2}{*}{$\beta=\betastar=\Theta\left(\min\left\{ n^{2/3}p, p^{1/2}n\right\} + n + p\right)$} & Upper Implicit (see App \ref{app:twobounds})\\
& & Lower New (see Thm \ref{thm:twobeta})\\
\hline
\multirow{2}{*}{$3$}
& $\beta = O\left( \min\left\{n^{2/3}p^{2/3}, \frac{n^2}{2^{C\log^* n}}\right\} + n + p \right)$ & \multirow{2}{*}{Implicit (see App \ref{app:threebounds})}\\
& $\beta = \Theta\left( n^{2/3} p^{2/3} \right)$ when $p \in \{n^{4/5}, n^{7/8}, n, n^{8/7}, n^{5/4}\}$ & \\
\hline
$4$ & $\beta = O\left(n^{3/4} p^{1/2} + n^{5/8} p^{11/16} + n + p\right)$ & New (see Thm \ref{thm:bnp4upper})\\
\hline

\multirow{5}{*}{$\infty$} 
& $\beta = O\left( \frac{p^2}{2^{C \log^* p}} + n \right) $ & New-ish (see Thm \ref{thm:rsreverse})\\
& $\beta = \Omega\left(n^{\frac{2}{d+1}}p^{\frac{d-1}{d}}\right)$, $d$ any positive integer & Implicit (see Cor \ref{cor:binftylb})\\
& $\betastar = O\left(\min\left\{n^{1/2} p, p^{1/2} n\right\} + n + p\right)$ & New-ish (see Thm \ref{thm:bstarupper})\\
& $\betastar = \Omega\left(n^{2/3}p^{2/3}\right)$ & New-ish (see Thm \ref{thm:bstarlower})\\
& $\betastar = \Thetaish(n^{4/3})$ when $p=n$, conditional on Hyp \ref{hyp:ordered_gap} & New-ish (see Thm \ref{thm:hypbstarlb})\\
\hline
\end{tabular}

\caption{\label{tbl:betabounds} Asymptotic Bounds for $\beta$ and $\beta^*$.  For brevity, we write $\beta$ in place of $\beta(n, p, k)$, and similar for $\beta^*$.  Results are marked as fundamentally new in this paper, implicit in prior work, or ``new-ish'' meaning that they reuse a major ingredient from prior work but also have a new idea.}
\end{center}
\label{table:asymp_fn}
\end{table}

\begin{figure}[htbp]
  \begin{center}
    \includegraphics[width=4.5in]{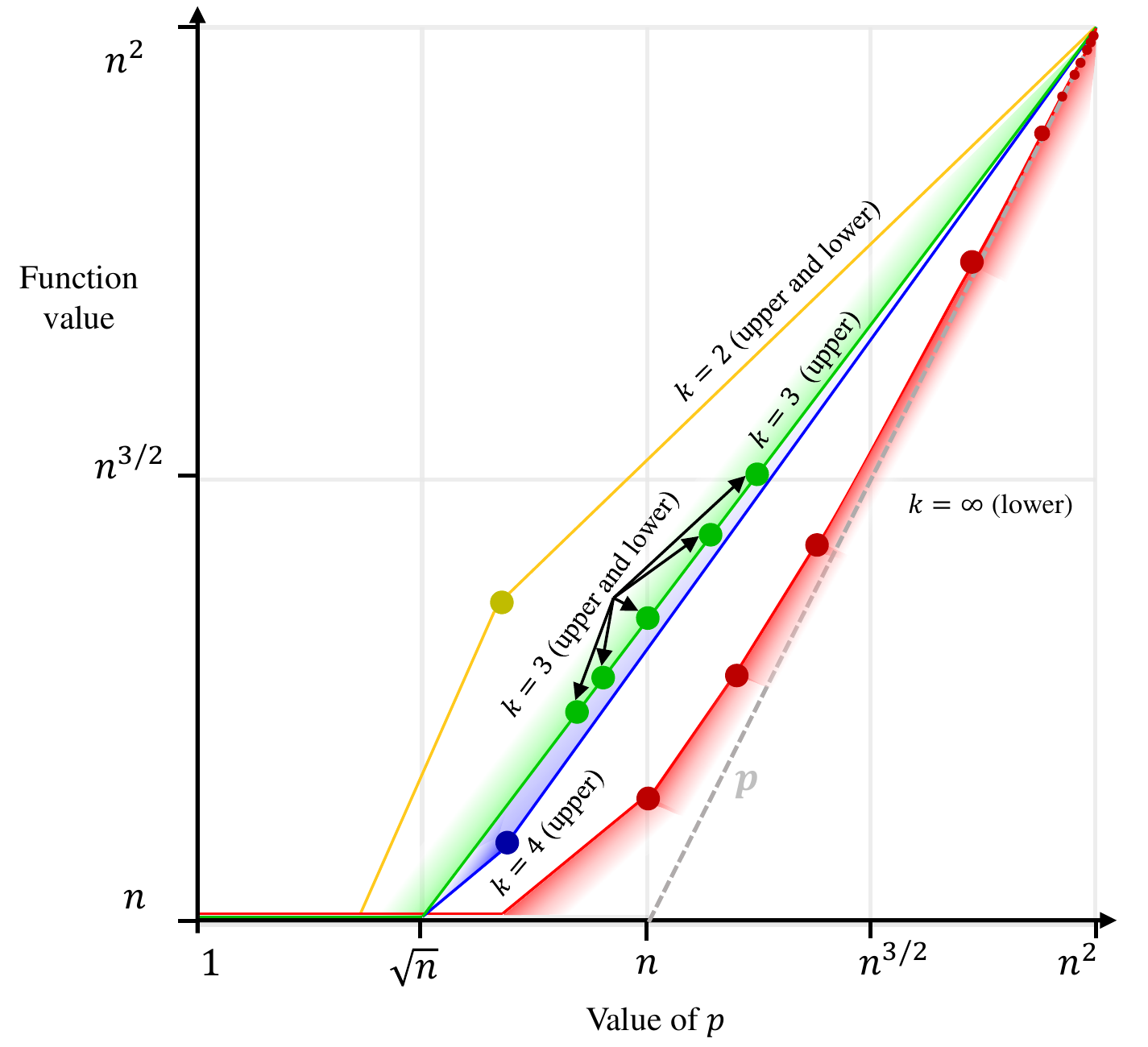}
    \end{center}
    \caption{Asymptotic bounds for $\beta(n, p, k)$, drawn to a logarithmic scale.}
    \label{Figs:tikzpicture2}
\end{figure}

\subsubsection{New Lower Bounds for $\beta(n, p, 2)$}

Our first main new result is a tight lower bound for $\beta(n, p, 2)$:
\begin{theorem} \label{thm:introbnp2}
$\beta(n, p, 2) = \Theta\left(\min\left\{ n^{2/3}p, p^{1/2}n\right\} + n + p\right)$.
\end{theorem}

The upper bound in this theorem is implicit in \cite{Bodwin21, CE06}, and the lower bound is new in this paper.
To explain our contribution, we discuss the previous (implicit) lower bound.
The simplest way to ensure that a path system avoids $2$-bridges is to simply ensure that any two paths intersect on at most one node.
The largest path systems with this property have long been known: they are \emph{finite projective planes}\footnote{Finite projective planes are set systems rather than path systems, but one can place an arbitrary ordering on the sets while retaining the property that any pair of resulting paths intersect on at most one node.}, which imply the lower bound
$$\beta(n, p, 2) = \Omega\left(\min\left\{ n^{1/2}p, p^{1/2}n\right\} + n + p\right).$$
If one wants to design a path system that polynomially exceeds this bound, it is necessary for a typical pair of paths $\pi_1, \pi_2$ to intersect on polynomially many nodes.
To avoid $2$-bridges, it would then be required that $\pi_1, \pi_2$ contain the nodes in $\pi_1 \cap \pi_2$ in exactly opposite orders.
It seems rather unlikely to obtain this opposite-order property for all pairs of paths simultaneously.
Cementing this intuition, the work of Coppersmith and Elkin \cite{CE06} implies that the finite projective plane lower bound on $\beta(n, p, 2)$ is indeed tight in the parameter regime $p \ge n$.
However, our Theorem \ref{thm:introbnp2} shows on the contrary that the finite projective plane is \emph{polynomially far} from optimal in the remaining parameter regime $p \ll n$, and a denser construction with exactly this reverse-order property can in fact be achieved.
The construction is a (dualized) version of the finite projective plane, based on \emph{quadratics} over finite fields instead of lines.

As discussed in Section \ref{sec:preservers}, the practical consequence of this lower bound is a technical limitation on the tool of \emph{consistent path systems}.
In network design, a common strategy to limit the number of edges in a graph is to show that it arises from a path system with the property that no two paths intersect, split apart, and then intersect again later.
Theorem \ref{thm:introbnp2} settles the worst-case size of such a path system, and thus to obtain better upper bounds than the ones in Theorem \ref{thm:introbnp2}, more careful technical arguments are needed.

\subsubsection{New Upper Bounds for $\beta(n, p, 4)$}

Our next main result is a new upper bound for $\beta(n, p, 4)$.
\begin{theorem} \label{thm:introbnp4upper}
$\beta(n, p, 4) = O\left(n^{3/4} p^{1/2} + n^{5/8} p^{11/16} + n + p\right)$.
\end{theorem}

It should be noted here that $\beta$ is inverse-monotonic in $k$; that is, $\beta(n, p, k_1) \leq \beta(n, p, k_2)$ if $k_1 \leq k_2$.
Hence, Theorem \ref{thm:introbnp4upper} also implies a polynomially improved upper bound for $\beta(n, p, 5), \dots, \beta(n, p, \infty)$.
As discussed in Section \ref{sec:preservers}, one corollary is a new polynomially improved upper bound for reachability preservers.
Another consequence of this result is that it provides a clear avenue for further progress towards understanding $\beta(n, p, \infty)$.
A point of this paper, reflected more precisely in Table \ref{tbl:reductions}, is that new \emph{lower} bounds for $\beta(n, p, \infty)$ would be very consequential in network design.
This in turn motivates the study of \emph{upper} bounds for $\beta(n, p, \infty)$, towards determining the extent to which these improved lower bounds might be possible.
The previous-best upper bounds on $\beta(n, p, \infty)$ were inherited all the way from the implicit upper bounds on $\beta(n, p, 3)$.
Our theorem is proof-of-concept that exploiting larger forbidden bridges is indeed a worthwhile avenue towards improved understanding of the value of $\beta(n, p, \infty)$.

At a technical level, the proof of Theorem \ref{thm:introbnp4upper} is considerably more involved than other upper bounds in the area, including those for $\beta(n, p, 2), \beta(n, p, 3),$ and the extremal functions of high-girth graphs (see Appendix \ref{app:girthtour}).
All previous bounds are based roughly on a \emph{forward-search} strategy, in which one picks a node, counts the paths intersecting that node, counts the nodes contained in those paths, and so on; bridge-freeness is used to argue that the nodes/paths that are witnessed are all distinct.
See Theorems \ref{thm:moore}, \ref{thm:cebp2}, \ref{thm:threebound}, for examples.
Our challenge is that forward-search does not work so well for $4$-bridges, in the sense that a lack of $4$-bridges does not imply that distinct nodes/paths are discovered at the appropriate level of the forward search.
This requires considerable technical work to overcome, and due to space constraints we defer further technical overviewing to Section \ref{sec:bnp4upper}.

\subsection{Future Directions and Open Problems \label{sec:open}}

Table \ref{tbl:combo} lists the quantitative bounds obtained by directly mixing the reductions from Table \ref{tbl:reductions} with the bounds on $\beta, \beta^*$ from Table \ref{tbl:betabounds}.
One category of open problem is to improve the quantitative upper or lower bounds for any of these objects, whether or not via reductions to $\beta, \beta^*$.

It would also be interesting just to \emph{recover} state-of-the-art quantitative bounds for network design problems via bridge girth reductions, so that we gain black-box improvements if and when the bounds for $\beta, \beta^*$ are improved.
Some good candidates for this program might include:
\begin{itemize}
\item \textbf{(Shortcut/Hopset Upper Bounds)} We have proved that one can recover state-of-the-art lower bounds on shortcut sets ($\sss$) by reduction to $\beta(n, p, \infty)$.
On the upper bounds side, a recent breakthrough of Kogan and Parter \cite{KP22a} proved that\footnote{These bounds are stated under a different parametrization than \cite{KP22a}: we use $p$ as the size of the hopset, and $\sss(n, p)$ as its hopbound, and thus $\sss$ is decreasing in $p$.}
$$\sss(n, p) = \begin{cases}
\Oish\left( n^{2/3} p^{-1/3} \right) & \text{when } p \ge n\\
\Oish\left( np^{-2/3} \right) & \text{when } p \le n
\end{cases}.$$
Subsequent work by Berenstein and Wein \cite{BW23} obtained a similar bound for $(1+\eps)$ hopsets.
Obtaining this bound with a bridge girth reduction would be interesting.

\item \textbf{(Exact Hopset Upper Bounds)} We have proved that one can recover state-of-the-art lower bounds on exact hopsets ($\ehh$) by reduction to $\beta^*(n, p, \infty)$.
On the upper bounds side, there is a simple folklore algorithm, sometimes attributed to Ullman and Yannakakis \cite{UY91}, that shows
$\ehh(n, p) = \Oish\left( np^{-1/2} \right).$
We refer to \cite{KP22, KP22a} for discussion of this algorithm.
We find the possibility of recovering this upper bound with a bridge girth reduction intriguing.

\item \textbf{(Flow-Cut Gap Upper Bounds)} We have proved that one can recover state-of-the-art lower bounds on the flow-cut gaps ($\fcg, \scg$) by reduction to $\beta(n, p, \infty)$.
On the upper bounds side, the state-of-the-art is
$\fcg(n), \scg(n) = \Oish(n^{11/23})$
obtained by Agarwal, Alon, and Charikar \cite{AAC07}.
Can we recover this bound with a bridge girth reduction?
\end{itemize}

A recent paper by Kogan and Parter \cite{KP22} perhaps makes some progress on this program by proving reductions among several important objects in network design (although naturally it does not directly consider bridge girth).

Another natural kind of open problem left by this paper is to obtain quantitative improved upper/lower bounds for $\beta, \beta^*$.
We would also consider \emph{self-reductions} very interesting, studying how the values of $\beta, \beta^*$ evolve as $k$ increases.
The following is a concrete open problem in this vein.
Notice that the extremal functions of high girth graphs cease to benefit from bridge girth parameters above $\log n$ (for example, $\gamma(n, \log n) = \Theta(\gamma(n, \infty)) = \Theta(n)$).
We think it is likely that a similar effect holds for $\beta, \beta^*$:

\begin{conjecture}\label{conj:beta_log}
For all $n, p$, we have $\beta(n, p, \log n) = \Theta(\beta(n, p, \infty))$, and $\beta^*(n, p, \log n) = \Theta(\beta^*(n, p, \infty))$.\footnote{We are grateful to an anonymous reviewer for suggesting this open problem.}
\end{conjecture}

In Theorem~\ref{eq_conj} we provide some additional evidence that Conjecture~\ref{conj:beta_log} is true, by showing that its first half is implied by a plausible equality between the sizes of approximate distance preservers and reachability preservers, analogous to results already known for undirected preservers \cite{KP22} and directed hopsets \cite{BW23}.

Finally, we discuss applications of $\beta, \beta^*$ with other values of $k$.
This paper directly motivates bridge girth parameters $k=2$ and $k=\infty$, which are the settings that arise most commonly in our reductions (see Table \ref{tbl:reductions}).
We also consider the parameter $k=3$ to be comparably important, because it generalizes the Ruzsa-\Szemeredi{} problem, which in turn captures prior work in network design \cite{Bodwin21, KUV19} (see Appendix \ref{app:girthtour} for more details).
What about finite $k \ge 4$?
Currently, we primarily use these setting as a tool to understand $\beta(n, p, \infty)$.
It is an interesting conceptual open problem to find \emph{direct} applications of $\beta(n, p, k), \beta^*(n, p, k)$, with intermediate choices of $k$, to problems in network design.

A candidate area in which these applications could arise is in the theory of ordered graphs and matrices.
We have already applied this theory a bit, in the connections between Hypothesis \ref{hyp:ordered_gap} and Theorem \ref{thm:hypbstarlb}.
More broadly, there is a line of work in extremal combinatorics on \emph{ordered matrix patterns}, as pioneered by Pach and Tardos \cite{PT06}.
In this problem, we consider binary $n \times n$ matrices, and we receive a collection of one or more forbidden submatrices.
The goal is to determine the maximum possible number of $1$'s that could appear in such a matrix.
One can naturally interpret an $n \times n$ binary matrix as the incidence matrix of an ordered, acyclic path system with $n$ nodes and $n$ paths, and high bridge girth in such a system corresponds to a collection of forbidden patterns.
These problems have applications in data structures \cite{Pettie10}.
We refer to survey \cite{Tardos18} for more on work in this space.





\clearpage
\begin{figure}[htbp]
  \begin{center}
    \includegraphics[width=6.5in]{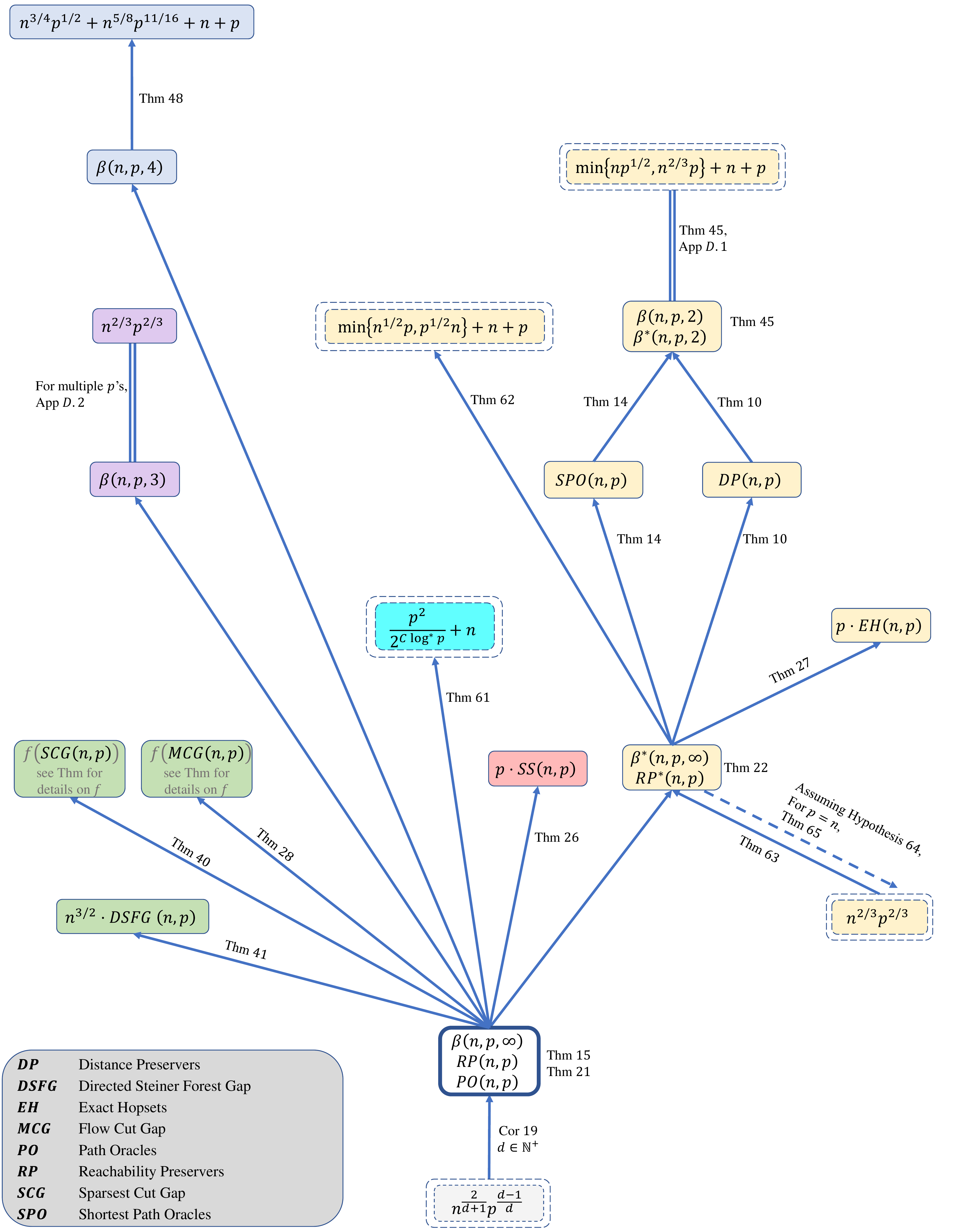}
    \end{center}
    \caption{A diagram representing the relationships between the different notions in this paper. Here, a directed arrow from $A$ to $B$ means that $A\leq B$, and an undirected double line means they are equal.}
    \label{Figs:Bridge}
\end{figure}
\clearpage


%% file: pathprelims.tex
\section{Preliminaries on Path Systems}

Here, we quickly review some standard definitions, notations, and technical lemmas for path systems that will be useful.
\begin{itemize}
\item For a path system $S = (V, \Pi)$, its \emph{incidence graph} is the bipartite graph where the nodes on one side of the bipartition correspond to $V$, the nodes on the other side of the bipartition correspond to $\Pi$, and there is an edge between $v \in V$ and $\pi \in \Pi$ iff $v \in \pi$.

\item A path system $S = (V, \Pi)$ is said to be \emph{acyclic} if it does not have any directed cycles; equivalently, there is a total order of $V$ (called a ``topological order'') such that the order of every path $\pi \in \Pi$ is simply the order of $V$ restricted to the nodes in $\pi$.

\item For a path $\pi$, we write $x <_{\pi} y$ to mean that $x, y$ are both nodes in $\pi$, and $x$ strictly precedes $y$ in $\pi$.
We use the notation $x \le_{\pi} y$ similarly.

\item For a path $\pi$, a \emph{subpath} is a (not necessarily contiguous) subsequence $\pi' \subseteq \pi$.

\item A path system $S' = (V', \Pi')$ is a \emph{subsystem} of $S = (V, \Pi)$, written $S' \subseteq S$, if one can obtain $S'$ from $S$ by a sequence of zero or more of the following operations: delete a node from $V$, delete a path from $\Pi$, or delete a single instance of a node from a single path in $\Pi$.
We say that $S'$ is the \emph{induced subsystem} on $V'$ if it is the system obtained by deleting all nodes in $V \setminus V'$.

\item For a path system $S = (V, \Pi)$, the \emph{degree} of a node $v \in V$, written $\deg(v)$, is the number of paths in $\Pi$ that contain $v$.

\item The \emph{length} of a path $\pi \in \Pi$, written $|\pi|$, is the number of nodes in $\pi$ (note that this length is bigger by $1$ than the length of $\pi$ when viewed as a path in an unweighted graph).

\item For a system $S$ with $n$ nodes and $p$ paths, the \emph{average degree} is the quantity
$$d = \sum \limits_{v \in V} \deg(v) / n$$
and the \emph{average length} is the quantity
$$\ell = \sum \limits_{\pi \in \Pi} |\pi| / p.$$
The \emph{size identity} is that
$$nd = \|S\| = p \ell.$$
\end{itemize}

The following ``cleaning lemma'' lets us assume some convenient regularity properties for the path systems realizing $\beta(n, p, k)$ and $\beta^*(n, p, k)$.

\begin{lemma} [Cleaning Lemma] \label{lem:cleaning}
For every triplet $n, p, k$ there exists a path system $S$ with $\le n$ nodes, $\le p$ paths, bridge girth $>k$, $\|S\| = \Omega(\beta(n, p, k))$, and the following two additional properties:
\begin{itemize}
\item (Approximately Degree-Regular) All nodes have degree $\Theta(d)$, where $d$ is the average degree in $S$, and

\item (Approximately Length-Regular) All paths have length $\Theta(\ell)$, where $\ell$ is the average length in $S$.
\end{itemize}
An identical claim holds for ordered path systems and ordered bridge girth.
\end{lemma}
Slight variants of this lemma are standard in the area, so we defer the proof to Appendix \ref{app:cleaning}.
In the rest of this paper, we will often use the cleaning lemma as a tool to make assumptions about path systems realizing $\beta$ or $\beta^*$.
In other words, as our proofs typically start along the lines of ``Let $S$ be a path system with $n$ nodes, $p$ paths, bridge girth $>k$, and $\|S\| = \Omega(\beta(n, p, k))$", we may often then use the cleaning lemma without loss of generality to guarantee that, additionally, $S$ is both approximately degree-regular and approximately length-regular.

%% file: bounds.tex
\section{Bounds on $\beta, \beta^*$ Functions}

In this section, we prove new bounds for $\beta, \beta^*$ in the settings $k = 2$, $k=4$, and $k=\infty$.

\subsection{Lower Bounds for $k=2$}

We prove:
\begin{theorem} \label{thm:twobeta}
$\beta(n, p, 2) = \betastar(n, p, 2) = \Theta\left(\min\left\{np^{1/2}, n^{2/3}p \right\} + n + p\right)$
\end{theorem}

Since $2$-bridges are not sensitive to ordering, we immediately have $\beta(n, p, 2) = \betastar(n, p, 2)$.
The upper bounds for Theorem \ref{thm:twobeta} are implicit in \cite{Bodwin21, CE06}; for completeness, we supply proofs in Appendix \ref{app:twobounds}.
The lower bound is new, and will be the focus of the rest of this section.
We recall that
$$\beta(n, p, 2) \ge \Omega(n+p)$$
is an immediate lower bound, by considering either $1$ path through all $n$ nodes (giving a lower bound of $\Omega(n)$), or by considering $p$ paths of $1$ node each (giving a lower bound of $\Omega(p)$).
It thus remains to prove
$$\beta(n, p, 2) \ge \Omega\left(\min\{np^{1/2}, n^{2/3}p\}\right).$$
These two minimized bounds meet at $p=n^{2/3}$.
We will begin by considering this special case: that is, our goal is to construct a 2-bridge-free path system $S = (V, \Pi)$ with $p = \Theta(n^{2/3})$ paths and $\|S\| = \Omega(n^{4/3})$.
We will then generalize to the full bound at the end.

\subsubsection{Construction of $S$}

\paragraph{The nodes.} Let $q$ be an arbitrary prime and let $F_q$ be the finite field on $q$ elements.
Let $Q$ be the set of polynomials over $F_q$ of degree $\le 2$.
The polynomials in $Q$ will ultimately correspond to the \emph{nodes} of the path system: $V = Q$, so $n = |V| = q^3$.

\paragraph{The paths.} For each $(x, y) \in F_q^2$, let $Q_{(x, y)} \subseteq Q$ be the set of polynomials that intersect the point $(x, y)$; that is,
$$Q_{(x, y)} := \left\{ f \in Q \ \mid \ f(x) = y\right\}.$$
There are $q^2$ points $(x, y)$, and so there are $q^2$ such sets in total.
Our plan for defining our paths is to put a circular ordering on the elements of each $Q_{(x, y)}$, and then split the ordering into three parts, giving three paths for each $(x, y)$.
Hence there will be $3q^2$ paths in total.

To define an ordering on $Q_{(x, y)}$: for a polynomial $f(x) = ax^2 + bx + c \in Q$, we define its \emph{derivative} as $f'(x) := 2ax + b$, which we interpret as an element of $F_q$.
Circularly order the polynomials in $Q_{(x, y)}$ by derivative $f'(x)$.
We note that some polynomials in $Q_{(x, y)}$ will have tied derivatives; these ties may be broken arbitrarily.
We then equitably partition the circular ordering into three contiguous parts, and add all three parts as paths in $\Pi$.

\paragraph{Size analysis.} We have $q^3 =: n$ nodes (polynomials).
We have $q^2$ points $(x, y) \in F_q^2$; each point is associated to three paths with $q^2$ nodes between them.
Thus we have
$$\|S\| = \sum \limits_{v \in V} \deg(v) = q^4.$$

\subsubsection{Proof of 2-bridge-freeness}

We will need the following structural lemma:
\begin{lemma} [See Figure \ref{fig:quadcirc} for intuition] \label{lem:quadcirc}
Fix some $x_1 \ne x_2, y_1, y_2$, and let
$$Z := Q_{(x_1, y_1)} \cap Q_{(x_2, y_2)}.$$
Then we have:
\begin{itemize}
\item For $z \in Z$, the values $z'(x_1)$ are pairwise distinct,
\item For $z \in Z$, the values $z'(x_2)$ are pairwise distinct, and
\item The circular ordering of $Q_{(x_1, y_1)}$ restricted to $Z$ is exactly the reverse of the circular ordering of $Q_{(x_2, y_2)}$ restricted to $Z$.
\end{itemize}
\end{lemma}
\begin{proof}
For any $z(x) = ax^2 + bx + c \in Z$, subtracting the equations $z(x_1) = y_1$ and $z(x_2) = y_2$, we get
$$a(x_2^2 - x_1^2) + b(x_2 - x_1) = y_2 - y_1$$
and so, solving for $b$, we have
$$b = \frac{y_2 - y_1 - a(x_2^2 - x_1^2)}{x_2 - x_1}.$$
Rearranging $z(x_1) = y_1$, we also have
$$c = y_1 - ax_1^2 - bx_1.$$
Thus the polynomials $z \in Z$ can be written in the following form, parameterized by $a \in F_q$ (only):
$$z(x) = ax^2 + \left(\frac{y_2 - y_1 - a(x_2^2 - x_1^2)}{x_2 - x_1}\right) x + \left(y_1 - ax_1^2 - \left(\frac{y_2 - y_1 - a(x_2^2 - x_1^2)}{x_2 - x_1}\right)x_1\right)$$
We then have
$$z'(x) = 2ax + \frac{y_2 - y_1 - a(x_2^2 - x_1^2)}{x_2 - x_1} = 2ax + \frac{y_2 - y_1}{x_2 - x_1} - a(x_2 + x_1)$$
and so we compute
$$z'(x_1) = \frac{y_2 - y_1}{x_2 - x_1} + a (x_1 - x_2)$$
and
$$z'(x_2) = \frac{y_2 - y_1}{x_2 - x_1} + a (x_2 - x_1).$$
Both of these functions are affine in $a$, and hence they take different values for each possible choice of $a$, proving the first two points.
The third point follows from the observation that, letting $z_a, z_{-a} \in Z$ be quadratics with parameters $a, -a \in F_q$ respectively, the previous two equations imply that
\begin{align*}
z'_a(x_1) = z'_{-a}(x_2).
\end{align*}

In particular: let $a_i$ be the choice of parameter $a$ such that $z'_{a_i}(x_1) = i$ (note that, since $z'(x_1)$ is affine in $a$, such a choice $a_i$ must exist).
Then the circular ordering of quadratics in $Q_{(x_1, y_1)}$ is
$$\left(z_{a_0}, z_{a_1}, z_{a_2}, \dots, z_{a_{q-2}}, z_{a_{q-1}}, z_{a_q}=z_{a_0}  \right).$$
Meanwhile, using that $z'_{a_i}(x_1) = z'_{-a_i}(x_2)$, the circular ordering of quadratics in $Q_{(x_2, y_2)}$ is
$$\left(z_{a_0}, z_{a_{-1}}, z_{a_{-2}}, \dots, z_{a_{-(q-2)}}, z_{a_{-(q-1)}}, z_{a_{-q}} = z_{a_0}  \right).$$
Since the parameter $a$ is taken mod $q$, these are reverse circular orderings, completing the proof.
\end{proof}

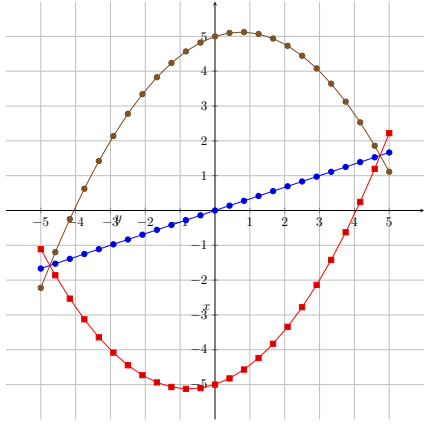
\begin{figure}
\begin{center}
\begin{tikzpicture}[scale=0.5]
\begin{axis}[
axis y line=center,
axis x line=middle,
axis equal,
grid=both,
xmax=6,xmin=-6,
ymin=-6,ymax=6,
xlabel=$x$,ylabel=$y$,
xtick={-5,...,5},
ytick={-5,...,5},
width=5in,
height=5in,
anchor=center,
]

\addplot {x/3} ;
\addplot {x^2/4.5 + x/3 - 5};
\addplot {5 + x/3 - x^2/4.5};
\end{axis}
\end{tikzpicture}
\end{center}
\caption{\label{fig:quadcirc} In $\rr^2$, if one plots quadratics that pass through two fixed points, the orderings by derivative at these two points are exactly opposite.
Lemma \ref{lem:quadcirc} proves that an analogous fact holds in $F_q^2$.}
\end{figure}

\begin{lemma} 
$S$ is $2$-bridge-free.
\end{lemma}
\begin{proof}
Let $\pi_1, \pi_2 \in \Pi$, and recall that $\pi_1, \pi_2$ are respectively constructed with respect to two points $(x_1, y_1), (x_2, y_2) \in F_q^2$.
We consider two cases:
\begin{itemize}
\item If $x_1 = x_2$ then by construction $\pi_1, \pi_2$ are node-disjoint, either because $y_1 \ne y_2$ and so $Q_{(x_1=x_2, y_1)}$ and $Q_{(x_1=x_2, y_2)}$ are disjoint, or because $y_1 = y_2$ and $\pi_1, \pi_2$ represent different parts in the node-disjoint partition of $Q_{(x_1=x_2, y_1=y_2)}$.
Hence $\pi_1, \pi_2$ do not form a $2$-bridge.

\item If $x_1 \ne x_2$, then the points common to $\pi_1, \pi_2$ correspond to a subset of the polynomials in
$$Z = Q_{(x_1, y_1)} \cap Q_{(x_2, y_2)}.$$
By Lemma \ref{lem:quadcirc}, the polynomials in $Z$ have distinct derivatives and opposite circular orderings in $Q_{(x_1, y_1)}, Q_{(x_2, y_2)}$.
It follows that, when we partition the points of $Q_{(x_1, y_1)}, Q_{(x_2, y_2)}$ into thirds to form $\pi_1, \pi_2$, they have opposite orderings of any points $z_1, z_2 \in \pi_1 \cap \pi_2$.
(Here it is important that we take $\ll 1/2$ of the circular ordering to form each path, to avoid the possibility of two paths wrapping around either side of the circular ordering to intersect both at the beginning and the end.)
Hence $\pi_1, \pi_2$ do not form a $2$-bridge. \qedhere
\end{itemize}
\end{proof}

\subsubsection{Remaining Lower Bound}

We have now completed the lower bound proof in the special case $p = \Theta(n^{2/3})$, and it remains to discuss the extension to general $p$.
To obtain the remaining points on our lower bound curve, we can post-process our construction in one of two ways:
\begin{itemize}
\item Suppose we delete nodes from the construction arbitrarily, until only $n' \ll n$ nodes remain.
In our original construction, we had $n$ nodes of degree $\Theta(n^{1/3}) = \Theta(p^{1/2})$ each.
Thus, after deletions, we have $\|S\| = \Theta(n' \cdot p^{1/2})$, which provides one part of our lower bound curve.

\item Alternately, suppose we delete paths from the construction arbitrarily, until only $p' \ll p$ paths remain.
In our original construction, we had $p = \Theta(n^{2/3})$ paths of length $\Theta(n^{2/3})$ each.
Thus, after deletions, we have $\|S\| = \Theta\left(p' n^{2/3}\right)$, which provides the other part of our lower bound curve.
\end{itemize}

\subsection{Upper Bounds for $k=4$ \label{sec:bnp4upper}}

We will prove:
\begin{theorem} \label{thm:bnp4upper}
$\beta(n, p, 4) = O\left( n^{3/4} p^{1/2} + n^{5/8} p^{11/16}  + n + p \right)$.
\end{theorem}

\subsubsection{Technical Lemma: A Bound on the Sum Square of Path Lengths}

A major technical lemma for our proof will be an upper bound on the quantity
$$\|T\|_2^2 := \sum \limits_{\pi \in \Pi} |\pi|^2,$$
i.e., the squared $L^2$ norm of path lengths in a path system (note that the size notion $\|T\|$ may be viewed as the $L^1$ norm of path lengths).
We name this path system $T$ rather than $S$ here because our plan is \emph{not} to apply this lemma to the entire path system $S$ that we analyze in Theorem \ref{thm:bnp4upper}, but rather to a specific subsystem $T \subseteq S$ that we will construct later.
We prove the following bound:
\begin{lemma} \label{lem:2normbound}
Let $T = (V, \Pi)$ be a path system with $n$ nodes, $p$ paths, bridge girth $> 3$, maximum path length $L$, and average path length at least a sufficiently large constant.
Then we have
$$\|T\|_2^2 = O\left( nL + p^{1/3} n^{4/3} \right).$$
\end{lemma}

We will split our proof into a few claims.
Let us say that a path $\pi \in \Pi$ is:
\begin{itemize}
\item \emph{long} if $|\pi| > Cn^{1/2}$, where $C$ is a sufficiently large absolute constant that we leave implicit,
\item \emph{medium} if $\frac{\|T\|_2}{2p^{1/2}} \le |\pi| \le Cn^{1/2}$, or
\item \emph{short} if $|\pi| \le \frac{\|T\|_2}{2p^{1/2}}$.
\end{itemize}

Let us say that the \emph{long paths dominate} if the sum square of long paths is at least as large as the sum square of medium paths and as the sum square of short paths, and the \emph{medium/short paths dominate} if the analogous property holds for the medium/short paths.
The two terms added together in Lemma \ref{lem:2normbound} respectively arise from the cases where the long or medium paths dominate.
The following lemma dispatches with the remaining case:
\begin{lemma} \label{lem:noshortdom}
The short paths do not dominate.
\end{lemma}
\begin{proof}
There are $\le p$ short paths, and by definition each one has length $\le \|T\|_2 / (2p^{1/2})$.
By unioning, their sum square is at most
$$p \cdot \frac{\|T\|_2^2}{4p} = \frac{\|T\|_2^2}{4}.$$
Thus the short paths contribute at most $1/4$ of the total value of $\|T\|_2^2$, so they cannot dominate.
\end{proof}

The following technical lemma will be useful towards bounding $\|T\|_2^2$ in both cases where the long or medium paths dominate:
\begin{lemma} \label{lem:jsbound}
Let $j$ be a parameter that is at least a sufficiently large constant, and let $T_j \subseteq T$ be the subsystem of $T$ that contains exactly the paths $\pi$ of length $j \le |\pi| \le 2j$.
Then:
$$\|T_j\|_2^2 = \begin{cases}
O\left( j n \right) & \text{if } j \ge n^{1/2}\\
O\left( j^{-1} n^2 \right) & \text{if } j \le n^{1/2}.
\end{cases}$$
\end{lemma}
\begin{proof}
Let $p_j$ be the number of paths in $T_j$.
Since $T_j$ has bridge girth $>3$, we may apply the bounds on $\beta(n, p, 3)$ implicit in prior work, which give:
$$\|T_j\| = O\left( n^{2/3} p_j^{2/3} + n + p_j \right)$$
(see Theorem \ref{thm:threebound} in the appendix for a formal proof).
Since $p_j \le n^2$, the term $+p_j$ never dominates, so we may simplify this bound to
$$\|T_j\| = O\left( n^{2/3} p_j^{2/3} + n \right).$$
We then have:
\begin{align*}
j \le \frac{\|T_j\|}{p_j} &= O\left( n^{2/3} p_j^{-1/3} + n p_j^{-1} \right).
\end{align*}
In the case where the first term in the right-hand sum dominates, we continue
\begin{align*}
p_j^{1/3} &\le O\left( n^{2/3} j^{-1} \right)\\
p_j &\le O\left( n^2 j^{-3} \right).
\end{align*}
In the case where the second term in the right-hand sum dominates, we continue
$$p_j = O\left( n j^{-1} \right).$$
Combining these, we get
$$p_j = O\left( n^2 j^{-3} + n j^{-1}\right).$$
Plugging back into our bound on $\|T_j\|$, we get
\begin{align*}
\|T_j\| &= O\left( n^{2/3} \left(n^2 j^{-3} + nj^{-1}\right)^{2/3} + n\right)\\
&= O\left( n^{2/3} \left(n^{4/3} j^{-2} + n^{2/3} j^{-2/3} \right) + n \right)\\
&= O\left( n^{2} j^{-2} + n^{4/3} j^{-2/3} + n \right).
\end{align*}
In the case where $j \ge n^{1/2}$, the latter $+n$ term dominates the sum, and so this gives $\|T_j\| = O(n)$.
Thus $\|T_j\|_2^2$ is the sum of $O(n j^{-1})$ paths, each of which contribute $\Theta(j^2)$ to the sum, so its total is $O(jn)$.
On the other hand, in the case where $j \le n^{1/2}$, this gives $\|T_j\| = O(n^2 j^{-2})$.
Thus $\|T_j\|_2^2$ is the sum of $O(n^2 j^{-3})$ paths, each of which contributes $\Theta(j^2)$ to the sum, so its total is $O(j^{-1} n^2)$.
\end{proof}

Our next lemma counts the contribution of the long paths:
\begin{lemma} [Long Path Gap Bound] \label{lem:longgapbound}
$\sum \limits_{\pi \in \Pi \ \mid \ \pi \text{ long}} |\pi|^2 = O\left( Ln\right).$
\end{lemma}
\begin{proof}
By Lemma \ref{lem:jsbound}, for any parameter $j \ge n^{1/2}$, we have
$$\sum \limits_{\pi \in \Pi \ \mid \ \pi \text{ long and } j \le |\pi| \le 2j} |\pi|^2 = O\left( jn \right).$$
We may therefore control the sum square of long paths by partitioning into subsets of paths of length $j \le |\pi| \le 2j$, and summing the contribution of these subsets.
This gives:
$$\sum \limits_{\pi \in \Pi \ \mid \ \pi \text{ long}} |\pi|^2 = O\left( (n^{1/2})n \right) + O\left((2n^{1/2})n\right) + O\left((4n^{1/2})n\right)+ \dots.$$
This is a geometric sum, which is thus dominated by its largest term.
Recall that we have assumed that all paths in $T$ have length $\le L$, and so the last term has the form $O(Ln)$, proving the lemma.
\end{proof}

Next, we count the contribution of the medium paths.
\begin{lemma} [Medium Path Gap Bound] \label{lem:mediumgapbound}
$\sum \limits_{\pi \in \Pi \ \mid \ \pi \text{ medium}} |\pi|^2 = O\left( \frac{p}{\|T\|_2} n^2 \right)$.
\end{lemma}
\begin{proof}
Let
$$\frac{\|T\|_2}{2p^{1/2}} \le j \le n^{1/2} $$
be a parameter.
From Lemma \ref{lem:jsbound} and the fact that medium paths have length $\le Cn^{1/2}$, we have
$$\sum \limits_{\pi \in \Pi \ \mid \ \pi \text{ medium and } j \le |\pi| \le 2j} = O\left( j^{-1} n^2 \right).$$
As in the long path case, we can bound the sum square of medium path lengths by partitioning the medium paths into parts where all paths in a part have $j \le |\pi| \le 2j$.
This gives
$$\sum \limits_j O\left( j^{-1} n^2 \right)$$
where $j$ ranges from $\|T\|_2 / (2p^{1/2})$ to $n^{1/2}$ by multiples of $2$.
This is again a geometric sum which is dominated by its largest term.
The largest term occurs when $j$ is smallest, i.e., $j=\|T\|_2 / (2p^{1/2})$, and we get
$$O\left( \frac{p^{1/2}}{\|T\|_2} n^2 \right),$$
completing the proof.
\end{proof}

Now we put the parts together:
\begin{proof} [Proof of Lemma \ref{lem:2normbound}]
We consider two cases:
\begin{itemize}
\item If the long paths dominate, then by Lemma \ref{lem:longgapbound} we have $\|T\|_2^2 = O\left( Ln \right)$.

\item If the medium paths dominate, then by Lemma \ref{lem:mediumgapbound} we have
$$\|T\|_2^2 = O\left( \frac{p^{1/2}}{\|T\|_2} n^2 \right).$$
Rearranging, we get
\begin{align*}
\|T\|_2^3 &= O\left( p^{1/2} n^2 \right)\\
\|T\|_2^2 &= O\left( p^{1/3} n^{4/3} \right).
\end{align*}
\end{itemize}
Finally, we recall by Lemma \ref{lem:noshortdom} that the short paths do not dominate, and so this completes the proof.
\end{proof}

\subsubsection{Setup and a Stronger Cleaning Lemma}

In addition to the usual cleaning lemma, for technical reasons it will be helpful to assume an additional property for the path system $S$ that we analyze.
Let us say that a \emph{$2$-cycle} in a path system is a pair of nodes $u, v$, and a pair of paths $\pi_1, \pi_2$ with $u <_{\pi_1} v$ and $v <_{\pi_2} u$.
We use the following lemma to remove $2$-cycles from $S$:
\begin{lemma} \label{lem:2cyccleaning}
For any $n, p$ and any $k \ge 3$, there exists a path system $S$ with $n$ nodes, $p$ paths, bridge girth $>k$, size $\|S\| = \Theta(\beta(n, p, k))$, and no $2$-cycles.
\end{lemma}
\begin{proof}
Start with a path system $S = (V, \Pi)$ with $n$ nodes, $p$ paths, $\|S\| = \beta(n, p, k)$, and bridge girth $>k$ (which may have $2$-cycles).
Construct a path system $S'$ as follows.

Initially $S' = (V, \emptyset)$ is empty.
For each $\pi \in \Pi$ in an arbitrary order, add a subpath $\pi' \in \pi$ to $\Pi'$ generated as follows.
For each node $v \in \pi$, omit $v$ from $\pi'$ if there exists a node $u$ and a previously-added path $q \in \Pi'$ for which $v <_{\pi} u$ and $u <_q v$.
Otherwise, include $v \in \pi'$.
In the following picture, if $\pi$ is the wavy path on top and $q$ is the straight path at the bottom, the two hollow nodes would be omitted and the four solid nodes would be included in $\pi'$ (unless another choice of path $q$ causes them to be omitted).
\begin{center}
\begin{tikzpicture}
\draw [thick, <-] (0, 0) -- (6, 0);
\draw [fill=black] (6, 0) circle [radius=0.15];
\draw (4, 0) circle [radius=0.15];
\draw (2, 0) circle [radius=0.15];
\draw [fill=black] (5, 1) circle [radius=0.15];
\draw [fill=black] (3, 1) circle [radius=0.15];
\draw [fill=black] (1, 1) circle [radius=0.15];
\draw [thick, ->] plot [smooth] coordinates {(0, 1) (1, 1) (2, 0) (3, 1) (4, 0) (5, 1) (6, 0) (7, 1)};

\node at (-0.5, 0) {$q$};
\node at (-0.5, 1) {$\pi$};

\end{tikzpicture}
\end{center}

It is immediate from the construction that we do not complete any $2$-cycles in $S'$, and since $S' \subseteq S$ we still have that $S'$ has bridge girth $>k$.
So it only remains to prove that $\|S'\| = \Theta(\|S\|)$.
Consider a fixed node $v \in V$.
Each time we consider a path $\pi$ with $v \in \pi$, we either keep $v \in \pi$ or we omit it.
If we keep $v$, then $\pi$ contributes $+1$ to the degree of $v$ in $S'$.
If we omit $v$, we do so because of a previously-added path $q$ with $v \in q$.
In this case, let us say that $q$ is \emph{marked} by this action.
We claim that each path can only be marked once.
To see this: suppose for contradiction, that there are two different paths $\pi, \pi'$, which both contain $v$ and which both mark $q$.
This implies that $q, \pi, \pi'$ form a $3$-bridge, as in the following picture (with $v$ as the first node and $\pi$ as the river):
\begin{center}
\begin{tikzpicture}
\draw [fill=black] (6, 0) circle [radius=0.15];
\draw [fill=black] (1, 0) circle [radius=0.15];
\draw [fill=black] (3, 0) circle [radius=0.15];

\draw [thick, <-] (0, 0) -- (7, 0);
\draw [thick, ->] plot [smooth] coordinates {(0, 1) (1, 0) (2, 1) (5, 1) (6, 0) (7, 1)};
\draw [thick, ->] plot [smooth] coordinates {(0, 0.5) (1, 0) (2, 0.5) (3, 0) (4, 0.5)};

\node at (-0.5, 0) {$q$};
\node at (4.3, 0.5) {$\pi'$};
\node at (7.2, 1.2) {$\pi$};
\node at (1, -0.4) {$v$};

\node at (9, 0.5) {$3$-bridge};
\end{tikzpicture}
\end{center}
Thus, each time we omit a node $v \in \pi$ from $\pi$ in the construction of $S'$, we may amortize this against a previously-added path $q$ that kept $v$ in $S'$.
It follows that $\deg_{S'}(v) \ge \deg_S(v)/2$.
Since this holds for all nodes $v$, we have $\|S'\| \ge \|S\|/2$, completing the proof.
\end{proof}

Using this lemma, let $S = (V, \Pi)$ be a path system with bridge girth $>4$, no $2$-cycles, $n$ nodes, $p$ paths, and size $\|S\| = \Theta(\beta(n, p, 4))$.
By the Cleaning Lemma (Lemma \ref{lem:cleaning}), we may further let $\ell, d$ be the average path length and node degree in $S$ (respectively), and assume without loss of generality that all paths have length $\Theta(\ell)$ and that all nodes have degree $\Theta(d)$.\footnote{Technically, to assume the properties of the cleaning lemma and $2$-cycle-freeness simultaneously, we need to use the fact that the construction in the cleaning lemma cannot create $2$-cycles.  This is immediate from the proof.}
We assume that $\ell, d$ are both at least sufficiently large constants (if not, then we immediately have $\|S\| = O(n + p)$).
Under all these assumptions, our goal is now to prove that
$$\|S\| = O\left( n^{3/4} p^{1/2} + n^{5/8} p^{11/16} \right).$$

\subsubsection{The Random Subsystem $S'$}

Our next step is construct a particular subsystem of $S$ that will be useful in analysis.
Consider the following process, parametrized by a positive integer $h \le \ell$ that we choose later, that generates a random subsystem $S' \subseteq S$:
\begin{itemize}
\item Choose a path $\pi_b \in \Pi$ uniformly at random, called the \emph{base path}.

\item Let $Q \subseteq \Pi$ be the set of paths that intersect $\pi_b$ at exactly one node.

\item \textbf{(Vertices of $S'$)} Flip a coin to choose either ``forwards'' or ``backwards.''
Let $V'$ be the set of nodes $v \in V$ with the following property:
\begin{itemize}
\item If we choose ``forwards,'' then the property is: there exists a node $u \in \pi_b$ and a path $q \in Q$ with $u \le v$ in $q$ and with $|q[u \leadsto v]| < h$.
(That is, $u$ weakly precedes $v$ along $q$, and these nodes are at most $h-1$ positions apart in $q$.)
\item If we choose ``backwards,'' the property is similar except that we require $v \le u$ with $|q[v \leadsto u]| < h$.  (That is, $u$ weakly \emph{follows} $v$ along $q$, and these nodes are at most $h-1$ positions apart in $q$.)
\end{itemize}

\item $S' = (V', \Pi')$ is the induced subsystem of $S$ on the vertex set $V'$.
Recall: this means that $\Pi'$ contains the subpath $\pi \cap V'$ for each $\pi \in \Pi$.
\end{itemize}

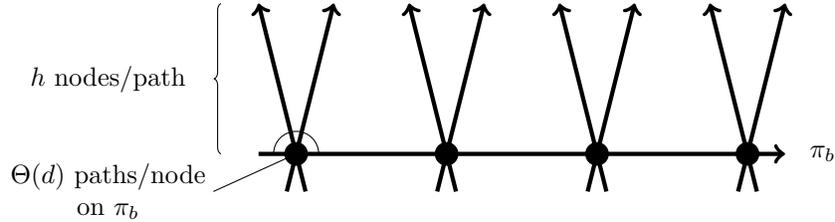
\begin{figure} [ht]
\begin{center}
    \begin{tikzpicture}
    \draw [fill=black] (0, 0) circle [radius=0.15];
    \draw [fill=black] (2, 0) circle [radius=0.15];
    \draw [fill=black] (4, 0) circle [radius=0.15];
    \draw [fill=black] (6, 0) circle [radius=0.15];
    \draw [ultra thick, ->] (-0.5, 0) -- (6.5, 0);
    \draw [ultra thick, ->] (-0.125, -0.5) -- (0.5, 2);
    \draw [ultra thick, ->] (0.125, -0.5) -- (-0.5, 2);
    
    \draw [ultra thick, ->] (1.875, -0.5) -- (2.5, 2);
    \draw [ultra thick, ->] (2.125, -0.5) -- (1.5, 2);
    
    \draw [ultra thick, ->] (3.875, -0.5) -- (4.5, 2);
    \draw [ultra thick, ->] (4.125, -0.5) -- (3.5, 2);
    
    \draw [ultra thick, ->] (5.875, -0.5) -- (6.5, 2);
    \draw [ultra thick, ->] (6.125, -0.5) -- (5.5, 2);
    
    \draw [decorate, decoration = {brace}] (-1, 0) --  (-1, 2);
    \node at (-2.5, 1) {$h$ nodes/path};
    \node at (7, 0) {$\pi_b$};
    
    \draw (0.3, 0) arc (0:180:0.3);
    
    \node [align=center] at (-2.5, -0.5) {$\Theta(d)$ paths/node\\on $\pi_b$};
    \draw (-1.1, -0.5) -- (0, 0);
    
    \end{tikzpicture}
\end{center}
\caption{The random subsystem $S'$, generated with respect to ``height'' parameter $h$, and a ``forwards'' coin flip.}
\end{figure}

Let $n' := |V'|$ be the number of surviving nodes in $S'$, let $p' := |\Pi'|$ be the number of surviving paths in $S'$, and let $\ell' := \|S'\|/p'$ be the average path length in $S'$.
Note that part of each path in $Q$ survives in $S'$; we call the surviving parts $Q' \subseteq \Pi'$.
We have the following controls on their values:

\begin{lemma} [Properties of $S'$] \label{lem:sprimeprops} ~
\begin{itemize}
\item $|Q'| = \Theta(\ell d)$
\item All nodes in $S'$ have degree $\Theta(d)$
\item $n' \le O(\ell d h)$
\item The maximum path length in $S'$ is $O(\ell)$
\end{itemize}
\end{lemma}
\begin{proof}
The fact that all nodes in $S'$ have degree $\Theta(d)$ is inherited directly from $S$, since the degrees of surviving nodes do not change in an induced subsystem.
Similarly, the fact that the maximum path length in $S'$ is $O(\ell)$ is inherited from $S$.

Since $|\pi_b| = \Theta(\ell)$, and each node has degree $\Theta(d)$, and $d$ is a sufficiently large constant, there exist $\Theta(\ell d)$ paths that intersect $\pi_b$.
We notice that these paths are pairwise-distinct: if a path $q$ hits $\pi_b$ at two different nodes, then (since $S$ has no $2$-bridges) $q, \pi_b$ must form a $2$-cycle, but from Lemma \ref{lem:2cyccleaning} there are no $2$-cycles in $S$.
Thus we have $|Q'| = \Theta(\ell d)$.
(This is one of two places where the assumption of no $2$-cycles will be useful.)

By construction every node in $V'$ is contained in a path in $Q'$, and every path in $Q'$ has length $O(h)$, so we also have $n' \le O(h \ell d)$.
\end{proof}

Although the previous lemma acknowledges that the node degrees in $S'$ do not change from $S$, we unfortunately have no such guarantee for the path lengths.
That is, the average path length $\ell'$ in $S'$ might be very different from the average path length $\ell$ in $S$, and moreover the path lengths in $S'$ might fluctuate wildly instead of all being $\Theta(\ell')$.
We unfortunately won't be able to enforce approximate length regularity in $S'$ by re-applying the cleaning lemma, either.
The problem is that our strategy in analysis will be to bound $\|S'\|_2^2$, and while the cleaning lemma gives approximate length regularity while preserving $\|S'\|$, it can significantly change the value of $\|S'\|_2^2$.

It will, however, be helpful in the following analysis to assume that $\ell'$ is at least a sufficiently large constant.
We enable this assumption using the following lemma, which provides a good bound on the size of $S$ in the case where $\ell'$ is only a constant:
\begin{lemma} \label{lem:degenSprime}
If $\mathbb{E}[\ell'] = O(1)$, then $\|S\| = O\left(n^{2/3} p^{2/3} h^{-1/3} \right)$.
\end{lemma}
\begin{proof}
First, we claim that $\mathbb{E}[n'] = \Theta(\ell d h)$.
This follows by noticing two facts.
First, for any path $q \in Q$, the expected length of the corresponding path in $Q'$ (over the forwards/backwards coin flip) is $\Theta(h)$.
Second, the paths in $Q'$ are pairwise node-disjoint, except possibly on nodes in $\pi_b$.
This holds because, if we have two paths $q_1, q_2 \in Q'$ that intersect on a node $v \notin \pi_b$, then $q_1, q_2, \pi_b$ form a $3$-bridge:
\begin{center}
    \begin{tikzpicture}
    \draw [fill=black] (2, 0) circle [radius=0.15];
    \draw [fill=black] (4, 0) circle [radius=0.15];
    \draw [fill=black] (3, 1) circle [radius=0.15];
    
    \draw [ultra thick, ->] (-0.5, 0) -- (6.5, 0);
    \draw [ultra thick, ->] (1.5, -0.5) -- (3.5, 1.5);
    \draw [ultra thick, ->] (4.5, -0.5) -- (2.5, 1.5);
    
    \node at (7, 0) {$\pi_b$};
    \node at (1.4, -0.8) {$q_1$};
    \node at (4.6, -0.8) {$q_2$};

    \node at (9, 0.5) {$3$-bridge};
    \end{tikzpicture}
\end{center}
Thus we have $|Q'| = \Theta(\ell d)$ (deterministically), and each path in $Q'$ contributes $\Theta(h)$ nodes to $V'$ in expectation, so we have $\mathbb{E}[n'] = \Theta(\ell d h)$.
Additionally, from the previous lemma, we have $n' = O(\ell d h)$ (deterministically).
Together, these imply that there is positive constant probability $c$ that $n' = \Theta(\ell d h)$.

Next, assuming that $\mathbb{E}[\ell'] = O(1)$, we may apply Markov's inequality (with a sufficiently large hidden constant in the $O(1)$) to conclude that
$$\Pr[\ell' = O(1)] > 1-c.$$
Hence, by an intersection bound, there is positive probability that we simultaneously have $n' = \Theta( h \ell d)$ and $\ell' = O(1)$.
In this event, we have $\|S'\| = \Theta(h \ell d^2)$, and thus
$$p \ge p' = \Omega\left( h \ell d^2 \right).$$
Rearranging, we have
\begin{align*}
\frac{p^2 n^2}{h} &= \Omega\left(\ell p \cdot (nd)^2 \right)\\
\frac{p^2 n^2}{h} &= \Omega\left( \|S\|^3 \right)\\
\|S\| &= O\left(n^{2/3} p^{2/3} h^{-1/3} \right). \tag*{\qedhere}
\end{align*}
\end{proof}

In the rest of the proof, we will make two simplifying assumptions: (1) that $\ell'$ is at least a large enough constant (otherwise we apply the previous lemma), and (2) that the forwards/backwards coin flip comes up ``forwards:'' the only place we need to flip this coin, rather than deterministically choosing ``forwards,'' is to argue that $\mathbb{E}[n'] = \Theta(\ell d h)$ in the previous lemma.
Every lemma in the rest of the proof can be immediately proved by a symmetric argument in the case where the coin flip comes up ``backwards,'' so we assume ``forwards'' for simplicity.

\subsubsection{Analysis of $\|S'\|_2^2$ and $\|S\|_2^2$}

We next make some structural observations on the intersection patterns exhibited by paths in $S$ or $S'$.
For an ordered pair of paths $(\pi_1, \pi_2) \in \Pi^2$, we define 
$$R_S(\pi_1, \pi_2) := \left\{(x, y)\in \pi_1 \times \pi_2 \ \mid \ \text{ there exists } \pi \in \Pi \text{ with } \pi \cap \pi_1 = \{x\}, \pi \cap \pi_2 = \{y\}\right\}.$$
and $R_{S'}$ is defined similarly, with paths taken from $\Pi'$ rather than $\Pi$.
We make a few observations in order to motivate this definition. 
\begin{itemize}
\item Suppose that $\pi_1, \pi_2 \in \Pi'$ and that $\pi_1$ intersects $\pi_b$ before $\pi_2$ (i.e. $(\pi_1 \cap \pi_b) <_{\pi_b} (\pi_2 \cap \pi_b)$). Then for a path $\pi$ witnessing a pair $(x, y) \in R_{S'}(\pi_1, \pi_2)$, we must specifically have that $x <_{\pi} y$.
This follows by noticing that, if instead $y <_{\pi} x$, then $\pi_1, \pi_b, \pi_2, \pi$ together form a $4$-bridge, with $\pi_1$ as the river.

\begin{center}
    \begin{tikzpicture}
    \draw [ultra thick, ->] (0, -0.5) -- (0, 1.5);
    \draw [ultra thick, ->] (1, -0.5) -- (1, 1.5);
    \draw [ultra thick, ->] (-0.5, 0) -- (1.5, 0);
    \node at (-1, 0) {$\pi_b$};
    \node at (0, -1) {$\pi_1$};
    \node at (1, -1) {$\pi_2$};
    \draw [fill=black] (0, 0) circle [radius=0.15];
    \draw [fill=black] (1, 0) circle [radius=0.15];
    
    \draw [fill=black] (0, 1) circle [radius=0.15];
    \draw [fill=black] (1, 1) circle [radius=0.15];
    
    \draw [ultra thick, <-] (-0.5, 1) -- (1.5, 1);
    
    \node at (-1, 1) {$\pi$};

    \node at (3, 0.5) {$4$-bridge};
    
    \node at (-0.3, 1.3) {$x$};
    \node at (1.3, 1.3) {$y$}; 
    
    \end{tikzpicture}
\end{center}



\item Our next observation is that it is \emph{not} possible to have pairs $(x, y), (x', y') \in R_S(\pi_1, \pi_2)$ that strictly cross each other, with $x <_{\pi_1} x'$ and $y' <_{\pi_2} y$ as in the following picture, if these pairs $(x, y), (x', y')$ are witnessed by two different paths.
The reason for this is that otherwise, they imply a $4$-bridge, with the path intersecting $(x, y)$ as the river.

\begin{center}
    \begin{tikzpicture}
    \draw [ultra thick, ->] (0, -0.5) -- (0, 1.5);
    \draw [ultra thick, ->] (1, -0.5) -- (1, 1.5);
    \draw [ultra thick, ->] (-0.5, -0.5) -- (1.5, 1.5);
    \node at (-0.4, 0) {$x$};
    \node at (-0.4, 1) {$x'$};
    \node at (1.4, 0) {$y'$};
    \node at (1.4, 1) {$y$}; 

    \node at (0, -1) {$\pi_1$};
    \node at (1, -1) {$\pi_2$};
    \draw [fill=black] (0, 0) circle [radius=0.15];
    \draw [fill=black] (1, 0) circle [radius=0.15];
    
    \draw [fill=black] (0, 1) circle [radius=0.15];
    \draw [fill=black] (1, 1) circle [radius=0.15];
    
    \draw [ultra thick, ->] (-0.5, 1.5) -- (1.5, -0.5);

    \node at (3, 0.5) {$4$-bridge};
    \end{tikzpicture}
\end{center}

\item At first, one might worry that crossing node pairs $(x, y), (x', y') \in R(\pi_1, \pi_2)$ can arise if both node pairs are caused by a single path $\pi$, as in the following picture.
However, this can arise only if $\pi$ intersects $\pi_1$ at both $x, x'$, and $\pi$ intersects $\pi_2$ at both $y, y'$.
This would imply a $2$-cycle, which we have removed from $S$ via Lemma \ref{lem:2cyccleaning}.
So this does not occur.
(This is our last use of removing $2$-cycles in the argument.)

\begin{center}
    \begin{tikzpicture}
    \draw [ultra thick, ->] (0, -0.5) -- (0, 1.5);
    \draw [ultra thick, ->] (1, -0.5) -- (1, 1.5);
    \node at (-0.4, 0) {$x$};
    \node at (-0.4, 1) {$x'$};
    \node at (1.4, 0) {$y'$};
    \node at (1.4, 1) {$y$}; 

    \node at (0, -1) {$\pi_1 $};
    \node at (1, -1) {$\pi_2$};
    \coordinate (x) at (0, 0);
    \coordinate (y') at (1, 0);
    \coordinate (x') at (0, 1);
    \coordinate (y) at (1, 1);

    \draw [fill=black] (x) circle [radius=0.15];
    \draw [fill=black] (y') circle [radius=0.15];
    
    \draw [fill=black] (x') circle [radius=0.15];
    \draw [fill=black] (y) circle [radius=0.15];  
    
    \draw [ultra thick, ->, red] plot [smooth] coordinates {(x') (x) (y) (y')};
    
    \node at (3, 0.5) {$2$-cycle};
    \end{tikzpicture}
\end{center}

\item The previous observations imply that the node pairs in $R_S(\pi_1, \pi_2)$ are arranged roughly as in the following picture, with their points of intersection with $\pi_1, \pi_2$ increasing along both paths.

\begin{center}
    \begin{tikzpicture}
    \draw [ultra thick, ->] (0, -0.5) -- (0, 2);
    \draw [ultra thick, ->] (1, -0.5) -- (1, 2);
    \draw [ultra thick, ->] (-0.5, 0) -- (1.5, 0);
    \node at (-1, 0) {$\pi_b$};
    \node at (0, -1) {$\pi_1$};
    \node at (1, -1) {$\pi_2$};
    \draw [fill=black] (0, 0) circle [radius=0.15];
    \draw [fill=black] (1, 0) circle [radius=0.15];
    
    \draw [fill=black] (0, 0.75) circle [radius=0.15];
    \draw [fill=black] (1, 1) circle [radius=0.15];

    \draw [fill=black] (0, 1.5) circle [radius=0.15];
    \draw [fill=black] (1, 1.5) circle [radius=0.15];
    
    \draw [ultra thick, ->] (-0.5, 0.625) -- (1.5, 1.125);
    \draw [ultra thick, ->] (-0.5, 1.5) -- (1.5, 1.5);
    
    \node at (-1, 0.625) {$r_1$};
    \node at (-1, 1.5) {$r_2$};

    \node at (3, 0.5) {Can happen};
    
    \end{tikzpicture}
\end{center}
\end{itemize}

Let
$$\mathcal{R}_S := \sum \limits_{(\pi_1, \pi_2) \in \Pi^2} \left| R_S(\pi_1, \pi_2) \right|.$$
The size of $\|S\|_2^2$ can be related to $\mathcal{R}_S$ as follows:
\begin{lemma} \label{lem:arrbound}
$\|S\|_2^2 = \Theta\left( \dfrac{ \mathcal{R}_S }{d^2} \right)$
\end{lemma}
\begin{proof}
Recall we have used the cleaning lemma to assume that $\ell$ is at least a large constant, and so by approximate length-regularity, we may assume that all paths in $\Pi$ have $\ge 2$ nodes.
A given path $\pi$ contributes $+|\pi|^2$ to the value of $\|S\|_2^2$.
We may therefore only count node pairs satisfying $x <_{\pi} y$, as there are $\binom{|\pi|}{2} = \Theta(|\pi|^2)$ such node pairs.

For each such node pair $(x, y)$, there are $\Theta(d)$ paths in $\Pi$ that intersect $x$ and $\Theta(d)$ paths in $\Pi$ that intersect $y$.
Thus, this node pair $(x, y)$ appears in $\Theta(d^2)$ different sets $R_S(\pi_1, \pi_2)$ with $(\pi_1, \pi_2) \in \Pi^2$.
So the pair $(x, y)$ contributes $\Theta(d)^2$ points to the value of $\mathcal{R}_S$.
It follows that
$$\|S\|_2^2 \cdot d^2 = \Theta\left(\mathcal{R}_S\right),$$
and the lemma follows by rearranging.
\end{proof}

Consider a fixed, ordered pair of paths $(\pi_1, \pi_2) \in \Pi^2$.
When we create $S'$, let us say that an ordered node pair $(x, y)$ is \emph{charged} to the pair $(\pi_1, \pi_2)$ if:
\begin{itemize}
\item $\pi_1, \pi_2 \in Q$, and hence subpaths $\pi'_1 \subseteq \pi_1, \pi'_2 \subseteq \pi_2$ are in $Q'$, and
\item $(x, y) \in R_{S'}(\pi'_1, \pi'_2)$.
\end{itemize}

The next lemma gives a lower bound on the expected number of node pairs that get charged to $(\pi_1, \pi_2)$.
Note that, in the case where $\pi_1 \notin Q$ or $\pi_2 \notin Q$, then $0$ node pairs are charged to $(\pi_1, \pi_2)$.

\begin{lemma} \label{lem:paircharge}
For each $(\pi_1, \pi_2) \in \Pi^2$, the expected number of node pairs charged to $(\pi_1, \pi_2)$ is
$$\Omega\left(|R_S(\pi_1, \pi_2)|^2 \cdot \frac{h}{\ell p}\right).$$
\end{lemma}
\begin{proof}
For pairs $(x, y), (x', y') \in R_S(\pi_1, \pi_2)$, let us say that $(x', y')$ is \emph{close behind} $(x, y)$ if
we have
$$1 \le |\pi_1[x \leadsto x']| < h \qquad \text{and} \qquad 1 \le |\pi_2[y \leadsto y']| < h.$$
We note that this definition implies that $(x, y)$ is considered to be close behind itself.

\begin{center}
    \begin{tikzpicture}
    \draw [ultra thick, ->] (0, -0.5) -- (0, 3);
    \draw [ultra thick, ->] (1, -0.5) -- (1, 3);
    \draw [ultra thick, ->] (-0.5, 0) -- (1.5, 0);
    \draw [ultra thick, ->] (-0.5, 1.625) -- (1.5, 2.125);
    \node at (0, -1) {$\pi_1$};
    \node at (1, -1) {$\pi_2$};
    \draw [fill=black] (0, 0) circle [radius=0.15];
    \draw [fill=black] (1, 0) circle [radius=0.15];
    
    \draw [fill=black] (0, 1.75) circle [radius=0.15];
    \draw [fill=black] (1, 2) circle [radius=0.15];

    \node at (-0.3, -0.3) {$x$};
    \node at (1.3, -0.3) {$y$};
    \node at (-0.3, 1.4) {$x'$};
    \node at (1.3, 1.7) {$y'$};

    \node [blue] at (-0.5, 0.7) {$< h$};
    \node [blue] at (1.5, 1) {$< h$};

    \node [align=center] at (7, 1) {$(x', y')$ is close behind $(x, y)$ if\\both marked segments contain $<h$ nodes.};
    
    \end{tikzpicture}
\end{center}

The point of this definition is that, if the path containing $(x, y)$ is selected as the base path $\pi_b$, then we will charge some node pairs to $(\pi_1, \pi_2)$, and the number of such node pairs is exactly the number of pairs in $R_S(\pi_1, \pi_2)$ that are close behind $(x, y)$.

Let us say that $(x, y)$ is \emph{typical} if, for some sufficiently large constant $C$, there are at least
$$|R_S(\pi_1, \pi_2)| \cdot \frac{h}{C\ell}$$
pairs in $R_S(\pi_1, \pi_2)$ that are close behind $(x, y)$.
If we happen to sample a base path $\pi_b$ that contains a typical pair $(x, y) \in R_S(\pi_1, \pi_2)$, then the number of pairs charged to $(\pi_1, \pi_2)$ is
$$\Omega\left(|R_S(\pi_1, \pi_2)| \cdot \dfrac{h}{\ell}\right).$$
Thus, to prove the lemma, it suffices to prove that we sample a base path $\pi_b$ that contains a typical pair $(x, y) \in R_S(\pi_1, \pi_2)$ with probability
$$\Omega\left(|R_S(\pi_1, \pi_2)| \cdot \frac{1}{p} \right);$$
that is, a constant fraction of the pairs in $R_S(\pi_1, \pi_2)$ are typical.
We show this using the following counting argument:

By our analysis of the structure of path intersections, the pairs $(x, y) \in R_S(\pi_1, \pi_2)$ may be totally ordered by positions of $x \cap \pi_1, y \cap \pi_2$.
Consider these pairs $(x, y)$ in \emph{increasing} order, that is, the first pair $(x, y)$ considered is the one that intersects $\pi_1, \pi_2$ closest to their start nodes.
When each pair $(x, y)$ is considered, if it is typical, then add $1$ to the count of typical pairs.
Otherwise, if $(x, y)$ is not typical, then we throw away the next
$$|R_S(\pi_1, \pi_2)| \cdot \frac{h}{C \ell}$$
pairs in the ordering (including $(x, y)$; these pairs are thrown out without increasing the count whether or not they are typical), and then we continue.
Each time we skip a pair $(x, y)$, notice that the next pair considered $(x', y')$ is \emph{not} close behind $(x, y)$, which means we have
$$|\pi_1[x \leadsto x']| \ge h \qquad \text{or} \qquad |\pi_2[y \leadsto y']| \ge h.$$
That is, we progress at least $h$ nodes along $\pi_1$, or at least $h$ nodes along $\pi_2$.
Since $|\pi_1| = O(\ell)$ and $|\pi_2| = O(\ell)$, we can thus only perform this skip operation $O(\ell/h)$ times. 
By unioning, it follows that only
$$O\left(|R_S(\pi_1, \pi_2)| \cdot \frac{h}{C\ell} \right) \cdot O(\ell/h) = O\left(|R_S(\pi_1, \pi_2)| \cdot \frac{1}{C\ell} \right)$$
pairs get discarded.
By choice of large enough $C$, this is only a constant fraction of the total pairs in $R_S(\pi_1, \pi_2)$, which means a constant fraction of the pairs in $R_S(\pi_1, \pi_2)$ are counted as typical.
By the previous discussion, the lemma follows.
\end{proof}

\begin{lemma} \label{lem:arrsumbound}
$\mathbb{E}\left[ \| S' \|_2^2 \right] = \Omega\left(\frac{h}{\ell p^3} \cdot \mathcal{R}_S^2 \right)$
\end{lemma}
\begin{proof}
We can lower bound the expected value of $\|S'\|_2^2$ by counting the expected number of node pairs $(x, y)$ that get charged to some $(q_1, q_2) \in Q^2$.
By Lemma \ref{lem:paircharge}, this is
\begin{align*}
\mathbb{E}[\{|(x, y) \in V^2 \ \mid \ (x, y) \text{ charged} \}] &= \Omega\left(\sum \limits_{(\pi_1, \pi_2) \in \Pi^2} \left| R_S(\pi_1, \pi_2)\right|^2 \cdot \frac{h}{\ell p} \right)\\
&= \frac{h}{\ell p^3} \cdot \Omega\left(\sum \limits_{(\pi_1, \pi_2) \in \Pi^2} \left| R_S(\pi_1, \pi_2)\right|^2  \right) \cdot p^2\\
&=  \Omega\left(\frac{h}{\ell p^3} \cdot \left( \sum \limits_{(\pi_1, \pi_2) \in \Pi^2} \left| R_S(\pi_1, \pi_2)\right|\right)^2 \right) \tag*{Cauchy-Schwarz}\\
&=  \Omega\left(\frac{h}{\ell p^3} \cdot \mathcal{R}_S^2 \right). \tag*{\qedhere}\\
\end{align*}
\end{proof}

\subsection{Proof Wrapup}

The remainder of the proof is essentially just algebra.
Using the previous two lemmas, we have:
\begin{lemma} \label{lem:shbound}
$\|S\| = O\left( n^{3/4} p^{1/2} + n^{8/13} p^{9/13} h^{1/13}
 \right)$.
\end{lemma}
\begin{proof}
We have
\begin{align*}
\|S\|_2^2 &= \Theta\left( d^{-2} \mathcal{R}_S \right) \tag*{Lemma \ref{lem:arrbound}}\\
&= d^{-2} \cdot O\left( \frac{\ell p^3}{h} \cdot \mathbb{E}\left[\|S'\|_2^2 \right] \right)^{1/2} \tag*{Lemma \ref{lem:arrsumbound}}\\
&= O\left( \frac{\ell p^3}{d^4 h} \cdot \mathbb{E}\left[\|S'\|_2^2 \right] \right)^{1/2}.
\end{align*}
Recall from Lemma \ref{lem:sprimeprops} that $S'$ has $O(\ell d h)$ nodes, $\le p$ paths, maximum path length $O(\ell)$, and bridge girth $>3$.
We may thus apply Lemma \ref{lem:2normbound} to conclude that the following bound holds deterministically:
$$\|S'\|_2^2 = O\left( (\ell d h)\cdot \ell + p^{1/3} (\ell d h)^{4/3} \right).$$
Using this as an upper bound for expectation, we may continue:
\begin{align*}
\|S\|_2^2 &= O\left( \frac{\ell p^3}{d^4 h} \cdot \left( \ell^2 d h + p^{1/3} (\ell d h)^{4/3} \right) \right)^{1/2}\\
&= O\left( \ell^3 p^3 d^{-3} + \ell^{7/3} p^{10/3} d^{-8/3} h^{1/3} \right)^{1/2}.
\end{align*}
Our next step will be to apply the Cauchy-Schwarz inequality, which gives $\|S\|_2^2 \cdot p \ge \|S\|^2$.
Using this, we may continue:
\begin{align*}
\frac{\|S\|^2}{p} &= O\left( \ell^3 p^3 d^{-3} + \ell^{7/3} p^{10/3} d^{-8/3} h^{1/3} \right)^{1/2}\\
\frac{\|S\|^4}{p^2} &= O\left( \ell^3 p^3 d^{-3} + \ell^{7/3} p^{10/3} d^{-8/3} h^{1/3} \right)\\
\|S\|^4 &= O\left( \ell^3 p^5 d^{-3} + \ell^{7/3} p^{16/3} d^{-8/3} h^{1/3}  \right)\\
\|S\|^4 &= O\left( n^3 p^2 + \|S\|^{-1/3} n^{8/3} p^{3} h^{1/3} \right) \tag*{$\|S\| = nd = p\ell$.}
\end{align*}
We next split into two cases, by which of these terms in the right-hand side dominate.
If the first term dominates, then we get
$$\|S\| = O\left(n^{3/4} p^{1/2}\right).$$
If the second term dominates, then we get
\begin{align*}
\|S\|^4 &= O\left( \|S\|^{-1/3} n^{8/3} p^{3} h^{1/3} \right)\\
\|S\|^{13} &= O\left( n^8 p^9 h\right)\\
\|S\| &= O\left( n^{8/13} p^{9/13} h^{1/13}\right).
\end{align*}
Summing the two cases gives our claimed bound.
\end{proof}

The next step is to choose $h$ to balance terms.
We are balancing the term $n^{8/13} p^{9/13} h^{1/13}$ in the previous lemma, which applies in the case where the expected average path length in $S'$ is at least a large constant, with the term $n^{2/3} p^{2/3} h^{-1/3}$ from Lemma \ref{lem:degenSprime} which applies in the case where the average path length in $S'$ is bounded by a constant.
The proper setting of $h$ is computed as:
\begin{align*}
n^{8/13} p^{9/13} h^{1/13} &= n^{2/3} p^{2/3} h^{-1/3}\\
h^{16/39} &= n^{2/39} p^{-1/39}\\
h &= n^{1/8} p^{-1/16}.
\end{align*}
These terms then balance at
\begin{align*}
&n^{2/3} p^{2/3} \left( n^{1/8} p^{-1/16} \right)^{-1/3}\\
=&n^{5/8} p^{11/16}.
\end{align*}
Thus, the total bound on the size of $\|S\|$ is
$$\|S\| = O\left( n^{3/4} p^{1/2} + n^{5/8} p^{11/16} \right).$$
Finally, we acknowledge that this proof used the assumption that $\ell, d$ were both at least sufficiently large constants, so the final size of $\|S\|$ also absorbs a $O(n+p)$ term.

\subsection{Bounds for $k=\infty$}
\label{sec:inf_bounds}
\begin{theorem} \label{thm:rsreverse}
$\beta(n, p, \infty) = O\left( \frac{p^2}{2^{O(\log^* p)}} + p \right)$
\end{theorem}
\begin{proof}
In this proof, it will be helpful to write $\abeta(n, p, k)$ for the extremal function of \emph{acyclic} path systems of high bridge girth.
Let $S$ be an acyclic path system with $n$ nodes, $p$ paths, bridge girth $>3$, and $\|S\| = \abeta(n, p, 3)$.
Let $G$ be its incidence graph, which has $|E(G)| = \|S\|$ edges.
We claim that $G$ has girth $>6$.
To see this, we argue:
\begin{itemize}
\item Since $G$ is bipartite it does not have $3$- or $5$-cycles.

\item A $4$-cycle in $G$ implies that there are two paths $\pi_1, \pi_2 \in \Pi$ that intersect the same pair of nodes $u, v \in V$.
Since $S$ is acyclic, these paths would need to use $u, v$ in the same order, and hence they would form a $2$-bridge.
But since $S$ has bridge girth $>3$, this means $G$ may not have a $4$-cycle.

\item Similarly, a $6$-cycle in $G$ implies that there are three paths $\pi_1, \pi_2, \pi_3$ and three nodes $t, u, v$ with $t, u \in \pi_1, u, v \in \pi_2, t, v \in \pi_3$.
Since $S$ is acyclic, we may assume without loss of generality that these nodes are ordered $t, u, v$ in a topological sort.
Thus $\pi_1, \pi_2, \pi_3$ form a $3$-bridge with $\pi_3$ as the river.
Since $S$ has bridge girth $>3$, this means $G$ may not have a $6$-cycle.
\end{itemize}
Hence $G$ is a bipartite graph with $n, p$ nodes per side, girth $>6$, and $\abeta(n, p, 3)$ edges.
It follows that\footnote{In fact this reduction can be reversed, showing that these functions are equal.  But we will only need to use the inequality in one direction, so we omit the proof on the other side.}
$$\abeta(n, p, 3) \le \gamgam(n, p, 6).$$
We also notice that $\gamgam$ is symmetric in its first two parameters, and so
$$\gamgam(n, p, 6) = \gamgam(p, n, 6).$$
With these two facts in mind, we now have:
\begin{align*}
\beta(n, p, \infty) &= \Theta\left( \abeta(n, p, \infty)\right) \tag*{Corollary \ref{cor:inftyacyclic}}\\
&\le O\left( \abeta(n, p, 3) \right)\\
&\le O\left( \gamgam(n, p, 3) \right)\\
&= O\left( \gamgam(p, n, 3) \right)\\
&\le O\left( \frac{p^2}{\rs(p)} + n\right) \tag*{\cite{de1997dense, Fox11, MS19}. \qedhere}
\end{align*}
\end{proof}

\begin{theorem} \label{thm:bstarupper}
$\betastar(n, p, \infty) = O\left( \min\left\{ n^{1/2}p, np^{1/2} \right\} + n + p\right)$
\end{theorem}
\begin{proof}
This follows as a consequence of Corollary \ref{cor:bstaracyclic}.
Recall that it suffices to bound the maximum possible size of an \emph{acyclic} ordered path system $S$ with bridge girth $\infty$.
If we drop the path ordering in $S$, and treat it as a normal unordered path system, then it still has bridge girth $>2$.
Since $S$ is acyclic, this implies that any two paths intersect on at most one node.
So if we further drop the order of nodes within each path, and instead treat $S$ as a \emph{set} system, it has girth $>2$.
It thus satisfies
$$\|S\| = O\left( \min\left\{ n^{1/2}p, np^{1/2} \right\} + n + p\right)$$
by well-known bounds on the size of high-girth set systems (or bipartite graphs) \cite{van12}.
\end{proof}

\begin{theorem} \label{thm:bstarlower}
$\betastar(n, p, \infty) = \Omega\left(n^{2/3} p^{2/3} + n + p\right)$
\end{theorem}
\begin{proof}
This construction is a slight modification of a construction from \cite{ST83}, although the analysis is somewhat new.
We will construct an ordered path system $S = (V, \Pi)$ on $n = |V| $ nodes and $p = |\Pi|$ paths realizing this lower bound.
Let $1 \le \ell \le n^{1/2}$ be an integer parameter; all paths in our system will have length exactly $\ell$.
Let
$$V := [1, \ell] \times \left[1, \frac{n}{\ell}\right]$$
be a rectangular subset of the integer lattice $\mathbb{Z}^2$.\footnote{We will assume for convenience that $n/\ell$ and similar terms are integral; otherwise, rounding to the nearest integer affects our argument only by lower-order terms which may be ignored.}
Our paths in $\Pi$ will correspond to lines in $\mathbb{Z}^2$.
The starting points of our lines (paths) are captured by the set
$$X := \{1\} \times \left[1, \frac{n}{2\ell}\right].$$
The slopes of our lines are captured by the set
$$W := \left\{(1, i) \mid i \in \left[1, \frac{n}{2\ell^2}\right]\right\}.$$
For each $x \in X$ and $\vec{w} \in W$, we add the path
$$\left(x, x + \vec{w}, x + 2\vec{w}, \dots, x + (\ell-1)\vec{w}\right)$$
to $\Pi$.
To order our paths: for paths $\pi, \pi' \in \Pi$, we assign $\pi < \pi'$ if $\Vec{w}_y < \Vec{w}'_y$, where $\vec{v}_y$ denotes the $y$-component of vector $\vec{v} \in \mathbb{Z}^2$; if $\Vec{w}_y = \Vec{w}'_y$ then the tie may be broken arbitrarily.
Intuitively, this orders our paths in $\Pi$ by increasing value of the slope of the corresponding line in $\mathbb{Z}^2$.
This completes the construction, and we now check its parameters.
We have
$$p = |\Pi| = |X||W| = \Theta\left(\frac{n^2}{\ell^3}\right).$$
Additionally, all paths in $\Pi$ have length exactly $\ell$, and so
$$\|S\| = p\ell = p \cdot \Theta\left(\frac{n^{2/3}}{p^{1/3}}\right) = \Theta\left(n^{2/3}p^{2/3}\right).$$

Now it only remains to verify that $S$ is an ordered bridge-free path system.
Let $\pi \in \Pi$ be a path constructed via start point $x \in X$ and slope $\vec{w} \in W$, and suppose for the sake of contradiction that $\pi$ is the river for an ordered bridge in $S$.
Since $S$ is layered, this implies there is a collection of vectors $\vec{v}_1, \vec{v}_2, \dots, \vec{v}_{\ell-1} \in W$ (not necessarily distinct), which correspond to the differences among adjacent nodes along the path formed by the arcs of the bridge, and such that
$$\sum_{i=1}^{\ell-1} \vec{v}_i = (\ell-1)\vec{w}.$$
Moreover, the vectors $\{\vec{v}_i\}$ may not all be identical to $\vec{w}$.
This implies that there exists at least one vector $\vec{v}_i$ with $(\vec{v}_i)_y > \vec{w}_y$, and there exists at least one vector $\vec{v}_i$ with $(\vec{v}_j)_y < \vec{w}_y$.
However, since we order our paths by increasing slope, the arc corresponding to $\vec{v}_i$ would be placed later in the ordering than the path corresponding to $\vec{w}$.
Thus this arc cannot participate as an arc of an ordered bridge with $\vec{w}$ corresponding to the river.
This completes the contradiction, and we conclude that $\pi$ has no ordered bridge.
\end{proof}

Finally, we show that this lower bound is conditionally tight in the setting $p=n$.
We first give some background on the relevant condition.
An \emph{ordered graph} is a simple graph with a total ordering on its vertices.
A natural problem is to investigate how classic results from extremal graph theory extend to this setting.
Let:
\begin{itemize}
\item $\text{ex}(n, H)$ be the \Turan{} function of the graph $H$; that is, the maximum possible number of edges in an $n$-node graph that does not contain $H$ as a subgraph.
\item  $\text{ex}_<(n, H)$ be the \emph{ordered} \Turan{} function of the \emph{ordered} graph $H$, defined analogously.
\end{itemize}
These extremal functions satisfy the basic inequality 
$ \text{ex}_<(n, H) \geq \text{ex}(n, \overline{H})$, where $\overline{H}$ denotes the (unordered) graph underlying the ordered graph $H$.
Tardos \cite{Tardos18} asked how high the ratio can be, that is, the value of
$$\max_H \dfrac{\text{\normalfont ex}_<(n, H)}{\text{\normalfont ex}(n, \overline{H})}.$$
The current lower bound is $\Omega(n^{1/3 - \varepsilon})$, and the current upper bound is $O(n^{1-\varepsilon})$.
A reasonable hypothesis could be that the lower bound is closer to the correct answer:
\begin{hyp}
\label{hyp:ordered_gap}
Let $H$ be an ordered graph with greater than two vertices and at least one edge. Then
$\frac{\text{\normalfont ex}_<(n, H)}{\text{\normalfont ex}(n, \overline{H})} = O(n^{1/3})$.
\end{hyp} 
Under this hypothesis, our new lower bound for $\beta^*(n, p, \infty)$ is nearly-tight when $p = n$. 
To be clear, we do not necessarily think there is evidence that Hypothesis \ref{hyp:ordered_gap} is true.
However, we do think that it represents a natural limitation on current methods in the theory of ordered graphs, and a significant new idea will be needed to prove or refute it.
Thus, our point is simply that it is likely beyond the reach of current techniques to significantly improve our lower bound on $\beta^*(n, n, \infty)$ (if it is improvable at all).

\begin{theorem} \label{thm:hypbstarlb}
    Under Hypothesis \ref{hyp:ordered_gap}, $\beta^*(n, n, \infty) = \Thetaish(n^{4/3})$. 
\end{theorem}
\begin{proof}
Note that by Theorem \ref{thm:bstarlower}, $\beta^*(n, n, \infty) = \Omega(n^{4/3})$. We will prove that under Hypothesis \ref{hyp:ordered_gap}, $\beta^*(n, n, \infty) = O(n^{4/3} \log n)$. 

Let $S = (V, \Pi)$ be an acyclic ordered path system with $n$ nodes,  $n$ paths, no ordered bridges, and size $\|S\| = \Omega(\beta^*(n, n, \infty))$. ($S$ must exist by the cleaning lemma and by Corollary \ref{cor:bstaracyclic}.)
Since $S$ is acyclic, there is a total order $\sigma_1$ on $V$ such that for all $s, t \in V$,  $s <_{\sigma_1} t$ if $s <_{\pi} t$ for some $\pi \in \Pi$.  
Likewise, let $\sigma_2$ be the total order on $\Pi$ in the ordered path system $S$.

Let $G_S$ be the \textit{incidence graph} corresponding to $S$.
Recall that this means: $G_S = (L \cup R, E)$ is the bipartite graph such that $L := V$,  $R := \Pi$, and for all $v \in V$ and $\pi \in \Pi$, $(v, \pi) \in E$ if vertex $v$ is contained in path $\pi$ in $S$. Note that $|E| = \|S\|$. Define the following total order on the vertices  $L \cup R$ of $G_S$:
\begin{itemize}
    \item If $u, v \in L$, then apply ordering $\sigma_1$. 
    \item If $u, v \in R$, then apply ordering $\sigma_2$.
    \item If $u \in L$ and $v \in R$, then let $u < v$. 
\end{itemize}
We claim that graph $G_S$ (with vertices ordered as above) does not contain any simple, ordered $2k$-cycles $H_k$ of form 
$$
H_k := (v_1, \pi_1, v_2, \pi_2, \dots, v_k, \pi_k),
$$
where 
\begin{itemize}
    \item $v_1 < \pi_1$,
    \item $v_i < v_j$ if $i < j \in [1, k]$,
    \item $\pi_i < \pi_k$ for $i \in [1, k-1]$,
    \item and $k \geq 2$. 
\end{itemize}
Suppose for the sake of contradiction that such an ordered graph $H_k$ is contained in $G_S$. Then since $G_S$ is bipartite and $v_1 < \pi_1$ in $H_k$, we must have that $v_i \in L$ and $\pi_i \in R$  for $i \in [1, k]$. Additionally, by the choice of edges $E$ in $G_S$, it follows that $v_i, v_{i+1} \in \pi_i$ for $i \in [1, k-1]$ and $v_1, v_k \in \pi_k$ in $S$. Moreover, by our choice of ordering $\sigma_1$, since $v_i < v_{i+1}$ in $H$ and $v_i, v_{i+1} \in \pi_i$ it follows that $v_i <_{\pi_i} v_{i+1}$ and $v_1 <_{\pi_k} v_k$ in $S$. Finally, since $\pi_i < \pi_k$ for all $i \in [1, k-1]$ in $H_k$,  path $\pi_k$ comes last in the total order $\sigma_2$ of paths $\Pi$ in $S$, so nodes $v_1, \dots, v_k \in V$ and paths $\pi_1, \dots, \pi_k \in \Pi$ correspond to an ordered $k$-bridge in $S$. This contradicts our assumption that $S$ has no ordered bridges, so we conclude that $G_S$ does not contain $H_k$. 
Consequently,  the size of $S$ is at most
$$
\|S\| = |E| \leq \text{ex}_<(2n, H_k).
$$

By known bounds on the extremal function of cycles,  $\text{ex}(n, \overline{H_k}) = O(k n^{1+1/k})$  \cite{bondy1974cycles}. Then under Hypothesis \ref{hyp:ordered_gap},
$$
\|S\| \leq \text{ex}_<(2n, H) \leq \text{ex}(2n, \overline{H}) \cdot  O(n^{1/3}) = O(kn^{4/3 +1/k}). 
$$
Taking $k = \log n$ completes the proof.
\end{proof}

%% file: reductions.tex
\section{Extremal Reductions to Bridge Girth}

In this section, we discuss various objects in network design where the extremal state-of-the-art upper or lower bounds on size can be reduced to extremal functions of bridge girth.

\subsection{Distance Preservers}

Recall Definition \ref{def:dps} for the formal definition of distance preservers and their associated extremal function $\dpp$.
Our goal is to prove:
\begin{theorem} \label{thm:distpreschain}
$\Omega\left(\beta^*(n, p, \infty)\right) \le \dpp(n, p) \le \beta(n, p, 2)$.
\end{theorem}

We begin by proving the upper bounds of Theorem \ref{thm:distpreschain}.
As a warmup, let us restrict attention to a very specific kind of distance preserver input $G = (V, E, w), P$ that enjoys a property that we will call \emph{independence} among the demand pairs:
\begin{definition} [Independence]
For a graph $G = (V, E, w)$ and set of demand pairs $P \subseteq V \times V$, we say that $P$ is \emph{independent} in $G$ if for all $(s, t) \in P$ there is a unique shortest path $\pi(s, t)$ in $G$, and these paths are pairwise edge-disjoint.
\end{definition}

An independent input instance $G = (V, E, w), P$ can be naturally associated to a path system $S = (V, \Pi)$, where the paths in $\Pi$ are precisely the node sequences corresponding to the unique shortest paths for demand pairs in $P$.
We claim that this system $S$ is $2$-bridge-free.
To see this, notice that $S$ has a $2$-bridge iff there are two distinct nodes $u, v \in V$ and two distinct paths $\pi_1, \pi_2 \in \Pi$ that both contain $u$ and then $v$ (in that order).
On one hand, we cannot have $\pi_1, \pi_2$ coincide on their $u \leadsto v$ subpaths, as this would imply that the associated paths in $G$ share edges, violating independence.
On the other hand, we cannot have that $\pi_1, \pi_2$ use distinct $u \leadsto v$ subpaths, as this would violate the property that $\pi_1, \pi_2$ are each unique shortest paths in the underlying graph.
Thus it is not possible for $S$ to have a $2$-bridge, and so we have
$$\|S\| \le \beta(n, p, 2).$$
The natural distance preserver for $G, P$ is obtained by overlaying the unique edge-disjoint shortest paths for the demand pairs in $P$, and it has exactly $\|S\| - p$ edges.\footnote{The $-p$ term arises since $\|S\|$ counts the number of \emph{nodes} in each path, while for the distance preserver we count the number of \emph{edges} in each path.}
Thus any such independent instance $G, P$ has a distance preserver on $\beta(n, p, 2) - p$ edges, which satisfies Theorem \ref{thm:distpreschain}.

This part of the proof is not exactly surprising, and it is essentially a rephrasing of the well-known fact that unique shortest paths in graphs exhibit \emph{consistency}.
The more interesting part of the proof is to show that independence of input instances may be assumed without loss of generality.
This is accomplished in the following lemma:

\begin{lemma} [Independence Lemma for Distance Preservers] \label{lem:dpindep}
For any positive integers $n, p$, there exists an $n$-node graph $G$ and a set of $|P| \le p$ \textbf{independent} demand pairs such that the minimal distance preserver of $G, P$ has exactly $\dpp(n, p)$ edges.
\end{lemma}
\begin{proof}
Let $G, P$ be a (not necessarily independent) input instance on $n$ nodes and $p$ paths realizing $\dpp(n, p)$.
We may assume without loss of generality that $G$ itself has exactly $\dpp(n, p)$ edges (otherwise, replace $G$ with a distance preserver of $G, P$ and then perform the following analysis).
We may also assume without loss of generality that every demand pair in $P$ has a unique shortest path in $G$.
This follows by the standard method of random reweighting: that is, for each edge $e$, randomly choose a real number in the range $[0, \eps]$ and add this number to $w(e)$.
Shortest path ties are broken with probability $1$, and if we choose $\eps>0$ small enough, the changes in edge weights will not cause a previously non-shortest path to become a shortest path.

The instance $G, P$ might still not be independent, because the unique shortest paths for the pairs in $P$ might overlap on edges.
To remove overlap, we will further modify $G, P$ by executing either of the following two steps until neither one is possible.
In the following, for a demand pair $(s, t) \in P$ we will write $\pi(s, t)$ for its unique shortest path, and we will say that $\pi(s, t)$ \emph{uniquely uses} an edge $e$ if $e \in \pi(s, t)$ and there is no other demand pair $(s', t') \in P$ with $e \in \pi(s', t')$.
\begin{itemize}
%

    \item If there exists a demand pair $(s, t) \in P$ that does not uniquely use any edges, we delete $(s, t)$ from $P$.
    Note that it is still the case that every edge in $G$ is used by at least one unique shortest path for a demand pair.
    
    \item Suppose that there exists a demand pair $(s, t) \in P$ and a sequence of three contiguous nodes $(x, y, z) \in \pi(s, t)$, such that $\pi(s, t)$ uniquely uses one of the two edges $\{(x, y), (y, z)\}$, but the other of these two edges is used by another unique shortest path $\pi(s', t')$ as well.
    For ease of notation we will assume that $(x, y)$ is the edge uniquely used by $\pi(s, t)$; the other case is symmetric.
    We then add $(x, z)$ to $G$ as a new edge, and we set its weight to $w(x, z) := w(x, y) + w(y, z)$, and we delete the edge $(x, y)$.
    Notice that:
    \begin{itemize}
    \item We add $(x, z)$ to $G$ and we remove $(x, y)$ from $G$, so the number of edges in $G$ stays the same.  The demand pair $(s, t)$ still has a unique shortest path, which now uses the edge $(x, z)$ in place of the $2$-path $(x, y, z)$.
    \item No other unique shortest path besides $\pi(s, t)$ is affected by this change to the structure of the edges in $G$.
    This holds because $(x, y, z)$ must be the unique $x \leadsto z$ shortest path, and $\pi(s, t)$ is the only path that uses $(x, y)$, and therefore $\pi(s, t)$ is the only path that contains the nodes $(x, z)$ in that order.
    \item The other edge $(y, z)$ is still used by a unique shortest path, since by hypothesis we have $(y, z) \in \pi(s', t')$ for some other demand pair $(s', t')$.
    \end{itemize}
\end{itemize}

\begin{figure}[htbp]
  \begin{center}
    \includegraphics[width=5.0in]{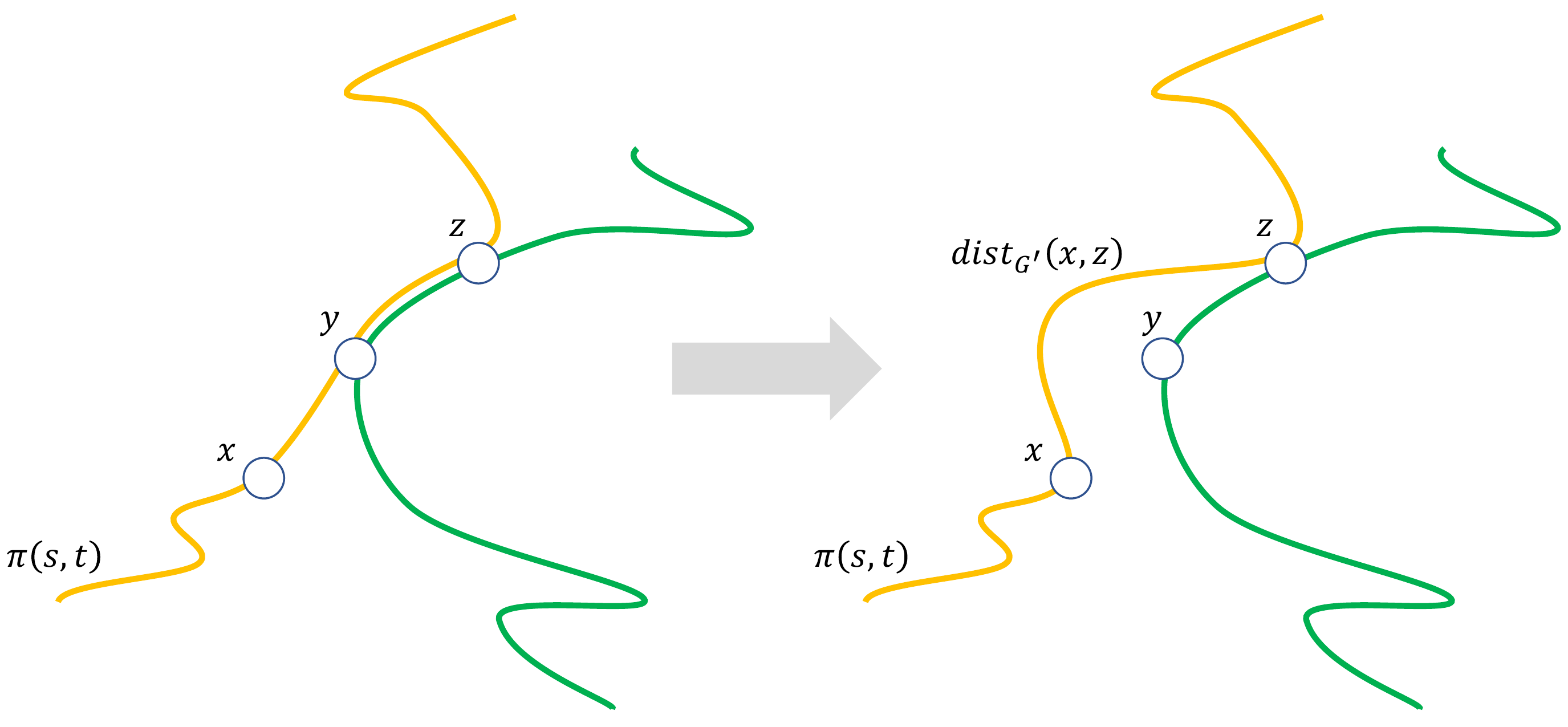}
    \end{center}
    \caption{The ``edge skip'' operation used in the second case of the proof of Lemma \ref{lem:dpindep}}
    \label{Figs:independence}
\end{figure}
    
In either case, the sum of lengths of the paths $\{\pi(s, t)\}_{(s, t) \in P}$ decreases by at least 1, and thus the process eventually terminates.
Once it terminates, every unique shortest path $\pi(s, t)$ uniquely uses at least one edge, and it does not contain two consecutive edges where one is uniquely used and the other is not.
Therefore $\pi(s, t)$ uniquely uses \emph{all} of its edges, which implies independence.
\end{proof}





This implies:
\begin{lemma} \label{thm:dirdps}
$\dpp(n, p) \le \beta(n, p, 2)$.
\end{lemma}
\begin{proof}
By the independence lemma (\ref{lem:dpindep}), there exists an $n$-node graph and a set of $|P| \le p$ independent demand pairs that have $\dpp(n, p)$ edges in their union.
By an earlier discussion, any such instance can be associated to a $2$-bridge-free path system of size $\dpp(n, p) + p$.
It follows that $\dpp(n, p) \le \beta(n, p, 2)$.
\end{proof}



Now we turn to the lower bounds:
\begin{lemma} \label{lem:dplower}
$\Omega(\beta^*(n, p, \infty)) \le \dpp(n, p)$
\end{lemma}
\begin{proof}
Let $S = (V, \Pi)$ be an ordered path system with $n$ nodes, $p$ paths, no ordered bridges, and size $\|S\| = \beta^*(n, p, \infty)$.
We will use $S$ to construct an independent instance $G, P$ that requires $\|S\| - p$ edges for any distance preserver.

Starting with an empty graph $G = (V, \emptyset)$, consider the paths in $\Pi$ in their order in $S$.
When a path $\pi(s, t) \in \Pi$ is considered, add $(s, t)$ as a demand pair to $P$, and then add all consecutive pairs of nodes on $\pi(s, t)$ as new \textit{directed} edges in $G$.
Since $S$ is $2$-bridge-free, the paths $\{\pi(s, t)\}$ are pairwise edge-disjoint, and moreover since $S$ has no ordered bridges, each path $\pi(s, t)$ is the unique $s \leadsto t$ path in $G$ at the time it is added.
Hence it is the unique \emph{shortest} $s \leadsto t$ path at the time it is added, regardless of the edge weights we assign to its edges.
We may therefore choose sufficiently large weights for the new edges on $\pi(s, t)$,  so that 
no previously-added demand pair will gain a new shortest path using any edges in $\pi(s, t)$.

Once all paths in $\Pi$ have been considered, we have unique edge-disjoint shortest paths in $G$ for all $p$ demand pairs, and the union of these paths contains 
$$\betastar(n, p) - p = \Omega(\betastar(n, p))$$
edges.
Thus, we can interpret this graph $G$ together with the set $P$ holding the endpoints of the paths in $\Pi$ as an input instance for distance preservers.
All edges in paths in $\Pi$ must remain in a distance preserver, which implies
\begin{align*}
\Omega\left(\betastar(n, p)\right) \le \dpp(n, p). \tag*{\qedhere}
\end{align*}
\end{proof}

\subsection{Shortest Path Oracles}

We next prove an \emph{incompressibility theorem} for distance preservers.
We consider shortest path oracles, which are the natural data structure version of distance preservers:
\begin{definition} [Path Oracles]
Given a directed graph $G = (V, E, w)$ and a set of demand pairs $P$, a shortest path oracle is a data structure that, when queried with $(s, t) \in P$, can report a shortest $s \leadsto t$ path in $G$ (or ``no path'' if none exists).

We define $\texttt{SPO}(n, p)$ as the smallest integer such that every $n$-node graph and set of $|P|=p$ demand pairs has a path oracle on $\le \texttt{SPO}(n, p)$ bits.
\end{definition}

Note that a distance preserver of $G, P$ on $m$ edges implies a shortest path oracle of $G, P$ on $O(m \log n)$ bits.
Consequently, we have
$$\texttt{SPO}(n, p) \le O\left( \dpp(n, p) \log n \right).$$
The following theorem states that we cannot expect much smaller shortest path oracles in general.

\begin{theorem} \label{thm:shortestpathoracle}
$\Omega\left( \dpp(n, p) \right) \le \texttt{SPO}(n, p) \le O\left( \dpp(n, p) \log n \right)$.
\end{theorem}
\begin{proof}
Let $G = (V, E, w)$ be an $n$-node graph and let $P$ be a set of $|P| = p$ independent demand pairs such that the minimal distance preserver of $G, P$ has exactly $\dpp(n, p)$ edges. Note that such $G$ and $P$ exist by the independence lemma for distance preservers (Lemma \ref{lem:dpindep}). We associate with $G, P$ a path system $S = (V, \Pi)$ where the paths $\pi(s, t)$ in $\Pi$ are precisely the node sequences corresponding to the unique shortest paths for demand pairs $(s, t)$ in $P$. Note that
$$\|S\| = \sum_{\pi \in \Pi}|\pi| =  \dpp(n, p) + p.$$
Let $S' = (V, \Pi')$ be a path system obtained from $S$ by replacing each $s \leadsto t$ path $\pi(s, t)$ in $\Pi$ with an arbitrary $s\leadsto t$ subpath $\pi'(s, t)$. Let $\mathcal{S}$ denote the set of all path systems $S'$ generated in this way. Note that the total number of pairs $(v, \pi) \in V \times \Pi$ such that $v$ is an internal node in $\pi$ is  $\dpp(n, p) - p$, so $|\mathcal{S}| = 2^{\dpp(n, p) - p}$.

For each path system $S' \in \mathcal{S}$, we construct a graph $G_{S'} = (V, E', w')$ as follows. For each path $\pi'(s, t) \in \Pi'$, we add all consecutive pairs of nodes on $\pi'(s, t)$ to $E'$. For each edge $e = (u, v) \in E'$, we assign the weight $\dist_G(u, v)$ to the edge $(u, v)$ in $G_{S'}$. Observe that by our choice of weights, for all $s, t \in V$,   $\dist_G(s, t) \leq \dist_{G_{S'}}(s, t)$. Additionally, for every $(s, t) \in P$, the path in $G_{S'}$ corresponding to $\pi'(s, t) \in \Pi'$ has path length exactly $\dist_G(s, t)$. Since $G, P$ is independent, $\pi(s, t)$ is a unique shortest $s \leadsto t$ path in $G$, and so $\pi'(s, t)$ is a unique shortest $s \leadsto t$ path in $G_{S'}$.

We have shown that every demand pair $(s, t) \in P$ has a unique shortest path in $G_{S'}$ that corresponds exactly to path $\pi'(s, t) \in \Pi'$ of $S'$. Then any shortest path oracle for $G_{S'}, P$ will have to output the path $\pi'(s, t)$ when queried with $(s, t)$. Now consider the graph family $\mathcal{G} = \{G_{S'}\}_{S' \in \mathcal{S}}$. Any two distinct graphs $G_1, G_2 \in \mathcal{G}$ will require distinct shortest path oracle data structures, since the corresponding path systems $S_1', S_2'\in \mathcal{S}$ are distinct. Consequently, at least one of the graphs in $\mathcal{G}$ will need 
$$\Omega\left(\log |\mathcal{G}|\right) = \Omega\left(\dpp(n, p)\right)$$
bits to represent its shortest path oracle data structure.
\end{proof}


\subsection{Reachability Preservers}

Here we will prove:
\begin{theorem} \label{thm:rpreduction}
$\rpp(n, p) = \Theta(\beta(n, p, \infty))$.
\end{theorem}

We note that $\rpp(n, p)$ is only well-defined in the range $p \le O(n^2)$, so naturally we prove Theorem \ref{thm:rpreduction} only in this parameter range.
We again start with the upper bound, and we will need another independence lemma.
We will overload the word ``independent'' for the analogous definition for reachability preservers:
\begin{definition} [Independence]
For a graph $G = (V, E)$ and set of demand pairs $P \subseteq V \times V$, we say that $P$ is \emph{independent} (in the context of reachability preservers) if for all $(s, t) \in P$ there is a unique path $\pi(s, t)$, and these paths are pairwise edge-disjoint.
\end{definition}

\begin{lemma} [Independence Lemma for Reachability Preservers]
For any positive integers $n, p$, there exists an $n$-node graph $G$ and a set of $|P|=p$ \textbf{independent} demand pairs such that any reachability preserver of $G, P$ has exactly $\rpp(n, p)$ edges.
\label{lem:rpindep}
\end{lemma}
\begin{proof}
The proof is somewhat analogous to Lemma \ref{lem:dpindep}, but it requires an additional technical ingredient.
The reason for the change in proof is essentially that for distance preservers we can assume that demand pairs have unique shortest paths, but for reachability preservers we cannot immediately make the analogous assumption that each demand pair has a unique path.

Let $G, P$ be a (not necessarily independent) instance realizing $\rpp(n, p)$.
First, we will use a helpful reduction from \cite{AB18}, allowing us to assume that $G$ is a DAG.
If not, we may consider each strongly connected component, add an in- and out-BFS tree from an arbitrary node to preserve reachability among all node pairs in that component, and then contract the component into a single super-node.
The resulting contracted graph is a DAG, and it suffices to build a reachability preserver on this graph.\footnote{The contraction step costs $O(n)$ edges, which may be safely ignored since we already have $\rpp(n, p) = \Omega(n)$, e.g.\ by considering a path on input for which a reachability preserver must keep $n-1$ edges.}

Next, let us introduce some terminology.
We will say that a demand pair $(s, t) \in P$ \emph{requires} an edge $e$ if every $s \leadsto t$ path includes $e$.
We will say that $(s, t)$ \emph{uniquely requires} $e$ if it requires $e$, and there is no other demand pair that also requires $e$.
We may assume without loss of generality that every edge in $G$ is required by at least one demand pair (or else that edge may be removed from $G$).
Thus $G$ itself is the unique reachability preserver of $G, P$, so it has $\rpp(n, p)$ edges.
We then further modify $G, P$ by the following steps:
\begin{itemize}
\item For each demand pair $(s, t)$, considered in arbitrary order, choose any $s \leadsto t$ path $\pi(s, t)$.  Then:
\begin{itemize}
\item If $(s, t)$ does not uniquely require any edge in $\pi(s, t)$, delete $(s, t)$ from $P$.

\item Otherwise, let $(u, v) \in \pi(s, t)$ be the first edge uniquely required by $(s, t)$.
Replace the demand pair $(s, t)$ with $(u, t)$, and replace $\pi(s, t)$ with its $u \leadsto t$ suffix.
\end{itemize}
We note that every edge in $G$ is still required by at least one demand pair, and in the end every remaining demand pair $(s, t)$ uniquely requires the \emph{first} edge in $\pi(s, t)$.

\item Next, repeat the following until no longer possible.
Find a demand pair $(s, t) \in P$ and a contiguous $3$-node subpath $(x, y, z) \subseteq \pi(s, t)$, such that the demand pair $(s, t)$ uniquely requires $(x, y)$ but it does not uniquely require $(y, z)$.
If there are several possible choices of $\{(s, t), (x, y, z)\}$, then we will specifically need to consider one in which the node $y$ comes as early as possible in the topological ordering of nodes in $G$.
(There may still be several possible choices using this same minimal node $y$, in which case we can choose among these arbitrarily.)
We then delete $(x, y)$ from $G$, and add $(x, z)$ to $G$ as a new edge We modify the path $\pi(s, t)$ by replacing its subpath $(x, y, z)$ with the single edge $(x, z)$. After this operation, the number of edges in $G$  stays the same.

For correctness, we now need to argue that after this change, it is still the case that every edge in $G$ is required by at least one demand pair.
We have:
\begin{itemize}
\item The new edge $(x, z)$ is uniquely required by $(s, t)$. This follows from the fact that $(x, y)$ was required by $(s, t)$, and so when $(x, y)$ is deleted (but before $(x, z)$ is added), there is no $s \leadsto t$ path.
When $(x, z)$ is added there is an $s \leadsto t$ path again, which implies that every $s \leadsto t$ path uses the edge $(x, z)$.

\item We also claim that the edge $(y, z)$ is still required by at least one demand pair.
We argue this as follows. Before our modification of $G$, we know that the demand pairs $(s,t)$ and $(s',t')$ each uniquely requires all edges in their prefixes $s\leadsto y$ and $s'\leadsto y$, respectively. This means that $s'$ still cannot reach $x$, and thus $(s',t')$ cannot use the new edge $(x,z)$ and still requires the edge $(y,z)$.

%
\end{itemize}
Each time we repeat this step, the sum of lengths of the  paths $\{\pi(s, t)\}_{(s, t) \in P}$ decreases by $1$.
Therefore, we halt after finitely many steps.
\end{itemize}
To summarize, once this process halts, every demand pair $(s, t)$ has the property that $\pi(s, t)$ uniquely requires its first edge, and moreover for any two consecutive edges $(x, y), (y, z) \in \pi(s, t)$, if $(s, t)$ uniquely requires the first edge $(x, y)$, then it also uniquely requires the second edge $(y, z)$.
Together, these properties imply that each demand pair $(s, t)$ uniquely requires every edge on its path $\pi(s, t)$, which implies independence. 
\end{proof}

\begin{lemma}
$\rpp(n, p) \le O(\beta(n, p, \infty))$.
\end{lemma}
\begin{proof}
By our independence lemma, there is an $n$-node graph $G = (V, E)$ and set of $|P|=p$ independent demand pairs for which any reachability preserver has at least $\rpp(n, p)$ edges.
We can naturally associate $G,P$ to a path system $S = (V, \Pi)$ by associating each demand $(s, t) \in P$ to the unique $s \leadsto t$ path in $G$.
We thus have $\|S\|=\rpp(n, p) + p$.
Moreover, this path system cannot have bridges, since the paths for demand pairs in $P$ are unique.
Thus we have constructed a $\infty$-bridge-free path system of size $\|S\| \ge \rpp(n, p)$, and the lemma follows. 
\end{proof}

We now turn to the lower bound:
\begin{lemma}
$\rpp(n, p)\geq \Omega\left(\beta(n, p,\infty)\right) $.
\end{lemma}
\begin{proof}
Let $S=(V, \Pi)$ be a $\infty$-bridge-free path system with $n$ nodes, $p$ paths, and size $\|S\| = \beta(n, p, \infty)$.
Let $G = (V, E)$ be the directed graph over the same vertex set, where we put an edge $(u, v) \in E$ iff there is a path in $\Pi$ where the nodes $u, v$ appear consecutively (in that order).
Define demand pairs $P$ to be the set of node pairs $(s, t)$ that are endpoints of the paths in $\Pi$.

Since $S$ is $\infty$-bridge-free, for every demand pair $(s, t) \in P$ there is a unique $s \leadsto t$ path in $G$, and these paths are pairwise edge-disjoint.
Thus, it is necessary and sufficient for a reachability preserver to keep all edges contained in these paths.
The number of edges contained in these paths is exactly $$\|S\| - p = \beta(n, p, \infty) - p.$$
This proves that
$$\beta(n, p, \infty)-p \le \rpp(n, p),$$
and we then notice that always $\beta(n, p, \infty)\geq \Omega(p)$,
and so in fact we have
\begin{align*}
\beta(n, p, \infty) \le O\left(\rpp(n, p)\right). \tag*{\qedhere}
\end{align*}
\end{proof}

The previous two lemmas imply Theorem \ref{thm:rpreduction}.
We now turn to its consequences.
\begin{corollary} \label{cor:binftylb}
For all positive integers $d$, we have $\beta(n, p, \infty) = \Omega\left( n^{\frac{2}{d+1}} p^{\frac{d-1}{d}} \right)$.
\end{corollary}
\begin{proof}
Follows from Theorem \ref{thm:rpreduction} and from plugging in the state-of-the-art lower bounds on $\rpp(n, p)$ from \cite{AB18} (which, in turn, are directly based on the distance preserver lower bounds from \cite{CE06}).
We remark that one does not \emph{really} need Theorem \ref{thm:rpreduction} to prove this lower bound on $\beta(n, p, \infty)$, in the sense that it is straightforward to interpret the reachability preserver lower bound construction from \cite{AB18} directly as a lower bound against $\beta(n, p, \infty)$.
\end{proof}

The following corollary uses the equivalence between $\beta(n, p, \infty)$ and reachability preservers more directly.
It shows that the extremal path systems realizing the lower bound for $\beta(n, p, \infty)$ have some extra structure: they must in fact be \emph{acyclic}.
\begin{corollary} \label{cor:inftyacyclic}
For all $n, p$, there exists an \textbf{acyclic} path system $S$ with $n$ nodes, $p$ paths, bridge girth $\infty$, and size $\|S\| = \Theta(\beta(n, p, \infty))$.
\end{corollary}
\begin{proof}
Let $\abeta(n, p, \infty)$ be the maximum possible size of an $n$-node, $p$-path, \emph{acyclic} path system of bridge girth $\infty$.
Let $\overrightarrow{\rpp}(n, p)$ be the maximum number of edges needed for a reachability preserver of an $n$-node DAG and $p$ demand pairs.
We first notice that, by exactly the same reduction as in Theorem \ref{thm:rpreduction}, we have
$$\abeta(n, p, \infty) = \Theta\left(\overrightarrow{\rpp}(n, p)\right).$$
Next, it is proved in \cite{AB18} that
$$\overrightarrow{\rpp}(n, p) = \Theta\left(\rpp(n, p)\right).$$
That is, they show a reduction from finding reachability preservers in general graphs to DAGs. To briefly summarize this reduction, suppose we are given a graph $G$ and demand pairs $P$, and we wish to construct a reachability preserver.
For each strongly-connected component $C$, choose an arbitrary node $c \in C$ and add two trees in $C$ rooted at $c$; one with edges pointing away from $c$, and one with edges pointing towards $c$.
Thus, reachability is preserved between all pairs of nodes in $C$, and we can contract $C$ into a single node before proceeding.
This reduction costs at most $2n = O(\rpp(n, p))$ edges in total.

Finally, from Theorem \ref{thm:rpreduction} we have
$$\rpp(n, p) = \Theta(\beta(n, p, \infty)).$$
Putting the parts together, we have
$$\abeta(n, p, \infty) = \Theta(\beta(n, p, \infty)),$$
as required.
\end{proof}

\subsection{Path Oracles}

We next prove an \emph{incompressibility theorem} for reachability preservers, much like the one proved previously for distance preservers.
We consider path oracles, which are the natural data structure version of reachability preservers:
\begin{definition} [Path Oracles]
Given a directed graph $G$ and a set of demand pairs $P$, a path oracle is a data structure that, when queried with $(s, t) \in P$, can report an $s \leadsto t$ path in $G$ (or ``no path'' if none exists).

We define $\texttt{PO}(n, p)$ as the smallest integer such that every $n$-node graph and set of $|P|=p$ demand pairs has a path oracle on $\le \texttt{PO}(n, p)$ bits.
\end{definition}

Note that a reachability preserver of $G, P$ on $m$ edges implies a path oracle of $G, P$ on $O(m \log n)$ bits, by simply writing down a description of the reachability preserver.
Consequently, we have
$$\texttt{PO}(n, p) = O(\rpp(n, p) \log n).$$
The following theorem states that we cannot expect much smaller path oracles in general.

\begin{theorem} \label{thm:pathoracle}
$\Omega(\rpp(n, p)) \le \texttt{PO}(n, p) \le O(\rpp(n, p) \log n)$.
\label{thm:rpp_incompress}
\end{theorem}
\begin{proof}
The upper bound follows from the previous discussion.
The lower bound follows follows from an argument identical to that of Theorem \ref{thm:shortestpathoracle}, except we use the independence lemma for reachability preservers (Lemma \ref{lem:rpindep}), and we replace $\dpp$ with $\rpp$, and we replace ``unique shortest path'' with ``unique path.''
\end{proof}

\subsection{Online Reachability Preservers}

Here we discuss the online version of the reachability preserver problem.
There are several ways to reasonably define such online versions; this one is a slight variation of the one introduced \cite{GLQ21} in the context of online directed Steiner forest algorithms.\footnote{More specifically: we allow the adversary to add edges to the graph $A$ throughout the game, whereas \cite{GLQ21} essentially require the adversary to commit to a graph in preprocessing.}
\begin{definition} [Online Reachability Preservers]
The online reachability preserver game is the following two-player game, between a \emph{builder} and an \emph{adversary}:
\begin{itemize}
\item The adversary starts with an $n$-node directed graph $A = (V, \emptyset)$, and the builder starts with an $n$-node directed graph $B = (V, \emptyset)$.
Both graphs are initially empty.
The builder is trying to minimize the final number of edges in $B$, and the adversary is trying to maximize the final number of edges in $B$.

\item Repeat the following for $p$ rounds:
\begin{itemize}
\item (Adversary's Turn) The adversary adds any number of edges to $A$, and then names a pair of nodes $(s, t)$ such that an $s \leadsto t$ path in $A$ exists.

\item (Builder's Turn) The builder must respond by choosing a set of edges that are currently in $A$, and adding those edges to $B$.
Afterwards, we require that an $s \leadsto t$ path must exist in $B$.
\end{itemize}

\item The value of the game is the final number of edges in the graph $B$.
\end{itemize}
We define $\rpp^*(n, p)$ as the min-max value of this game, where the adversary is maximizing and the builder is minimizing (the value of the game), relative to parameters $n, p$.
\end{definition}

In the same way that (offline) reachability preservers are captured by $\beta(n, p, \infty)$, as in Theorem \ref{thm:rpreduction}, we claim that online reachability preservers are captured by $\beta^*$:
\begin{theorem} \label{thm:rpstarreduction}
$\beta^*(n, p, \infty) = \Theta(\rpp^*(n, p))$
\end{theorem}

First we will prove an upper bound for $\beta^*(n,p,k)$:
\begin{lemma}
$\rpp^*(n, p) \ge \Omega\left(\beta^*(n, p, \infty)\right)$.
\end{lemma}
\begin{proof}
The strategy of the adversary works as follows.
First, they think of an ordered path system $S = (V, \Pi)$ with $n$ nodes, $p$ paths, no ordered bridges, and size $\|S\| = \beta^*(n, p, \infty)$.
Let $\pi_i \in \Pi$ denote the $i^{th}$ path in the ordering.
In each round $i$ of the game, the adversary considers $\pi_i$, and adds each consecutive pair of nodes along $\pi_i$ as a new edge in $A$.
Then, they name the endpoints $(s, t)$ of $\pi$ as the pair for this round.
Since $S$ has no ordered bridges, currently $\pi_i$ is the unique simple $s \leadsto t$ path in $A$.
Thus the builder has no choice but to add the $|\pi|-1$ edges corresponding to $\pi$ to $B$.
In total, the builder thus adds $\beta^*(n, p, \infty) - p = \Omega\left( \beta^*(n, p, \infty) \right)$
edges to $B$.
\end{proof}

Next, we prove a matching lower bound:

\begin{lemma}
$\rpp^*(n, p) \le \beta^*(n, p, \infty)$
\end{lemma}
\begin{proof}
In the online reachability preserver game, we will assume only that the builder adds a \emph{minimal} set of edges in each round.
That is, when the builder adds edge set $E_i$ in round $i$, we assume that there is no proper subset $E'_i \subsetneq E_i$ that could have been added instead, and still satisfy the adversary's demand.
We claim that, so long as the builder's choices satisfy this property, they will add $\le \beta^*(n, p, \infty)$ edges to $B$ in total.

Indeed, given a sequence of choices made by a builder and an adversary, let us track an auxiliary ordered path system $S$ as follows.
The vertex set $V$ of $S$ is the same as the vertex set of the graphs $A, B$ in the game.
We next describe the paths of $S$.
By minimality of the builder's choices, their selected edge set $E_i$ may be interpreted by considering a simple path $\pi_i$ between the adversary's demand pair, and setting
$E_i := \pi_i \setminus E(B).$
In round $i$, we add a path $q_i$ to the auxiliary path system $S$, where $q_i$ is the sequence of vertices $v$ for which there is an edge in $E_i$ entering $v$, ordered by appearance in $\pi_i$.
We have $|q_i| = |E_i|$, and therefore $\|S\| = |E(B)|$.
So it only remains to show that $S$ has no ordered bridges, and thus $|E(B)| \le \beta^*(n, p, \infty)$.

Suppose for contradiction that $S$ has an ordered bridge, with path $q_i$ as its river, and nodes $s, t \in q_i$ participating in the bridge.
Since $t \in q_i$, there exists an edge of the form $(u, t) \in E_i$.
However, since nodes $s, t$ participate in the bridge, there exists an $s \leadsto t$ path in $B$ before round $i$, that is, the one corresponding to the arcs of said bridge.
Therefore, the builder could have added $E_i \setminus \{(u, t)\}$ in round $i$.
This contradicts minimality of the builder's choices, and thus $S$ has no ordered bridges.
\end{proof}

Analogous to Corollary \ref{cor:inftyacyclic}, the following corollary implies that, without loss of generality, 
\emph{acyclic} graphs are enough to achieve lower bounds for $\beta^*(n, p, \infty)$: 
\begin{corollary} \label{cor:bstaracyclic}
For all $n, p$, there exists an acyclic ordered path system $S$ with $n$ nodes, $p$ paths, ordered bridge girth $\infty$, and size $\|S\| = \Omega(\beta^*(n, p, \infty))$.
\end{corollary}
\begin{proof}
In the parameter regime where $\beta^*(n, p, \infty) = O(n)$, one can take $S$ to be any ordered path system with one path of length $n$ and the remaining paths of length 1, and so the claim is trivial.
In the following, we assume that $\beta^*(n, p, \infty) \ge cn$ for a sufficiently large constant $c$.

By Theorem \ref{thm:rpstarreduction}, it suffices to prove that in the online reachability preserver game, the adversary has a (near-)optimal strategy in which the underlying graph $G$ is always acyclic.
To show this, imagine the following strategy that the builder could use.
Any time the adversary adds an edge to $G$ that completes a directed cycle $C$, the builder immediately adds all edges in $C$ to their reachability preserver, and then for the rest of the game they treat $C$ as a single contracted supernode.
Since each contraction step costs $|C|$ edges and reduces the number of nodes in $G$ by $|C|-1$, the builder pays only $O(n)$ edges in total for these contraction steps, which is negligible.

We now shift perspective back to the adversary.
Any time the adversary \emph{would} add an edge $(u, v)$ to the graph that completes a directed cycle $C$, they could instead omit $(u, v)$ and contract $C$ into a single supernode in their internal representation of the graph.
By the above analysis, this is without loss of generality, and affects the overall min/max value of the game by at most $O(n)$, which is negligible.
Thus, the adversary never completes a directed cycle in $G$, and the theorem follows.
\end{proof}

\subsection{Shortcut Sets and Exact Hopsets}

Here we show that the lower bounds for reachability and distance preservers can be extended to shortcut sets and exact hopsets, respectively.
The proofs are similar to each other in spirit.

\begin{definition} [Shortcut Sets]
For a directed graph $G = (V, E)$, a $\Delta$-diameter-reducing shortcut set is a set of additional directed edges $H$ such that every edge $(u, v) \in H$ is in the transitive closure of $G$, and
$$\max \limits_{s, t \in V, \text{ exists } s \leadsto t \text{ path in } G} \dist_{G \cup H}(s, t) \le \Delta.$$
We write $\sss(n, p)$ for the smallest integer $\Delta$ such that every $n$-node graph has an $\Delta$-diameter-reducing shortcut set of $|H|=p$ edges.
\end{definition}






\begin{theorem} \label{thm:sslb}
$\sss(n, p) = \Omega\left(\frac{\beta(n, p, \infty)}{p} \right)$
\end{theorem}
\begin{proof}
Let $S = (V, \Pi)$ be an $\infty$-bridge-free path system on $n = |V|$ nodes and $p+1 = |\Pi|$ paths, such that\footnote{This latter equality is intuitive, but for completeness it is formally proved in Lemma \ref{lem:paramadjust} in the appendix.}
$$\|S\| = \Theta\left(\beta(n, p+1, \infty)\right) = \Theta(\beta(n, p, \infty)).$$
The average path length in $S$ is thus
$$\ell = \Theta\left( \frac{\beta(n, p, \infty)}{p+1} \right) = \Theta\left( \frac{\beta(n, p, \infty)}{p} \right).$$
By the Cleaning Lemma (Lemma \ref{lem:cleaning}), we may further assume without loss of generality that all paths in $S$ have length $\Theta(\ell)$.
Let us now associate to $S$ a directed graph $G$ as usual, by including a directed edge $(u, v)$ for each pair of nodes $u, v$ that appear consecutively on any path in $S$ (so $|E(G)| = \Theta(\beta(n, p, \infty))$).
We also define a set $P$ of $|P|=p+1$ demand pairs as the endpoints of the paths in $S$.
Recall that shortcut sets require diameter reduction among \emph{all} node pairs, rather than just a specific set of demand pairs, but nonetheless it will be helpful to focus our analysis on these demand pairs.

Since $S$ has bridge girth $\infty$, for each demand pair $(s, t) \in P$ there is a unique simple $s \leadsto t$ path in $G$; let us denote this path by $\pi(s, t)$.
Since these paths $\{\pi(s, t)\}$ do not have $2$-bridges, we additionally have that for any ordered pair of nodes $(x, y)$, there is at most one demand pair $(s, t)$ with $x <_{\pi(s, t)} y$.
In particular, let $H$ be an arbitrary shortcut set of size $|H|=p$.
Since $|P| = p+1$ but $|H|=p$, there exists a demand pair $(s, t) \in P$ such that there is no $(x, y) \in H$ with $x <_{\pi(s, t)} y$.
It follows that $\pi(s, t)$ remains the unique simple $s \leadsto t$ path in the graph $G \cup H$.
The number of hops in this path is
$$|\pi(s, t)|-1 = \Theta(\ell) = \Omega\left(\frac{\beta(n, p, \infty)}{p} \right),$$
and so the hopset $H$ cannot reduce diameter below this threshold.
\end{proof}

We now give an analogous proof for exact hopsets.

\begin{definition} [Exact Hopsets]
For a directed weighted graph $G = (V, E)$, a $\beta$-hop exact hopset is a set of additional directed weighted edges $H$ such that every edge $(u, v) \in H$ has weight $w(u, v) = \dist_G(u, v)$, and for all node pairs $s, t$, there exists a shortest $s \leadsto t$ path in $G \cup H$ that uses at most $\beta$ edges.

We write $\ehh(n, p)$ for the smallest integer $\beta$ such that every $n$-node graph has a $\beta$-hop exact hopset of $|H| = p$ edges.
\end{definition}

Although we give the following theorem in terms of $\beta^*$, we note that the following proof essentially shows that
$$\ehh(n, p) \ge \Omega\left(\frac{\dpp(n, p)}{p}\right) \geq \Omega\left( \frac{\beta^*(n, p, \infty)}{p} \right)$$
where the latter equality is from Theorem \ref{thm:distpreschain}.

\begin{theorem} \label{thm:exacthopsets}
$\ehh(n, p) = \Omega\left( \frac{\beta^*(n, p, \infty)}{p} \right).$
\end{theorem}
\begin{proof}
Let $S = (V, \Pi)$ be an ordered path system with $n$ nodes, $p+1$ paths, ordered bridge girth $\infty$, and size $\|S\| = \Omega\left(\beta^*(n, p+1, \infty)\right)$.
So the average path length in $S$ is
$$\ell = \Omega\left(\frac{\beta^*(n, p+1, \infty)}{p+1}\right) = \Omega\left(\frac{\beta^*(n, p, \infty)}{p}\right).$$
By the Cleaning Lemma (Lemma \ref{lem:cleaning}), we may assume without loss of generality that all paths in $S$ have length $\Theta(\ell)$.
As in Lemma \ref{lem:dplower}, we may associate $S$ to a directed weighted $n$-node graph $G$ and set of demand pairs $P$ such that there is a unique shortest path for each demand pair, and these shortest paths are pairwise edge-disjoint.

Now let $H$ be an arbitrary exact hopset of $G$ of size $|H| \le p$.
The rest of the proof is identical to Theorem \ref{thm:sslb}.
In particular, since $|P|=p+1$ and $|H|=p$, there exists a demand pair $(s, t) \in P$ such that there is no edge $(x, y) \in H$ where both $x, y$ lie along the unique shortest path $\pi(s, t)$.
Thus $\pi(s, t)$ remains the unique shortest $s \leadsto t$ path in the graph $G \cup H$, and it has
$$\Theta(\ell) = \Omega\left(\frac{\beta^*(n, p, \infty)}{p}\right).$$
edges.
So $H$ cannot reduce the number of hops below this value.
\end{proof}

\subsection{Approximate Distance Preservers and Evidence for Conjecture \ref{conj:beta_log} \label{sec:adp}}

Next, we present evidence in favor of the first part of Conjecture \ref{conj:beta_log}.
We do so by considering approximate distance preservers (also sometimes called pairwise spanners):

\begin{definition} [$\alpha$-Approximate Distance Preservers] \label{def:apdps}
Let $G=(V, E, w)$ be a directed weighted graph, $P \subseteq V \times V$ a set of demand pairs, and $\alpha\geq 1$ a parameter.
An $\alpha$-approximate distance preserver is a subgraph $H \subseteq G$ in which, for all $(s, t) \in P$, we have $\dist_H(s, t) \leq \alpha \cdot \dist_G(s,t)$.

We define $\apdpp(n, p, \alpha)$ as the least integer such that every $n$-node graph and set of $|P|=p$ demand pairs has an $\alpha$-approximate distance preserver on $\le \apdpp(n,p, \alpha)$ edges.
\end{definition}

\begin{lemma} \label{lem:adplower}
$\beta(n, p, k) \le O(\apdpp(n, p, 2k/\ell))$, where $\ell=\beta(n,p,k)/p$ is the average path length in a system realizing $\beta(n,p,k)$.
\end{lemma}
\begin{proof}
Let $S = (V, \Pi)$ be a path system realizing $\beta(n, p, k)$.
By Lemma~\ref{lem:cleaning}, we may assume that all path lengths in $S$ are at most
$$\ell' < \frac{\beta(n,p,k)}{2p}.$$
Interpret the endpoints of the paths in $\Pi$ as demand pairs, and consider the unweighted directed graph $G$ that contains an edge $(u, v)$ iff there is a path that uses the nodes $u, v$ consecutively.
The size of $G$ is
$$|E(G)| = \Theta(\|S\|) = \Theta\left( \beta(n, p, k) \right).$$
Moreover, we claim that $G$ is the unique approximate distance preserver of $G, P$ with error parameter $\alpha = 2k/\ell$.
To see this, let $H \subsetneq G$ be a subgraph that does not contain an edge $(u, v)$ from $G$, and let $\pi \in \Pi$ be a path with endpoints $s, t$ that uses $u, v$ consecutively.
The distance from $s$ to $t$ in $H$ must be at least $k$.
Hence, the distance increases by a factor of
$$\frac{k}{\ell'} > \frac{2pk}{\beta(n, p, k)} = \frac{2k}{\ell}.$$
It follows that $H$ is not a $2k/\ell$ approximate distance preserver of $G$, and the lemma follows.
\end{proof}

We note that this lemma is \emph{not} interesting for constant $k$, since we will clearly have $\ell = \Omega(k)$ in this regime and so the approximation factor is meaningless.
Rather, the most interesting regime for this lemma is when 
$k\geq\omega(1)$ but $k \ne \infty$.
This lemma \emph{might} imply new bounds in this regime.
In particular, if Conjecture~\ref{conj:beta_log} is false, then choosing $k=\log n$ implies new lower bounds against $\apdpp$ (in the regime where $\ell = \Theta(\log n)$ as well).
Or, stated in the contrapositive, and using the fact that $\rpp(n, p) = \Theta(\beta(n, p, \infty))$ from Theorem \ref{thm:rpreduction}:
\begin{theorem}\label{eq_conj}
    If $\apdpp(n,p,\alpha)=O(\rpp(n,p))$ for all fixed $\alpha > 1$ then Conjecture~\ref{conj:beta_log} is true.
\end{theorem}

We think this premise is plausible.
In support, we note that results in \cite{KP22} could be interpreted as the analogous statement for \emph{undirected} weighted distance preservers, and that results in \cite{BW23} imply that for directed hopsets, the state-of-the-art bounds for error $\alpha=\infty$ (reachability) and for $(1+\eps)$ essentially match.

%% file: integralitygaps.tex
\section{Flow-Cut and Directed Steiner Forest Integrality Gaps}
\label{sec:int_gap}

In this section, we give polynomial integrality gap lower bounds for the standard LP relaxations of the directed multicut, directed sparsest cut, and directed Steiner forest problems as a function of $\beta(n, p, \infty)$. The duals of the LP relaxations of directed multicut and directed sparsest cut correspond to the well-studied maximum multicommodity flow  and maximum concurrent flow problems, so the integrality gaps for these LPs correspond to \textit{flow-cut gaps}. We also give new integrality gap lower bounds for the standard LP relaxation of directed Steiner forest. This LP relaxation can be interpreted as a generalization of minimum-cost flow to multiple demand pairs.

\subsection{Directed Edge Multicut and Statement of First Result}

In the directed multicut problem, we are given a directed graph $G = (V, E)$ and a set of $p$ demand pairs $P = \{(s_i, t_i) \in V\times V \mid i \in [1, p]\}$, and the objective is to find a minimum subset of $E$ whose removal separates all pairs of vertices $(s_i, t_i)$ in $P$. 
It will be helpful to phrase this as an integer program, as follows:
\begin{itemize}
\item For each edge $e \in E$, let $x_e \in \{0, 1\}$ be an indicator variable that takes value $1$ if $e$ is in the multicut solution and else $0$.
\item For each demand pair $(s_i, t_i) \in P$, denote by $\Pi_i$ the set of directed paths from $s_i$ to $t_i$ in $G$.
\item Then the multicut problem is equivalent to minimizing $\sum_e x_e$, subject to the constraint $\sum_{e \in \pi} x_e \geq 1$ for all $i \in [1, p]$ and $\pi \in \Pi_i$.
\end{itemize}
We will write $\mcut(G, P)$ for the value of the directed multicut problem on inputs $G, P$.
A natural LP relaxation of directed multicut is to let $x_e \geq 0$, so that cuts on edges can be fractional.
We will write the fractional value as $\fmcut(G, P)$, and we state the LP formally as: 
\begin{center}
\begin{tabular}{|c|c|}
\hline 
$\begin{array}{crccc}
(\boldsymbol{P}_{LP}) & \fmcut\\
\min & \sum_{e\in E} x_e \\
\text{s.t.} & \sum_{e\in \pi} x_e & \ge & 1 & \forall i\in [p],\pi\in\Pi_i\\
 & x_e & \ge & 0 & \forall  e\in E
\end{array}$  & $\begin{array}{crccc}
(\boldsymbol{D}_{LP}) & \mmflow\\
\max & \sum_{i\in [p]}\sum_{\pi\in \Pi_i} f_{\pi}\\
\text{s.t.} & \sum_{\pi:e\in\pi}f_{\pi} & \le & 1 & \forall e\in E\\
 & f_{\pi} & \ge & 0 &  \forall i\in [p],\pi\in \Pi_i 
\end{array}$\tabularnewline
\hline 
\end{tabular}
\par\end{center}

The dual program $\boldsymbol{D}_{LP}$ of this LP relaxation is equivalent to the maximum multicommodity flow problem.
In this problem, we get a directed graph $G$ and set of demand pairs $P$ on input, and for each demand pair $(s_i, t_i) \in P$ we choose a flow $f_i$ that has $s_i$ as its source, $t_i$ as its sink, and which is conserved at all other nodes.
Among the $p$ flows $\{f_1, \dots, f_p\}$, their \emph{total} must respect the edge capacity constraints; that is, the sum of flows on each edge must be $\le 1$.
The \emph{value} of flow $f_i$ is the amount of flow created at $s_i$ and destroyed at $t_i$.
Subject to these constraints, the goal is to maximize the sum of flow values.\footnote{The dual LP given here works by choosing a scalar for each $s_i \leadsto t_i$ \emph{path}, which implicitly defines flow values on the edges by the sum of the scalars on paths that contain that edge.}
We write $\mmflow(G, P)$ for this maximized sum of flow values.

For any $G, P$, we have:
$$\mcut(G, P) \ge \fmcut(G, P) = \mmflow(G, P).$$
The min-cut/max-flow theorem states that we have equality in the special case where $|P|=1$, but in general the functions can be quite different.
The largest possible gap between them is the \emph{flow-cut gap}, captured by the following function:
\begin{definition} [$\fcg$]
The function $\fcg(n)$ is the least integer $k$ such that, for every $n$-node directed graph $G$ and set of demand pairs $P$ (of any size), we have
$$\mcut(G, P) \le k \cdot \mmflow(G, P).$$
\end{definition}

Since $\fmcut(G, P) = \mmflow(G, P)$, we may equivalently interpret this function as the integrality gap of \mcut.
We refer to \cite{chuzhoy2009polynomial} for further discussion and proofs of all of the claims in the previous discussion.
Our goal is to prove:


\begin{theorem}
For all $n$, the flow-cut gap satisfies
$$\fcg(\beta(n, n, \infty)) = \Omega\left(\dfrac{\beta(n, n, \infty)}{n \log n}\right).$$
\label{thm:fcg}
\end{theorem}
As in \cite{chuzhoy2009polynomial}, we will actually prove something slightly stronger: this gap holds even when the multicut solution only needs to disconnect a $(1 - \epsilon)$-fraction of the demand pairs, for some $\epsilon > 0$. 
Our proof of Theorem \ref{thm:fcg} will closely follow the previous directed multicut integrality gap lower bound construction of \cite{chuzhoy2009polynomial}.

\subsection{Directed Vertex Multicut}

We will begin by lower bounding an integrality gap for the directed \emph{vertex} multicut problem.
This problem is defined as \mcut, except instead of deleting edges to disconnect demand pairs, we delete non-terminal \emph{vertices}.\footnote{A terminal vertex is one that appears as either endpoint of a demand pair.  We assume the input is such that no demand pair $(s, t)$ has an edge going directly from $s$ to $t$, so that the cut exists.}
The least number of nonterminal vertices required to disconnect demand pairs $P$ in a graph $G$ is written $\vmcut(G, P)$.
We will also consider the LP relaxation $\vfmcut$, defined as follows:

\begin{center}
\begin{tabular}{|c|}
\hline 
$\vfmcut$ \\
$\begin{array}{crccc}
\min & \sum_{v\in V} x_v \\
\text{s.t.} & \sum_{v \in \pi \cap V} x_v & \ge & 1 & \forall i\in [p],\pi\in\Pi_i\\
 & x_v & \ge & 0 & \forall  v\in V
\end{array}$
\tabularnewline
\hline
\end{tabular}
\end{center}

As noted in \cite{chuzhoy2009polynomial}, the integrality gap of the vertex multicut problem is at most the integrality gap of the (edge) multicut problem.
In particular, we will use the following lemma:
\begin{lemma} 
\label{lem:vmcut_reduction}
[\cite{chuzhoy2009polynomial}]
Suppose we can construct an $n$-node graph $G$ and set of demand pairs $P$ with
$$\dfrac{\vmcut(G, P)}{\vfmcut(G, P)} \ge k.$$
Then $\fcg(n) \ge k$.
\end{lemma}


Thus, our lower bound on $\fcg$ will work by lower bounding the integrality gap for $\vmcut$.
The authors of \cite{chuzhoy2009polynomial} used the same strategy.
Specifically, they constructed a directed graph $G$ and a set of demand pairs $P$ such that: 
\begin{enumerate}
\item\label{property:pathlength} For every pair $(s, t) \in P$,  the shortest $s \leadsto t$ path in $G$ is of length at least $L = \widetilde{\Omega}(n^{1/7})$.
    \item\label{property:cut} A vertex cut of size $\Omega(n)$ is required to separate the set of pairs $P$.
\end{enumerate}
By property~\ref{property:pathlength}, if we assign a $1/L$-fraction cut to each vertex in $G$, then we obtain a valid fractional vertex cut of $G$ of size $\widetilde{O}(n^{6/7})$. Together with property~\ref{property:cut} this implies an integrality gap of $\widetilde{\Omega}(n^{1/7})$. The authors of \cite{chuzhoy2009polynomial} construct the graph $G$ by combining a certain random graph $H$ and a graph $\mathcal{L}$ called a labeling scheme, which carries a set of demand pairs with long, unique paths.
Our proof of Theorem \ref{thm:fcg} largely follows their construction and analysis, but we generalize some piece of the argument to show that the particular properties of the graph $\mathcal{L}$ are not really needed, and instead we can use an \emph{arbitrary} bridge-free path system with $n$ nodes and $n$ paths that realizes the bound $\beta(n, n, \infty)$ bound. 


\subsection{Multicut Flow-Cut Gap Construction}
\label{sec:multicut_const}

Let $n$ be a parameter.
Let $S = (V, \Pi)$ be a path system with $n$ nodes, $p=n$ paths, bridge girth $\infty$, and size $\Omega(\beta(n, n, \infty))$.
By the cleaning lemma (Lemma \ref{lem:cleaning}) we may assume without loss of generality that $S$ is approximately degree-regular and length-regular. Specifically, we may assume every node has degree $\Theta(d)$, where $d$ is the average degree of $S$.
Since $p\ell = nd$ and $p=n$, $d$ is also the average path length in $S$; we thus also have that all paths have length $\Theta(d)$.
Note that $d = \omega(1)$, by applying known lower bounds on $\beta(n, n, \infty)$.

Our next step is to use $S$ to build a corresponding graph $G_S$ that will inherit the long, unique paths property of $S$.
We will use graph $G_S$ as a black box version of the labeling scheme $\mathcal{L}$ in \cite{chuzhoy2009polynomial}.

\begin{lemma} 
The set of paths $\Pi$ can be partitioned into $d$ nonempty sets $\Pi_1, \Pi_2, \dots, \Pi_{d}$ such that $|\Pi_i| = n/d$ for $i \in [1, d]$ and at least $\Omega(d)$  sets $\Pi_i$ satisfy the property that\footnote{We assume for convenience that $n$ is divisible by $d$.  If not, some part sizes may be rounded up or down while only affecting the following argument by lower-order terms.} 
\[\left|\bigcup_{\pi \in \Pi_i} \pi \right| = \Omega(n). \]
\label{lem:large_part}
\end{lemma}
\begin{proof}

Uniformly at random, partition $\Pi$ into $d$ sets $\Pi_i$ for $i \in [1, d]$, each containing $n/d$ paths.
Fix a node $v \in V$ and an index $i \in [1, d]$.
Note that the probability a node $v$ belongs to a randomly chosen path in $\Pi$ is at least $\Omega(d/n)$, since $S$ is approximately degree-regular.
Then $v$ belongs to a path in $\Pi_i$ with constant probability.
It follows that
$\mathbb{E}\left[ |\Pi_i| \right] = \Omega(n).$
Since the maximum possible size is $|\Pi_i| \le n$, by Markov's inequality we have $|\Pi_i| = \Omega(n)$ with constant probability.
Thus, over all choices of $i \in [d]$, the expected number of parts $\Pi_i$ satisfying $|\Pi_i| = \Omega(n)$ is $\Omega(d)$.
So there exists a possible partition in which $\Omega(d)$ parts all satisfy $|\Pi_i| = \Omega(n)$.
%
\end{proof}

\paragraph{Construction of $G_S$.}

We build our graph $G_S = (V_S \cup V'_S, E_S)$ corresponding to path system $S$ as follows.
Let $V_S$ denote the set of nonterminal nodes of $G_S$, and let $V_S := V$.  Let $V'_S$ denote the set of terminal nodes $s_{i, j}, t_{i, j}$, where $i \in [1, d]$ and $j \in [1, n / d]$.  Add all terminal pairs $(s_{i, j}, t_{i, j})$ to our set of demand pairs $P_S$. 
Let $\Pi_1, \Pi_2, \dots, \Pi_{d}$ be the partition of $\Pi$ as specified in Lemma \ref{lem:large_part}.  For every path $\pi \in \Pi_i$, if edge $e$ is in the transitive closure of path $\pi$, then add $e$ to edge set $E_i^S$. 
Here, we say that $e$ is in the transitive closure of $\pi$ if $\pi$ contains $e$ as a (possibly noncontiguous) subsequence. 
Additionally, for $i \in [1, d]$, order the $n/d$ paths in $\Pi_i$ arbitrarily, and let $\pi_{i, j}$ denote the $j$th path in $\Pi_i$, for $i \in [1, d]$ and $j \in [1, n/d]$. 
Add to edge set $E_i^S$ an edge from $s_{i, j}$ to the first vertex of $\pi_{i, j}$. Likewise, add to edge set $E_i^S$ an edge from the last vertex of $\pi_{i, j}$ to  $t_{i, j}$. 
We define the edge set $E_S$ of $G_S$ to be $E_S = \cup_i E_i^S$.  We refer to the edges of $G_S$ belonging to $E_i^S$ as edges of type $i$. We say that an $s_{i, j} \leadsto t_{i, j}$ path is canonical if it is composed exclusively of edges of type $i$. The following properties of $G_S$ follow immediately from  our choice of path system $S$.

\begin{observation}
\label{obs:G_S_props}
    Graph $G_S$ has the following properties:
    \begin{enumerate}
        \item Every $s_{i, j} \leadsto t_{i, j}$ path in $G_S$ is a canonical path, for all $(s_{i, j}, t_{i, j}) \in P_S$.
        \item There exists an $s_{i, j} \leadsto t_{i, j}$ path of length at least $\Omega(d)$ in $G_S$, for all $(s_{i, j}, t_{i, j}) \in P_S$.
    \end{enumerate}
\end{observation}

Recall that graph $G_S$ is intended to replace the \textit{labeling scheme} graph $\mathcal{L}$ in the argument of \cite{chuzhoy2009polynomial}. 
Our final graph $G$ will be the product of graph $G_S$ and an additional graph $H$ that we construct next. Roughly, the properties of $G_S$ summarized above will ensure that our final graph $G$ has long shortest paths between all demand pairs, and therefore has a small fractional vertex multicut. The graph $H$ will be a well-connected random graph, which will roughly ensure that our final graph $G$ has a large minimum vertex multicut of its demand pairs. These two properties together will ensure that $G$ has a large integrality gap between its fractional and integral vertex multicut. The graph $H$ will be essentially identical to the graph $H$ given  in \cite{chuzhoy2009polynomial}, but with different construction parameters.

\paragraph{Construction of $H$ (c.f. Section 3.1.2 of \cite{chuzhoy2009polynomial}).}
We build $H = (V_H \cup V_H', E_H)$ as follows. Let $V_H$ denote the set of nonterminal nodes of $H$, and let  $V_H := \{v_1, \dots, v_d\}$. Additionally, graph $H$ will have $d$ distinct pairs of terminal nodes $P_H := \{ (s_i, t_i) \mid i \in [1, d]\}$ as demand pairs, with $V_H'$ denoting the set of all terminal nodes.  Graph $H$ will be defined as the union of graphs $H_i = (V \cup \{s_i, t_i\}, E_i^H)$ for $i \in [1, d]$. 

We construct graph $H_i$ for $i \in [1, d]$ as follows. Graph $H_i$ has terminal nodes $s_i, t_i$ and will contain $d' := d /(2 \log d)$ layers each containing $\log d$ nonterminal nodes. We denote the layers as $L^1_i, L^2_i, \dots, L^{d'}_i \subseteq V_H$ and construct them sequentially as follows. To construct the $j$th layer $L^j_i$ for $j \in [1, d']$, select uniformly at random $\log d$ distinct nodes from $V_H \setminus (L^1_i \cup L^2_i \cup \dots \cup L^{j-1}_i)$. Note that by construction, $| \cup_j L^j_i | \leq d / 2$. Now  define the set of edges $E_i^H$ of $H_i$ as follows. Add an edge from $s_i$ to every vertex in $L^1_i$. Likewise, add an edge from every vertex in $L^{d'}_i$ to $t_i$. Finally, add an edge from every vertex in layer $L^j_i$ to every vertex in layer $L^{j+1}_i$ for $j \in [1, d' - 1]$. This concludes the construction of $H_i$. We define $H$ to be $\cup_{i} H_i$. Likewise, we define $E_H = \cup_{i} E_i^H$ to be the edge set of $H$. We refer to the edges in $H$ belonging to $E_i^H$ as edges of type $i$. We say that an $s_i \leadsto t_i$ path in $H$ is a canonical path if it contains only edges of type $i$. The following properties of graph $H$ will be used in the analysis. 
\begin{observation}
    Graph $H$ has the following properties:
    \label{obs:H_props}
\begin{enumerate}
    \item Every canonical $s_i \leadsto t_i$ path in $H$ contains at least $d'$ nonterminal nodes. 
    \item $\Omega(d)$ nonterminal nodes must be removed from $H$ to disconnect a constant fraction of demand pairs  in $P_H$. 
\end{enumerate}
\end{observation}
Property~\ref{property:pathlength} is immediate from construction, and property~\ref{property:cut} holds for $H$ with high probability, as proven in \cite{chuzhoy2009polynomial} and proven in a slightly stronger form below.
\begin{lemma}[c.f. Lemma $3.1$ of \cite{chuzhoy2009polynomial}]
\label{lem:big_multicut} Fix an $\varepsilon > 0$. For sufficiently large $d$, the following holds with probability $\ge 1-2^{-d}$:

There does not exist a set $\mathcal{S} \subseteq V_H$ of nonterminal nodes of $H$, of size $|\mathcal{S}| \leq d/16$, such that for more than $(1-\varepsilon)d$ distinct indices $i \in [1, d]$, removing $\mathcal{S}$ disconnects the demand pair $(s_i, t_i) \in P_H$ in graph $H_i$.
\end{lemma}
\begin{proof}
We defer the proof of this lemma to Appendix \ref{app:big_multicut} since it follows from the same argument as Lemma 3.1 of \cite{chuzhoy2009polynomial}.
\end{proof}

Note that $d$ is sufficiently large, since we assumed that $d  = \omega(1)$. Then we may assume that property~\ref{property:cut} of $H$, as formalized in Lemma \ref{lem:big_multicut}, holds for the specific graph $H$ we will use in our construction of $G$. 
Note that while all canonical $s_i \leadsto t_i$ paths in $H$ are of length at least $d'$ by property~\ref{property:pathlength} of $H$, this is not true in general for all $s_i \leadsto t_i$ paths in $H$. We will see that composing graph $H$ with graph $G_S$ will allow us to ensure that the shortest paths between all demand pairs are of length  $d'$ in the final graph $G$. 

\paragraph{Construction of $G$.}
 We now construct our final graph $G = (V_G \cup V_G', E)$ by composing $G_S$ and $H$ in a natural way. Let graphs $G_S = (V_S \cup V'_S, E_S)$ and $H = (V_H \cup V'_H, E_H)$ be as defined previously. Let $V_G := V_H \times V_S$ be the set of nonterminal vertices of $G$.  We let $V_G'$ denote the set of terminal vertices of $G$ and let $V'_G := V'_S$. Likewise, we let $P$ denote the set of $n$ demand pairs in $G$ and let $P := P_S$. 

 The set of edges $E$ of $G$ are defined as follows.  Let $(x, x')$ and  $(y, y')$ be nonterminal vertices in $V_G$.
 \begin{itemize}
 \item We add edge $((x, x'), (y, y'))$ to set $E_i$ if $(x, y) \in E_i^H$ and $(x', y') \in E_i^S$. 
 
 \item For $s_{i, j} \in V_G'$ and $(x, x') \in V_G$ we add edge  $(s_{i, j}, (x, x'))$ to $E_i$ if $(s_i, x) \in E_i^H$ and $(s_{i, j}, x') \in E_i^S$.
 
 \item Finally, we add edge $((x, x'), t_{i, j})$ to $E_{i}$ if $(x, t_i) \in E_i^H$ and $(x', t_{i, j}) \in E_i^S$. We let $E = \cup_{i}E_i$.
 \end{itemize}
 As with $G_S$ and $H$, we refer to the edges in $E_i$ as edges of type $i$, and we say an $s_{i, j} \leadsto t_{i, j}$ path in $G$ is canonical if its composed of only type $i$ edges. 
This completes the construction of $G$.

The number of nonterminal vertices in graph $G$ is $N := |V_H| \cdot |V_S| = dn = \beta(n, p, \infty)$. Then the value of $d'$ is 
\[
d' = \frac{d}{2 \log d} \geq \frac{N}{2n\log N}  = \Omega\left(  
\frac{\beta(n, p, \infty)}{n \log n}\right).
\]

In the following section, we will show that the  multicut integrality gap of $G$ is at least $d'$, which will complete our lower bound. 

\subsection{Integrality Gap Analysis of $G$}
We now analyze the gap between the fractional cost of a vertex multicut of $G, P$ and the cost of an integral vertex multicut of $G, P$. Our analysis largely follows that of \cite{chuzhoy2009polynomial}.

\paragraph{Fractional solution.} 
Assign a fractional cut of $1/d'$ to each nonterminal node in $G$. We will show  that this is a valid fractional vertex multicut of $G$ of size $N/d'$. It is clear that the size of this prospective cut is $N/d'$ as desired. The validity of this fractional cut will be immediate from the following claim.

\begin{claim}
For all $(s_{i, j}, t_{i, j}) \in P$, any $s_{i, j} \leadsto t_{i, j}$ path in $G$ contains at least $d'$ nonterminal nodes.
\label{claim:long_paths_in_G}
\end{claim}
\begin{proof}
Observe that every $s_{i, j} \leadsto t_{i, j}$ path $\pi$ for $(s_{i, j}, t_{i, j}) \in P$ corresponds to a $s_i \leadsto t_i$ path in $H$ and a  $s_{i, j} \leadsto t_{i, j}$ path in $G_S$. Then every $s_{i, j} \leadsto t_{i, j}$ path in $G$ is canonical, since every  $s_{i, j} \leadsto t_{i, j}$ path in $G_S$ is canonical, by property 1 of Observation \ref{obs:G_S_props}. If a $s_{i, j} \leadsto t_{i, j}$ path in $G$ is canonical, then the corresponding $s_i \leadsto t_i$ path in $H$ is canonical by the definition of $G$, and so by property 1 of Observation \ref{obs:H_props}, $\pi$  has at least $d'$ nonterminal nodes.
\end{proof}

\paragraph{Integral solution.}  We will show that $\Omega(N)$ vertices must be removed from $G$ to separate all demand pairs in $P$. Before proving this, we must first introduce some notation and prove some intermediate results. For the remainder of this section, fix $\mathcal{S} \subseteq V_G$ to be any subset of non-terminal vertices of $G$ with $|\mathcal{S}| \leq \lambda N$, where $\lambda > 0$ is a sufficiently small constant to be specified later. 
For any set $U \subseteq V_G$,  we let $U^H$ denote the preimage of $U$ in $H$. Namely, 
$$U^H := \{v \in V_H \mid (v, v') \in U \}.$$ 
For $v \in V_S$, we define the set $\mathcal{S}_v \subseteq V_G$ as
$$\mathcal{S}_v := \mathcal{S} \cap (V_H \times \{v\}).$$
Now for $i \in [1, d]$ and $v \in V_S$, we say that the pair $(i, v) \in [1, d] \times V_S$ is $H$-good if $\mathcal{S}_v^H$ does not disconnect demand pair $(s_i, t_i) \in P_H$ in graph $H_i$.

\begin{claim}
Fix an $\varepsilon > 0$. There are at least $(1-\varepsilon)(1-16\lambda)N$ pairs $(i, v) \in [1, d] \times V_S$ that are $H$-good. 
\end{claim}
\begin{proof}
Note that there are  $N = dn$ pairs $(i, v) \in  [1, d] \times V_S$. For any $v \in V_S$, we know that if $|\mathcal{S}_v^H| \leq d/16$, then there are at least $(1-\varepsilon)d$ pairs $(i, v)$ that are $H$-good  by Lemma \ref{lem:big_multicut}. Moreover, since $|\mathcal{S}| \leq \lambda N$, it follows that there are fewer than $16\lambda n$ vertices $v \in V_S$ such that  $|\mathcal{S}_v^H| > d/16$. Then the number of pairs  $(i, v) \in [1, d] \times V_S$ that are $H$-good is at least
\[(1-\varepsilon)(1-16\lambda )N. \qedhere\] 
\end{proof}

We say that a pair $(i, v) \in [1, d] \times V_S$ is \textit{$S$-good} if there is a path $\pi$ in $\Pi_i$ such that $v \in \pi$.

\begin{claim}
There are  $\Omega(N)$ pairs $(i, v) \in [1, d] \times V_S$ that are $S$-good. 
\label{clm:s_good}
\end{claim}
\begin{proof}
By Lemma \ref{lem:large_part} at least $\Omega(d)$ of the sets $\Pi_i$, $i \in [1, d]$,  satisfy $|\cup_{\pi \in \Pi_i} \pi| \geq \Omega(n)$.  Then the number of pairs $(i, v) \in  [1, d] \times V_S$ such that $(i, v)$ is $S$-good is  $\Omega(N)$.   
\end{proof}

Let $\lambda_1 > 0$ be a constant such that for sufficiently large $N$ at least $\lambda_1 N$ pairs $(i, v)$ in $[1, d] \times V_S$  are $S$-good. By Claim \ref{clm:s_good}, such a $\lambda_1$ must exist. We choose $\lambda$ to be  
\[
\lambda := \lambda_1 / 32.
\]
We say that a pair $(i, v) \in [1, d] \times V_S$ is \textit{$G$-good} if $(i, v)$ is $H$-good and $S$-good. 
Let $J \subseteq [1, d]$  be the set of  all indices $j \in [1, d]$ such that the number of $G$-good pairs in $\{j\} \times V_S$ is at least $\lambda_1^2 / 8  \cdot n$.

\begin{claim}
$|J| \geq \lambda_1^2 / 8 \cdot d$. 
\label{claim:J}
\end{claim}
\begin{proof}
Let $\varepsilon = \lambda_1/2$. Since there are at least $(1-\varepsilon)(1-16\lambda )N$ $H$-good pairs and at least $\lambda_1 N$ $S$-good pairs, by an overlap argument it follows that there are at least
$$
\left((1 - \varepsilon)(1 - 16\lambda) + \lambda_1 - 1\right)N = ((1-\lambda_1/2)^2 + \lambda_1 - 1)N = \lambda_1^2 / 4 \cdot N
$$
$G$-good pairs. Now note that for any $i \in J$, there are at most $n$ $G$-good pairs in $\{i\} \times V_S$. Likewise, for any $i \not \in J$, there are at most $\lambda_1^2/8 \cdot n$ $G$-good pairs in $\{i\} \times V_S$. Then there are at most 
$$n|J| + \lambda_1^2/8 \cdot n(d - |J|)$$ 
$G$-good pairs. We obtain the following inequality:
\[
n|J| + \lambda_1^2/8 \cdot n (d - |J|) \geq \lambda_1^2/4 \cdot N
\]
Solving for $|J|$, we conclude that $|J| \geq \lambda_1^2/8 \cdot d = \Omega(d)$. 
\end{proof}

We need one more claim before we can prove that $\mathcal{S}$ is not a valid multicut.  

\begin{claim}
    For every $i \in J$, there are $\Omega\left(\frac{n}{d}\right)$ paths $\pi \in \Pi_i$ such that at least $d'$ pairs in $\{i\} \times \pi$ are $G$-good. 
    \label{claim:J_i}
\end{claim}
\begin{proof}
Fix an $i \in J$. By the definition of $J$, there are at least $\lambda_1^2/8 \cdot n$ $G$-good pairs in $\{i\} \times V_S$. Now let $J_i \subseteq \Pi_i$ denote the set of all paths $\pi$ in $\Pi_i$ such that at least $d'$ pairs in $\{i\} \times \pi$ are $G$-good. Since $S$ is approximately length-regular by the cleaning lemma, each path $\pi$ in $J_i$ is of length at most $|\pi| \leq cd$ for some constant $c \geq 1$, and therefore $\{i\} \times \pi$ has at most $cd$ $G$-good pairs.  Likewise, each path $\pi  \in \Pi_i \setminus J_i$ has at most $d'$ $G$-good pairs. Recall that $(i, v) \in \{i\} \times V_S$ is $G$-good only if $v \in \pi$ for some $\pi \in \Pi_i$. Then by the above discussion there are at most  
$$cd|J_i| + d'\left(\frac{n}{d} - |J_i|\right)$$  
$G$-good pairs in $\{i\} \times V_S$.  
We obtain the following inequality:  
\[
cd |J_i| + d' \left(\frac{n}{d} - |J_i|\right) \geq \lambda_1^2 / 8 \cdot n 
\]
Using the fact that $d' = d /(2 \log d) \leq  \lambda_1^2 / 16 \cdot d  $ for sufficiently large $d$, we conclude that $$|J_i| \geq  \frac{\lambda_1^2}{16c} \cdot \frac{n}{d} \qedhere$$  
\end{proof}


We will now show that a constant fraction of the demand pairs $P$ remain connected in $G \setminus \mathcal{S}$. 

\begin{lemma}
$\Omega(n)$ demand pairs in $P$ are connected in $G \setminus \mathcal{S}$.  
\label{lem:big_cut_fail}
\end{lemma}
\begin{proof}
Fix an $i \in J$ and a $j \in [1, n/d]$ such that path $\pi_{i, j}$ satisfies the property of Claim \ref{claim:J_i}. (Recall that $\pi_{i, j}$ is the canonical $s_{i, j} \leadsto t_{i, j}$ path in $G_S$.) Now let $u_1, u_2, \dots, u_{d'}$ be a set of $d'$ vertices in $\pi_{i, j}$ (listed in the order in which they appear in $\pi_{i, j}$) such that all pairs in $\{i\} \times \{u_1, \dots, u_{d'}\}$ are $G$-good. 
Recall that by the construction of $G_S$, all edges in the transitive closure of $\pi_{i, j}$ are added to $E_i^S$ in graph $G_S$, so in particular, $(s_{i, j}, u_1), (u_k, t_{i, j}) \in E_i^S$ and  $(u_k, u_{k+1}) \in E_i^S$ for all $k \in [1, d' - 1]$.  

Now for all $k \in [1, d']$, we claim that there exists a vertex in $L^k_i \times \{u_k\} \subseteq V_G$ that survives in $G \setminus \mathcal{S}$, i.e. $(L^k_i \times \{u_k\}) \setminus \mathcal{S} \neq \emptyset$. Recall that $L^k_i$ denotes the vertices in the $k$th layer of $H_i$. We know that $(i, u_k)$ is a $G$-good pair, so it is also an $H$-good pair, which implies that $L^k_i \not \subseteq \mathcal{S}_v^H$.
(If $L^k_i \subseteq \mathcal{S}^H$, then $\mathcal{S}_v^H$ would disconnect $(s_i, t_i)$ in $H_i$, a contradiction.) 
Let $w_k$ denote a vertex in $L^k_i \setminus \mathcal{S}_v^H \in V_H$ for $k \in [1, d']$.  Note that since $w_k \in L_i^k$ for $k \in [1, d']$, it follows from the construction of $H$ that $(w_k, w_{k+1}) \in E_i^H$ for $k \in [1, d' - 1]$.  

Let $x_k = (w_k, u_k) \in V_G$ for $k \in [1, d']$. By the discussion in the last paragraph, $x_k \in V_G$   survives after $\mathcal{S}$ is removed from $G$, i.e. $x_k \in G \setminus \mathcal{S}$ for $k \in [1, d']$.  Furthermore, since $(u_k, u_{k+1}) \in E_i^S$ and $(w_k, w_{k+1}) \in E_i^H$ for $k \in [1, d' - 1]$, by the construction of $G$ it follows that edge $(x_k, x_{k+1}) \in E$ survives in $G$ after $\mathcal{S}$ is removed, for $k \in [1, d' - 1]$. Finally, note that by the construction of $G$, edges $(s_{i, j}, x_1)$ and $(x_{d'}, t_{i, j})$ are in $E$. Consequently, $(s_{i, j}, x_1, x_2, \dots, x_{d'}, t_{i, j})$ is a valid $s_{i, j} \leadsto t_{i, j}$ path in $G \setminus \mathcal{S}$,  so demand pair $(s_{i, j}, t_{i, j}) \in P$ is connected in $G \setminus \mathcal{S}$. Since $|J| = \Omega(d)$ by Claim \ref{claim:J} and for all $i \in J$  there are $\Omega\left(\frac{n}{d}\right)$ paths $\pi$ in $\Pi_i$ satisfying the property of Claim \ref{claim:J_i}, we conclude that  $\Omega(n)$ of the $|P| = n$ demand pairs in $P$ are connected in $G\setminus \mathcal{S}$. 
\end{proof}
We assumed $\mathcal{S}$ was an arbitrary vertex set of size $|\mathcal{S}| \leq N/1600$, so we conclude by Lemma \ref{lem:big_cut_fail} that any vertex multicut of $G, P$ is of size $\Omega(N)$. Then since $G, P$ has a fractional vertex multicut of cost $N/d'$,  we obtain an integrality gap of $\Omega(d')$ for the minimum vertex multicut problem. Additionally, even if the multicut solution only needs to disconnect a constant fraction  $(1-\varepsilon)n$ of the $n$ demand pairs in $P$ for some sufficiently small $\varepsilon > 0$, the size of the vertex multicut of $G, P$ remains $\Omega(N)$ by Lemma \ref{lem:big_cut_fail}.  (We will make use of this fact in the sparsest cut flow-cut gap argument.) 
Theorem \ref{thm:fcg} is immediate from the above discussion and Lemma \ref{lem:vmcut_reduction}.

\subsection{Sparsest Cut Flow-Cut Gap}

In the directed sparsest cut problem, we are given a directed graph $G = (V, E)$ and a set of $p$ demand pairs $P = \{(s_i, t_i) \in V \times V \mid i \in [1, p] \}$, and the objective is to find a subset $E'$ of $E$ that minimizes the ratio $|E'| / |P_{E'}|$, where $P_{E'}$ is the subset of $P$ that is disconnected in graph $G \setminus E'$. It will be helpful to phrase this as an integer program, as follows:
\begin{itemize}
    \item For each edge $e \in E$, let $x_e \in \{0, 1\}$ be an indicator variable that takes value 1 if $e$ is in the solution $E'$.  
    \item For each demand pair $(s_i, t_i) \in P$, denote by $\Pi_i$ the set of directed paths from $s_i$ to $t_i$ in $G$.
    \item For each $i \in [1, p]$, let $h_i \in \{0, 1\}$ be an indicator variable that takes value 1 if source-sink pair $(s_i, t_i)$ is disconnected in $G \setminus E'$, i.e. $(s_i, t_i) \in P_{E'}$. 
    \item Let $D := \sum_{i=1}^k h_i$ be the total number of disconnected pairs $|P_{E'}|$. For $e \in E$, let $x_e' := x_e/D$, and for $i \in [1, p]$, let $h_i' := h_i/D$. 
    \item Then the sparsest cut problem is equivalent to minimizing $\sum_e x_e'$ subject to $\sum_i h_i' \geq 1$ and $\sum_{e \in \pi} x_e' \geq h_i'$ for all $i \in [1, p]$ and $\pi \in \Pi_i$.
\end{itemize}

We will write $\textsc{SCut}(G, P)$ for the value of the directed sparsest cut problem on inputs $G, P$. A natural LP relaxation of directed sparsest cut is to let $x_e \geq 0$, so that cuts on edges can be fractional. We will write  the fractional value as $\widehat{\textsc{SCut}}(G, P)$, and we will state the LP formally as: 
\begin{center}
\begin{tabular}{|c|c|}
\hline 
$\begin{array}{crccc}
(\boldsymbol{P}_{LP}) & \widehat{\textsc{SCut}}\\
\min & \sum_{e\in E} x_e' \\
\text{s.t.} & \sum_{e\in \pi} x_e' & \ge & h_i' & \forall i\in [p],\pi\in\Pi_i\\
& \sum_{i=1}^p h_i' & \geq & 1 & \\
 & x_e', h_i' & \ge & 0 & \forall  e\in E, \forall i \in [p]
\end{array}$  & $\begin{array}{crccc}
(\boldsymbol{D}_{LP}) & \text{\textsc{MCFlow}}\\
\max & \lambda\\
\text{s.t.} & \sum_{\pi \in \Pi_i}f_{\pi} & \ge & \lambda & \forall i \in [p] \\
& \sum_{\pi : e \in \pi}  f_{\pi} & \le & 1 & \forall e\in E \\
 & f_{\pi} & \ge & 0 &  \forall i\in [p],\pi\in \Pi_i 
\end{array}$\tabularnewline
\hline 
\end{tabular}
\par\end{center}

The dual program $\boldsymbol{D}_{LP}$ of this LP relaxation is equivalent to the maximum concurrent flow problem. 
In this problem, we are given a directed graph $G$ and a set of demand pairs $P$ on input, and for each demand pair $(s_i, t_i) \in P$ we choose a flow $f_i$ that has $s_i$ as its source, $t_i$ as its sink, and which is conserved at all other nodes. As with minimum multicut, the sum of the flows on each edge must be $\leq 1$, and the value $|f_i|$ of each flow $f_i$ is the amount of flow created at $s_i$ and destroyed at $t_i$. Subject to these constraints, the goal is to maximize $\lambda = \min_i |f_i|$, the least amount of flow routed from $s_i$ to $t_i$ for any $(s_i, t_i) \in P$. 
We write $\textsc{MCFlow}(G, P)$ to denote this maximized $\lambda$. We define the following function to capture the flow-cut gap between maximum concurrent flow and sparsest cut:
\begin{definition} \label{def:scg}
    The function $\scg(n)$ is the least integer $k$ such that, for every $n$-node directed graph $G$ and set of demand pairs $P$ (of any size), we have
    $$
    \textsc{SCut}(G, P) \leq k \cdot \textsc{MCFlow}(G, P).
    $$
\end{definition}
Since $\widehat{\textsc{SCut}}(G, P) = \textsc{MCFlow}(G, P)$ by LP duality, we may equivalently interpret this function as the integrality gap of $\widehat{\textsc{SCut}}$. We refer to \cite{chuzhoy2009polynomial} for further discussion.  Using our graph $G$  from the proof of Theorem \ref{thm:fcg} and a standard reduction argument from \cite{chuzhoy2009polynomial}, we can prove:


\begin{theorem}
For all $n$, the flow-cut gap between maximum concurrent flow and sparsest cut satisfies
$$
\scg(\beta(n, n, \infty)) = \Omega\left( \frac{\beta(n, n, \infty)}{n \log n} \right).
$$
\label{thm:fcg_sc}
\end{theorem}

We defer the proof of this theorem to Appendix \ref{app:sc} since it is essentially identical to the argument in Section 3.2 of \cite{chuzhoy2009polynomial}.

\subsection{Integrality Gap of the Flow LP of Directed Steiner Forest \label{sec:dsfgap}}

In the Directed Steiner Forest problem, we are given a  weighted, directed graph $G = (V, E, w)$ with weight $w: E \mapsto \mathbb{R}_{\geq 0}$, and a set of $p$ demand pairs $P = \{(s_i, t_i) \in V \times V \mid i \in [1, p] \}$. We are asked to return a subgraph $H \subseteq G$ minimizing $\sum_{e \in E(H)}w_e$, subject to there being a directed path from $s_i$ to $t_i$ in $H$ for all $(s_i, t_i) \in P$. 
For each edge $e \in E$, let $x_e \in \{0, 1\}$ be an indicator variable that takes value 1 if $e$ is in the solution subgraph  and else 0. We can rephrase the Directed Steiner Forest problem as the following (informal) integer program:
\begin{center}
\begin{tabular}{|c|}
\hline 
 \textsc{DSF IP}  \\
$\begin{array}{clccc}
\min & \sum_{e \in E} w_e x_e & \\
\text{s.t.} & \text{edge capacities $\{x_e\}_{e \in E}$ support a one unit $s$-$t$ flow $f_{s, t}$ in $G$} & \forall (s, t) \in P\\
& x_e \in  \{0, 1\}
\end{array}$
\tabularnewline
\hline
\end{tabular}
\end{center}
We will write $\textsc{DSF}(G, P)$ for the value of the Directed Steiner Forest problem on inputs $G, P$. A natural LP relaxation of Directed Steiner Forest is to let $x_e \geq 0$, so that edge capacities can be fractional. We refer to this LP relaxation as the ``flow LP'' of Directed Steiner Forest, since we can interpret it as a generalization of the $s$-$t$ minimum-cost flow problem to multiple demand pairs.  We will write $\widehat{\textsc{DSF}}(G, P)$ for the value of the flow LP on inputs $G, P$. Now let $\Pi_i$ denote the set of all $s_i \leadsto t_i$-paths in $G$  for all $i \in [1, p]$. 
We formally state the flow LP as:
\begin{center}
\begin{tabular}{|c|}
\hline 
 \text{Flow LP for DSF} 
 \\

$\begin{array}{clccc}
\min & \sum_{e\in E} w_ex_e \\
\text{s.t.} & \sum_{\pi \in \Pi_{i}} f_{\pi}^{i}   & \ge & 1 & \forall i \in [1, p]  \\
& \sum_{\pi \in \Pi_{i}} f_{\pi}^{i} & \leq & x_e & \forall e \in E, i \in [1, p] \\
 & x_e & \ge & 0 & \forall  e\in E \\
 & f_{\pi}^{i} & \ge & 0 & \forall i \in [1, p], \pi \in \Pi_{i}
\end{array}$
\tabularnewline
\hline
\end{tabular}
\end{center}

We define the following function to capture the integrality gap of the flow LP of Directed Steiner Forest as a function of the sizes of the inputs $G$ and $P$. 

\begin{definition}[\dsfg] The function $\dsfg(n, p)$ is the least integer $k$ such that for every $n$-node weighted, directed graph $G$ and set of demand pairs $P$ of size $|P| = p$, we have 
$$
\textsc{DSF}(G, P) \leq k \cdot \widehat{\textsc{DSF}}(G, P). 
$$
\end{definition}
 
Our goal is to prove: 

\begin{theorem}
For all $n$ and $p \in [1, n^{2 - o(1)}]$,
$$
\dsfg(n, p) = \Omega\left(\frac{\beta(n, p, \infty)}{n^{3/2}}\right).
$$
In particular, $\dsfg(n, n^{2-o(1)}) = \Omega\left(n^{1/2 - o(1)}\right)$.
\label{thm:dsf_gap}
\end{theorem}

As with our flow-cut gap lower bounds, we will actually achieve this integrality gap by lower bounding the integrality gap of \textit{Vertex} Directed Steiner Forest, which we define below. 

\begin{definition}[Vertex Directed Steiner Forest] In the Vertex Directed Steiner Forest problem, we are given an $n$-node graph $G = (V \cup V', E)$, where $V$ denotes the nonterminal vertices of $G$ and $V'$ denotes the terminal vertices of $G$ and $V \cap V' = \emptyset$. We are also given a  set of demand pairs $P \subseteq V' \times V'$. We are asked to return a subgraph $H\subseteq G$ minimizing $|V(H) \cap V|$, subject to there being a directed path from $s$ to $t$ in $H$ for all $(s, t) \in P$.    
\end{definition}
We write $\textsc{VDSF}(G, P)$ for the value of the Vertex Directed Steiner Forest problem on inputs $G, P$. We will also consider a natural LP relaxation $\widehat{\textsc{VDSF}}$, defined as follows.
\begin{center}
\begin{tabular}{|c|}
\hline 
$\widehat{\textsc{VDSF}}$ \\
$\begin{array}{clccc}
\min & \sum_{v \in V} x_v & \\
\text{s.t.} & \text{node capacities $x_v$  support a one unit $s$-$t$ flow $f_{s, t}$ in $G$} & \forall (s, t) \in P\\
& x_v \geq 0
\end{array}$
\tabularnewline
\hline
\end{tabular}
\end{center}
Let $\widehat{\textsc{VDSF}}(G, P)$ denote the value of the  $\widehat{\textsc{VDSF}}$ LP on inputs $G, P$. To lower bound $\dsfg$, it will suffice to lower bound the integrality gap of $\widehat{\textsc{VDSF}}$.  
Specifically, we will need the following lemma: 
\begin{claim}
Suppose we can construct an $n$-node graph $G$ and a set of demand pairs $P$ of size $p$ with
$$
\frac{\textsc{VDSF}(G, P)}{\widehat{\textsc{VDSF}}(G, P)} \geq k.
$$
Then $\dsfg(2n, p) \geq k$.
\label{claim:vdsf_lp}
\end{claim}
\begin{proof}
    This claim follows from a standard reduction of maximum flow in node-capacitated graphs to maximum flow in edge-capacitated graphs. We defer the proof to Appendix \ref{app:vdsf_lp}.
\end{proof}

 Our Vertex Directed Steiner Forest instance $G, P$ will have the following two properties:
\begin{enumerate}
    \item For all $(s, t) \in P$, there are at least $\Omega(\beta(n, p, \infty)/n)$  pairwise internally vertex-disjoint $s \leadsto t$ paths in $G$. 
    \item Any feasible subgraph $H \subseteq G$ must satisfy $|V(H) \cap V| = \Omega(\sqrt{n})$.
\end{enumerate}

Note that by property 1, if we assign a fractional node capacity of $x_v = cn / \beta(n, p, \infty)$ for some sufficiently large constant $c > 0$ to each vertex $v \in V$, then we obtain a feasible solution to $\widehat{\textsc{VDSF}}$ LP of size $O(n^2 / \beta(n, p, \infty))$. This, together with property 2, implies an integrality gap for Vertex Directed Steiner Forest on the $G, P$ of size $\Omega(\sqrt{n} / (n^2 / \beta(n, p, \infty))) = \Omega\left(\frac{\beta(n, p, \infty)}{n^{3/2}}\right)$, as desired.

To construct $G, P$, we will start with a  path system $S$ on $n$ nodes and $p$ paths, and with bridge girth $\infty$ and size $\|S\| = \beta(n, p, \infty)$. 
To obtain our desired construction, we will need to modify $S = (V, \Pi)$ so that it is \textit{source-restricted} with respect to a set of source nodes $X \subseteq V$. 
We say:
\begin{definition} [Source-Restricted Path Systems]
A path system $S = (V, \Pi)$ is source-restricted with respect to some $X \subseteq V$ if every path $\pi \in \Pi$ has its first node in $X$ and every following node in $V \setminus X$.
\end{definition}
To obtain our desired source-restricted path system, we will use the following modified cleaning lemma. 

\begin{lemma}[Source-Restricted Cleaning Lemma]
\label{lemma:sr_cleaning_lemma}
For all $n, p, k$, there exists a path system $S$ with $\le n$ nodes, $\le p$ paths, bridge girth $>k$, $\|S\| = \Omega(\beta(n, p, k))$, and the following two additional properties:
\begin{itemize}
\item $S$ satisfies the properties of the original Cleaning Lemma (Lemma \ref{lem:cleaning}); that is:
\begin{itemize}
\item (Approximately Degree-Regular) All nodes have degree $\Theta(d)$, where $d$ is the average degree in $S$, and

\item (Approximately Length-Regular) All paths have length $\Theta(\ell)$, where $\ell$ is the average length in $S$.
\end{itemize}
\item $S$ is source-restricted with respect to a set $X$ of size $|X| = \Theta(p/d)$. 
\end{itemize}
\end{lemma}
\begin{proof}
We defer the proof of this lemma to Appendix \ref{app:sr_cleaning_lemma} since it's  similar to the original cleaning lemma (Lemma \ref{lem:cleaning}). 
\end{proof} 

Now, using the source-restricted cleaning lemma, we may assume $S = (V, \Pi)$ is a path system on $n$ nodes and $p \in [1, n^{2-o(1)}]$ paths 
that is source-restricted with respect to a set $X \subseteq V$ of size $|X| = \Theta(p/d)$; $S$ has bridge girth $\infty$ and size $\|S\| = \Theta(\beta(n, p, \infty))$; and $S$ is approximately degree-regular and approximately  length-regular. 
Let $d$ be the average node degree of $S$ and $\ell$ be the average path length of $S$; we will assume that  $\ell = \omega(1)$. We can easily guarantee this assumption by requiring that $p = \frac{n^2}{e^{\omega(\sqrt{\log n})}} \leq n^{2 - o(1)}$  (this follows from existing reachability preserver lower bounds implied by \cite{Behrend46}).

For each $x \in X$, let $\Pi_x$ denote the set of paths in $\Pi$ that start with node $x$. Note that by construction, $\{\Pi_x\}_{x \in X}$ is a partition of the set of paths in $\Pi$. 
We will use path system $S$, along with set $X$ and collection  $\{\Pi_x\}_{x \in X}$, to construct our directed graph $G$ and set of demand pairs $P$. Roughly, each node $x \in X$ will be a terminal source node in $G$, and for each $x \in X$ we will add a new terminal sink node $y_x$ to $G$. The paths in $\Pi_x$ will become the internally vertex-disjoint $x \leadsto y_x$ paths in $G$ for $(x, y_x) \in P$. We will explicitly construct $G$ using the following procedure.
\begin{itemize}
    \item Let $G \leftarrow (V, \emptyset)$, let $P \leftarrow \emptyset$, and let $V' \leftarrow X$. 

    \item Fix a path $\pi \in \Pi_x$ for some $x \in X$. For each consecutive pair of vertices in $\pi$, add a directed edge between the corresponding pair of vertices in $G$. Repeat this procedure for each $\pi \in \Pi_x$, where $x \in X$. 

    \item For each $x \in X$, add a new terminal vertex $y_x$ to $G$, so that $G \leftarrow G \cup \{y_x\}$. 
    Let $V' \leftarrow V' \cup \{y_x\}$. 
    For each path $\pi[x \leadsto y] \in \Pi_x$, add a directed edge $(y, y_x)$ to $G$. Add the demand pair $(x, y_x)$ to $P$, so that $P \leftarrow P \cup \{(x, y_x)\}$. 
    
    \item Return the directed graph $G$ and set of demand pairs $P$. The terminal vertices of $G$ will be $V'$, and the nonterminal vertices of $G$ will be $V \setminus V'$.
\end{itemize}

We will now prove that the resulting graph $G$ and set of demand pairs $P$ has our desired properties.

\begin{lemma}
\label{lem:VDSF_lem}
The above procedure outputs a directed graph $G = (V \cup V', E)$ and a set of demand pairs $P \subseteq V' \times V'$ of size $|P| = \Theta(p/d)$  satisfying the following properties:
\begin{itemize}
    \item for all $(x, y_x) \in P$, there are $\Theta(d)$ pairwise internally vertex-disjoint $x \leadsto y_x$ paths in $G$. 
    \item for all $(x, y_x) \in P$, any solution $H \subseteq G$ to vertex directed Steiner forest on $G, P$ must contain a subpath of length $\Omega(\ell)$ of some path in $\Pi_x$.  
\end{itemize}
\end{lemma}
\begin{proof} 
Fix a node $x \in X$. Note that $x$ has degree $\Theta(d)$ in $S$ by Lemma \ref{lemma:sr_cleaning_lemma}, so $|\Pi_x| = \Theta(d)$. Note that each path $\pi \in \Pi_x$ implies an $x \leadsto y_x$ path in $G$. These paths are pairwise internally node-disjoint, since this would otherwise imply a $2$-bridge in $S$. This proves the first property.

To see why the second property is true, observe that every $x \leadsto y_x$-path in $G$  contains as a subpath a path $\pi \in \Pi_x$ (otherwise, this would imply a bridge in $S$). Then since every path $\pi \in \Pi_x$ is of length $|\pi| = \Omega(\ell)$, the second property immediately follows.   \qedhere


\end{proof}

Now we are ready to lower bound the integrality gap of $\widehat{\textsc{VDSF}}$ on $G, P$ using Lemma \ref{lem:VDSF_lem}. 

\paragraph{Fractional solution.} For each vertex $v \in V$, let the fractional node capacity be $x_v = c/d$ for a sufficiently large constant $c > 0$. Then by property 1 of Lemma \ref{lem:VDSF_lem}, since there are $\Theta(d)$ pairwise internally vertex-disjoint $x \leadsto y_x$ paths in $G$ for all $(x, y_x) \in P$, our node capacities $\{x_v\}_{ v\in V}$ support one unit of flow for all demand pairs $(x, y_x) \in P$. Then our fractional solution is feasible and has size $cn/d = O(n/d)$. 

\paragraph{Integral solution.} By property 2 of Lemma \ref{lem:VDSF_lem}, for all $(x, y_x) \in P$, any solution subgraph $H \subseteq G$ to Vertex Directed Steiner Forest on $G, P$ must contain a subpath of length $\Omega(\ell)$ of some path in $\Pi_x$. Note that $\Pi_{x_1} \cap \Pi_{x_2} = \emptyset$ for distinct $x_1, x_2 \in X$. Additionally, note that for distinct $\pi_1, \pi_2 \in \Pi$, the corresponding paths in $G$  are edge-disjoint, since $S$ has bridge girth $\infty$. Then we conclude that any feasible solution subgraph $H \subseteq G$ to Vertex Directed Steiner Forest must have at least $$|E(H)| = |P|\cdot \Omega(\ell) = \Omega\left(\frac{\ell p}{d}\right) = \Omega(n)$$ edges. Moreover, $\Omega(n)$ of these edges must be in the induced subgraph $H[V]$, since path system $S$ is source-restricted with respect to $X$. If the number of edges in $H[V]$ is $|E(H[V])| = \Omega(n)$, then 
$$|V(H) \cap V| = |V(H[V])| =  \Omega(|E(H[V])|^{1/2})  = \Omega(\sqrt{n}).$$ We conclude that the integral solution of Vertex Directed Steiner Forest must have size at least $\Omega(\sqrt{n})$.

\vspace{3mm}
By the above analysis, the integrality gap of Vertex Directed Steiner Forest on $G, P$ is at least $\Omega(\sqrt{n} / (n/d)) = \Omega(d/\sqrt{n}) = \Omega(\beta(n, p, \infty) / n^{3/2})$, as desired. We note that in particular, the generalized Ruzsa-\Szemeredi{} lower bound constructions implied by \cite{Behrend46} prove that for $p = n^{2-o(1)}$, we have that  $\beta(n, p, \infty) = \Omega(n^{2-o(1)})$, and in particular, the average path length is $ \ell =  \beta(n, p, \infty) / p = \omega(1)$. Consequently, $\dsfg(n, n^{2 - o(1)}) = \Omega(n^{1/2 - o(1)})$. Theorem \ref{thm:dsf_gap} follows from the above discussion and Claim \ref{claim:vdsf_lp}.


%% file: girthtour.tex
\section{A Tour through Prior Work on Girth Problems \label{app:girthtour}}

\subsection{The Girth Problem}

We first recall the pioneering work on girth reductions by Alth{\"o}fer, Das, Dobkin, Joseph, and Soares \cite{ADDJS93}:
\begin{definition} [Multiplicative Spanners]
A (multiplicative) $k$-spanner of a graph $G$ is a subgraph $H$ satisfying $\dist_H(s, t) \le k \cdot \dist_G(s, t)$ for all nodes $s, t$.
The function $\ms(n, k)$ is the least integer such that every undirected weighted $n$-node graph has a $k$-spanner on $\le \ms(n, k)$ edges.
\end{definition}

\begin{definition} [Graph Girth]
The girth of a graph $G$ is the least number of edges in a cycle in $G$ (or $\infty$ if $G$ is a forest).
The function $\gamma(n, k)$ is the maximum possible number of edges in an $n$-node graph of girth $>k$.
\end{definition}

\begin{theorem}  [\cite{ADDJS93}] \label{thm:greedyspanners}
$\ms(n, k) = \gamma(n, k+1)$.
\end{theorem}
\begin{proof} [Proof Sketch]
First we show that $\gamma(n, k+1) \le \ms(n, k)$.
Let $G$ be an unweighted graph with $n$ nodes, girth $>k+1$, and $\gamma(n, k+1)$ edges.
If one removes any edge $(u, v)$ from $G$, then $\dist_G(u, v)$ changes from $1$ to $>k$.
Thus $G$ is the only $k$-spanner of itself.
So if $G$ is taken as an input to the multiplicative spanner problem, one must keep $\gamma(n, k+1)$ edges in the spanner, so $\ms(n, k) \ge \gamma(n, k+1)$.

Next we show that $\ms(n, k) \le \gamma(n, k+1)$.
Let $G = (V, E, w)$ be an $n$-node graph for which we want to build a $k$-spanner.
Consider the following greedy algorithm to build a $k$-spanner.
Initially $H = (V, \emptyset)$.
Consider the edges of $G$ in nondecreasing order of weight.
When each edge $(u, v)$ is considered, we add it to $H$ iff $w(u, v) \le k \cdot \dist_H(u, v)$, i.e., the edge is currently needed in the spanner.
One can show that (1) in the end $H$ is indeed a $k$-spanner of $G$, and (2) for any cycle $C$ in $G$ that contains $\le k+1$ edges, not all edges in cycle will be added to the spanner $H$.
This is roughly because, when we consider the last edge $(u, v) \in C$, then if all previous edges from $C$ were added to $H$ then there is already a $u \leadsto v$ path of length $\le k \cdot w(u, v)$ using these edges.
Thus $H$ has girth $>k+1$, so it has $\le \gamma(n, k+1)$ edges.
So $\ms(n, k) \le \gamma(n, k+1)$.
\end{proof}

The reduction of \Althofer{} et al.~\cite{ADDJS93} generalizes also to \emph{emulators} and more generally to \emph{distance oracles}, which are arbitrary data structures that can approximate the distances of the input graph on query (see also \cite{TZ05}).
Recently, tight reductions to $\gamma$ have been achieved for vertex fault tolerant spanners as well, which ask for the size bounds for $k$-spanners that retain their distance approximation even after a bounded number of vertices fail in both the spanner and the original graph \cite{BP19, BDPV18}.

The ``girth problem'' asks for the asymptotic value of $\gamma$, which hence would also determine the asymptotic value of $\ms$.
This is a major open question in extremal combinatorics and theoretical computer science.
The following upper bound is known:
\begin{theorem} [Moore Bounds, Folklore] \label{thm:moore}
For any integers $n, k$, we have $\gamma(n, 2k) = O\left(n^{1+1/k}\right)$.
\end{theorem}
\begin{proof} [Proof Sketch]
Let $G$ be an $n$-node graph of average degree $d$, and assume that $d$ is at least a sufficiently large constant.
A \emph{non-backtracking $k$-path} is a path in $G$, containing exactly $k+1$ nodes and $k$ edges, which may repeat nodes or edges but which never uses an edge $(u, v)$ followed consecutively by its reverse $(v, u)$.
The following are facts from graph theory:
\begin{itemize}
\item $G$ has $n \cdot \Omega(d)^k$ non-backtracking $k$-paths, and
\item If $G$ has two different non-backtracking $k$-paths with the same pair of endpoints $(s, t)$, then $G$ has a cycle on $\le 2k$.
\end{itemize}
Together, these imply that if $G$ has girth $>2k$, then it can only have $O(n^2)$ non-backtracking $k$-paths, and hence $n \cdot \Omega(d)^k = O(n^2)$.
Rearranging gives $d = O(n^{1/k})$, proving the theorem.
\end{proof}

Unfortunately, lower bounds are not as well understood.
The Moore bounds are known to be asymptotically tight when $k \in \{1, 2, 3, 5\}$ \cite{Wenger91, Tits59}.
The \emph{girth conjecture}, attributed to \Erdos{} \cite{Erdos75}, posits that the Moore bounds are tight for all other values of $k$ as well.
The girth conjecture is controversial, with no clear consensus from experts on whether it is likely to be true.

\subsection{The Weighted Girth Problem \label{sec:wtdgirth}}

Besides number of edges, in some applications one wants to minimize the \emph{total weight} of a spanner.
This is often measured as the \emph{lightness} of the spanner, relative to the input graph:

\begin{definition} [Spanner Lightness]
The lightness of a subgraph $H$ of a graph $G$ is
$$\ell(H \mid G) := \frac{w(H)}{w(\mst(G))}$$
where $\mst(G)$ is any minimum spanning tree of $G$ (or spanning forest if disconnected).
We write $\lms(n, k)$ for the least\footnote{Formally, one takes the $\inf$ of the values $L$ satisfying this condition} $L$ such that every $n$-node graph $G$ has a $k$-spanner $H$ of lightness $\ell(H \mid G) \le L$.
\end{definition}

In their study of light spanners, Elkin, Neiman, and Solomon \cite{ENS14} made the interesting point that an extension of the \Althofer{} et al~\cite{ADDJS93} reduction between $\gamma$ and $\ms$ also gives equivalence between the extremal function of graph lightness and \emph{weighted girth}, defined as follows:
\begin{definition} [Weighted Girth]
The \emph{weighted girth} of a graph $G$ is defined as
$$\min \limits_C \frac{w(C)}{\max \limits_{e \in C} w(e)}$$
where the min is over the set of cycles $C$ in $G$.
We define $\lambda(n, k)$ as the maximum\footnote{Formally, $\lambda$ is determined by the sup of the lightness of graphs satisfying this property.} lightness over $n$-node graphs of weighted girth $>k$.
\end{definition}

Note that weighted girth generalizes girth, in the sense that the concepts coincide for an unweighted graph.
Elkin et al.~\cite{ENS14} proved:
\begin{theorem} [\cite{ENS14}]
$\lms(n, k+1) = \lambda(n, k)$.
\end{theorem}

A natural next question is to ask for the relative values of $\lambda$ and $\gamma$.
It follows by considering the unweighted graph realizing $\gamma$ that
$$\lambda(n, k) = \Omega\left( \frac{\gamma(n, k)}{n} \right)$$
(note: we divide by $n$ on the right, since an unweighted graph has an MST of weight $n-1$).
A fascinating conjecture by Elkin et al.~\cite{ENS14}, known as the \emph{weighted girth conjecture}, implies that these bounds are asymptotically equal.
This remains open, but recent work of Le and Solomon \cite{LS22} implies that they are \emph{approximately} equal.

\subsection{The Bipartite Girth Problem}


The function $\gamma$ has a natural generalization to the setting of bipartite graphs:
\begin{definition} [The Extremal Function of Bipartite High-Girth Graphs]
The function $\gamgam(n, p, k)$ is the maximum possible number of edges in a bipartite graph with $n$ nodes on one side of the bipartition, $p$ nodes on the other side, and girth $>k$.
\end{definition}

We say that $\gamgam$ \emph{generalizes} $\gamma$, rather than merely being different, due to the following fact:
\begin{theorem} [Folklore] \label{thm:ggtog}
$\gamgam(n, n, k) = \Theta(\gamma(n, k))$.
\end{theorem}
\begin{proof} [Proof Sketch]
In one direction, we have
$$\gamgam(n, n, k) \le \gamma(2n, k) = \Theta(\gamma(n, k))$$
where the first inequality is immediate from the definitions, and the second inequality is by observing that $\gamma$ depends at most polynomially on its first parameter.
In the other direction, we show
$$\gamma(n, k) \le 2 \cdot \gamgam\left( \frac{n}{2}, \frac{n}{2}, k\right) \le O\left(\gamgam(n, n, k) \right).$$
The second inequality is immediate from the definitions.
For the first inequality, we start with a graph $G$ realizing $\gamma(n, k)$, and randomly bipartition its nodes into two parts of size $n/2$ each.
Let $G'$ be the bipartite subgraph that keeps only edges crossing the random bipartition.
Each edge survives in $G'$ with probability $\ge 1/2$.
Thus we have constructed a bipartite graph with $n/2, n/2$ nodes per side, girth $>k$, and $\ge \gamma(n, k)/2$ edges in expectation, which implies the first inequality. 
\end{proof}

Thus every extremal reduction to $\gamma$ can also be expressed as a reduction to a special case of $\gamgam$.
However, there are some further problems in distance sketching and extremal combinatorics that can \emph{only} be reduced to $\gamgam$, rather than the non-bipartite version.
First, the obvious bipartite generalizations of multiplicative spanners and related objects can be reduced to $\gamgam$, again by the reduction of \Althofer{} et al.~\cite{ADDJS93}.
More interestingly:
\begin{itemize}
\item Fern{\'a}ndez, Yasuda, and Woodruff \cite{FWY20} constructed lower bounds against the communication complexity of spanner construction, converting $\gamgam$ lower bounds to lower bound instances, and
\item Bodwin, Dinitz, and Robelle \cite{BDR22} used $\gamgam$ to provide lower bounds against edge fault tolerant spanners and edge distance sensitivity oracles (based on a construction from \cite{BDPV18}).
\end{itemize}

The following natural extension of the Moore bounds holds for bipartite graphs:
\begin{theorem} [Bipartite Moore Bounds]
For all $n, p, k$, we have
$$\gamgam(n, p, 2k) = \begin{cases} O\left( (np)^{1/2 + 1/(2k)} + n + p\right) & \text{if $k$ is odd}\\
O\left( n^{1/2 + 1/k} p^{1/2} + n + p\right) & \text{if $k$ is even}
\end{cases}$$
\end{theorem}

The proof is in the same spirit as Theorem \ref{thm:moore}, but with sensitivity to the average degree on either side of the bipartite graph.
The bipartite Moore bounds are known to be fully tight for girth parameters $2k$ when $k \in \{1, 2\}$, and they are also tight for various relative values of $n, p$ when $k \in \{3, 5, 7\}$ \cite{van12}.
Analogizing the girth conjecture, one might conjecture that the Moore bounds are tight for all $n, p, k$.
However, this was \emph{refuted} in an important paper by de Caen and \Szekely{} \cite{de1991maximum}, which showed that the Ruzsa-\Szemeredi{} theorem (discussed next) is equivalent to a (subpolynomial) improvement to the upper bounds on $\gamgam(n, p, 6)$, and thus it implies an improvement on the bipartite Moore bounds in a particular parameter setting.
At a technical level, this proof is very similar to Theorem \ref{thm:b3rs}, so we shall not repeat it here.

Recent work of Conlon, Fox, Sudakov, and Zhao \cite{CFSZ21} implies an analogous improvement to $\gamgam(n, p, 10)$; it is an interesting open problem to obtain an analogous improvement to $\gamgam(n, p, 2k)$ for any other odd $k$.

\subsection{The Ruzsa-\Szemeredi{} Problem \label{app:rs}}

The Ruzsa-\Szemeredi{} problem was introduced by Ruzsa and \Szemeredi{} \cite{RS78}, in the context of a combinatorial problem about hypergraphs.
Their result was one of the first major uses of the famous \Szemeredi{} regularity lemma \cite{Szemeredi75}.
Although it has been interpreted and reinterpreted over the years, the standard phrasing is as follows:
\begin{definition} [Induced Matchings and $\rs(n)$]
In a graph $G = (V, E)$, an \emph{induced matching} is an edge subset $M \subseteq E$ that is a matching, and also the edge subset of an induced subgraph.
In other words, for any two edges $(u_1, v_1), (u_2, v_2) \in M$, we have $(u_1, u_2), (u_1, v_2), (v_1, u_2), (v_1, v_2) \notin E$.

We define $\rs(n)$ as the largest integer such that, for every $n$-node graph $G$ whose edge set can be partitioned into $n$ induced matchings, we have
$$|E(G)| \le \frac{n^2}{\rs(n)}.$$
\end{definition}

Besides induced matchings, there are many other natural ways to interpret $\rs(n)$ \cite{CF13}.
The following transformation can be used to connect $\rs(n)$
One is: let $G$ be a graph that can be decomposed into $n$ induced matchings, which has $n^2 / \rs(n)$ edges.
Direct the edges of $G$ arbitrarily, and then add a new node $m_1, \dots, m_n$ for each of the $n$ induced matchings.
Then, for each $i$ and for each directed edge $(u, v)$ in the $i^{th}$ induced matching, interpret the triple $(m_i, u, v)$ as a $3$-path.
One can verify that this yields a bridge-free path system.
This transformation is well known, even though the description as a ``bridge-free path system'' is new.

It is not at all obvious from the definition of $\rs(n)$ that the function is nontrivial, i.e., super-constant.
But indeed, Ruzsa and \Szemeredi{} proved that $\rs(n) = \Omega(\log^* n)$.
The state-of-the-art upper bound is due to Fox \cite{Fox11}; a notable alternate proof was discovered by Moshkovitz and Shapira \cite{MS19}.
The state-of-the-art lower bound is due to Behrend \cite{Behrend46} (see also \cite{Elkin10}).
These bounds are:
$$2^{\Omega(\log^* n)} \le \rs(n) \le 2^{O(\log^c n)}.$$
While it is not clear from the definition that the Ruzsa-\Szemeredi{} problem should be regarded as a girth concept, an important paper by de Caen and \Szekely{} \cite{de1991maximum} explains its inclusion, by tightly reducing between $\rs(n)$ and $\gamgam$.
Specifically, their reduction may be interpreted as follows.
Given a value of $n$, let $p^*$ be the largest integer such that
$\gamgam(n, p^*, 6) \ge 3p^*$.
Then:
$$\gamgam(n, p^*, 6) = \Theta\left( \frac{n^2}{\rs(n)} \right).$$
Thus, the Ruzsa-\Szemeredi{} problem is a special case of the bipartite girth problem.
In network design, we mention three applications of the Ruzsa-\Szemeredi{} problem:
\begin{itemize}
\item Given an $n$-node undirected unweighted graph and a set of $p$ demand pairs, one can construct a distance preserver on $O(n^2 / \rs(n) + p)$ edges \cite{Bodwin21}.

\item For undirected unweighted input graphs with $n$ nodes and $O(n)$ edges, one can construct a distance labeling scheme with average label size $O(n / \rs(n)^c)$ \cite{KUV19}.

\item Bansal and Williams \cite{BW09} developed a combinatorial algorithm for All-Pairs Shortest Paths in unweighted graphs, by reducing to a certain \emph{algorithmic} version of the Ruzsa-\Szemeredi{} problem.
\end{itemize}

Some other miscellaneous uses of $\rs(n)$ or Ruzsa-\Szemeredi{} graphs in theoretical computer science include connections to the PCP theorem by H{\aa}stad and Wigderson \cite{HW03}, applications in Channel Scheduling by Birk, Linial, and Meshulam \cite{BLM93}, a line of work on maximum matching in streams \cite{Konrad15, Kapralov13, GKK12}, and a line of work on subgraph testing algorithms \cite{Alon02, AS03}.

\subsection{The Set Girth Problem}

The function $\gamgam(n, p, k)$ has an equivalent interpretation in the language of \emph{set systems} rather than graphs.
We consider:
\begin{definition} [Set Systems] ~
\begin{itemize}
    \item A set system is a pair $S = (V, \mathcal{T})$, where $V$ is a ground set of ``nodes'' and $\mathcal{T}$ is a multiset of node subsets.
    
    \item A \emph{$k$-cycle} in a set system $S$ is a circularly-ordered list of distinct nodes $v_0, v_1, \dots, v_k=v_0 \in V$ and sets $T_0, T_1, \dots, T_k=T_0 \in \mathcal{T}$ with each $v_i, v_{i+1} \in T_i$.
    
    \item The \emph{girth} of a set system is the smallest integer $k$ for which the system has a $k$-cycle.
    
    \item The size of a set system is written $\|S\| := \sum \limits_{T \in \mathcal{T}} |T|$.
\end{itemize}
\end{definition}

For example, a set system in which each set has size $2$ can be considered as an undirected graph.
Set systems are merely a rephrasing of bipartite graphs, and one can switch between them via \emph{incidence graphs}.
In particular:

\begin{theorem} [Folkore] \label{thm:ggtoset}
Over set systems $S$ with $n$ nodes, $p$ sets, and girth $>k$, the maximum possible value of $\|S\|$ is exactly $\gamgam(n, p, 2k)$.
\end{theorem}
\begin{proof} [Proof Sketch]
A set system $S = (V, \mathcal{T})$ can be naturally bijected with its \emph{incidence graph} $G_S$.
This is a bipartite graph whose nodes on the left correspond to $V$, whose nodes on the right correspond to $\mathcal{T}$, and whose edges correspond to set membership; that is, we put an edge between $v \in V$ and $T \in \mathcal{T}$ iff $v \in T$.
Set systems carry the same information as their incidence graph.
One can verify that (1) the size $\|S\|$ of the set system is the same as the number of edges $|E(G_S)|$ of its incidence graph, and (2) if the set system has girth $k$, then its incidence graph has girth $2k$.
The theorem follows from these properties.
\end{proof}

%% file: cleaning.tex
\section{Proof of Cleaning Lemma \label{app:cleaning}}

We now prove the Cleaning Lemma (Lemma \ref{lem:cleaning}).
We will state the proof only for unordered bridge girth; the proof for ordered bridge girth is completely identical.
We split the proof into the following two lemmas:

\begin{lemma}
Suppose $S$ is a path system with $n$ nodes, $p$ paths, bridge girth $b$, average node degree $d$, and average path length $\ell$.
Then there exists a path system $S'$ that has:
\begin{itemize}
\item $n' = \Theta(n)$ nodes,
\item $p' = \Theta(p)$ paths,
\item size $\|S'\| = \Theta(\|S\|)$,
\item bridge girth $\ge b$,
\item average degree $d' = \Theta(d)$, and all nodes have degree $\Theta(d')$,
\item average length $\ell' = \Theta(\ell)$, and all paths have length $\Theta(\ell')$.
\end{itemize}
\label{lem:true_cleaning}
\end{lemma}
\begin{proof}
We construct $S'$ by the following process.
Start with $S = (V, \Pi)$ as a path system with $n$ nodes, $p$ paths, bridge girth $>k$, and size $\|S\| = \beta(n, p, k)$.
Fix $\ell, d$ as the initial average length and degree of $S$.
Then, perform the following sequence of operations on $S$:

\begin{enumerate}
    \item While there exists a path $\pi \in \Pi$ of length $|\pi| \geq \ell/2$, split $\pi$ into two node-disjoint paths $\pi_1$ and $\pi_2$ in any way such that $\pi = \pi_1 \circ \pi_2$, $|\pi_1| \geq \ell/4$, and $|\pi_2| \geq \ell/4$.  
    
    \item While there exists a node $v \in V$ of degree $\ge d/2$, split $v$ into two new nodes $v_1, v_2$.
    Replace each occurrence of $v$ in a path with either $v_1$ or $v_2$, in any way such that $\deg(v_1) \geq d/4$ and $\deg(v_2) \ge d/4$.
    
    \item While there exists a node of degree $< d/4$, or a path of length $< \ell / 4$, delete that node or path from $S$.
\end{enumerate}

Let $S'$ be the resulting path system on $n'$ nodes and $p'$ paths.
First note that the construction must terminate, since no step of the construction increases the size of $\|S'\|$.
Our only operations are to delete nodes/paths and to split nodes/paths, which do not create bridges; thus, since $S$ does not have a bridge of size $<b$, $S'$ also has no bridge of size $<b$.
Operations that split nodes and paths do not change the size of $S$.
It is immediate from the construction that all surviving paths $\pi$ have length $\Theta(\ell)$ and that all surviving nodes have degree $\Theta(d)$.
Meanwhile, we only delete nodes of degree $<d/4$ and paths of length $<\ell/4$, so we have
$$\|S'\| > \|S\| - (nd/4 + p\ell/4) = \|S\|/2.$$
Thus we have $\|S'\| = \Theta(\|S\|)$.
Moreover, we notice that
$$\ell' p' = n' d' = \|S'\| > \|S\|/2 = nd/2 = p\ell/2.$$
Since $\ell' = \Theta(\ell)$ and $d' = \Theta(d)$, this implies that $p' = \Theta(p)$ and $n' = \Theta(n)$, completing the proof.
\end{proof}

\begin{lemma} \label{lem:paramadjust}
For any absolute constant $0 < c < 1$, we have $\beta(cn, cp, k) = \Omega(\beta(n, p, k))$.
\end{lemma}
\begin{proof}
We will prove for $\beta$; the proof for $\beta^*$ is identical.
Let $S = (V, \Pi)$ be path system with $n$ nodes, $p$ paths, bridge girth $>k$, and $\|S\| = \beta(n, p, k)$.
Let $S' \subseteq S$ be a subsystem obtained by choosing exactly $cn$ nodes in $V$ uniformly at random and $cp$ paths in $\Pi$ uniformly at random, and keeping these nodes and the paths induced on these nodes in $S'$, while deleting the rest of $S$. 
Then $S'$ has $cn$ nodes, $cp$ paths, bridge girth $>k$ and expected size
$$\mathbb{E}\left[\|S'\|\right] = c^2 \|S\| = \Omega(\beta(n, p, k)),$$
which completes the proof.
\end{proof}

We can now state the proof of the cleaning lemma.
Using the latter lemma, we can choose $c>0$ as a sufficiently small constant, and then start with $S$ as a path system with $\lceil cn \rceil$ nodes, $\lceil cp \rceil$ paths, bridge girth $>k$, and $\|S\| = \Omega(\beta(n, p, k))$.
Then, applying the former lemma, we can find a path system $S'$ that has $n' = \Theta(cn)$ nodes, $p' = \Theta(cp)$ paths, bridge girth $>k$, size $\|S'\| = \Omega(\beta(n, p, k))$, all nodes have degree $\Theta(d')$, and all paths have length $\Theta(\ell')$.
By choice of sufficiently small $c$, we have $n' \le n$ and $p' \le p$, and thus $S'$ satisfies the cleaning lemma.

%% file: integralitygaps_app.tex
\section{ Missing Proofs for Section \ref{sec:int_gap}} 
\subsection{Lemma \ref{lem:big_multicut}}
\label{app:big_multicut}
Fix an $\varepsilon > 0$. Let $\mathcal{S} \subseteq V_H$ be a set of vertices of size $|\mathcal{S}| \leq d/16$ in $H$. Fix an $i \in [1, d]$, and observe that when we are choosing the vertices in layer $L^j_i$, the size of the set $V_H \setminus  (L^1_i \cup L^2_i \cup \dots \cup L^{j-1}_i)$ is at least $d/2$.  Therefore, the probability that $L^j_i \subseteq \mathcal{S}$ is at most 
\[
\left ( \frac{|\mathcal{S}|}{d/2} \right ) ^{ \log d} \leq \left (\frac{1}{8}\right)^{\log d} = d^{-3}
\]
Note that $\mathcal{S}$ separates $(s_i, t_i)$ if and only if $L^j_i \subseteq \mathcal{S}$ for some $j \in [1, d']$. 
Then the probability that $\mathcal{S}$ separates $(s_i, t_i)$ is at most $d' \cdot d^{-3} \leq d^{-2}$, by the union bound and the fact that $d' \leq d$.   Now since our constructions of each graph $H_i$ are independent, the probability that at least $\varepsilon d$ distinct demand pairs $(s_i, t_i)$, $i \in [1, d]$ are disconnected by $\mathcal{S}$ is  at most
\[
\binom{d}{\varepsilon d} ( d^{-2} )^{\varepsilon d} = \left(\frac{1}{d}\right)^{\varepsilon d}
= 2^{-\varepsilon d \log d} < 2^{-d}
\]
for sufficiently large $d$. We have established our desired claim. 


\subsection{Theorem \ref{thm:fcg_sc}}
\label{app:sc}
Recall that our goal is to lower bound the flow-cut gap between concurrent multicommodity flow and (non-bipartite) sparsest cut.
We will accomplish this by lower bounding an integrality gap for the directed sparsest \textit{vertex} multicut problem, \textsc{SVCut}. This problem  is defined identically to sparsest cut except we choose a set of non-terminal \textit{vertices} $\mathcal{S}$ that minimizes the ratio $|\mathcal{S}| / |P_{\mathcal{S}}|$, where $P_{\mathcal{S}}$ is the set of demand pairs disconnected in $G \setminus \mathcal{S}$. Below we describe a natural  LP relaxation $\widehat{\textsc{SVCut}}$  of the directed sparsest \textit{vertex} cut problem.
\begin{center}
\begin{tabular}{|c|}
\hline 
$\widehat{\textsc{SVCut}}$ \\
$\begin{array}{crccc}
\min & \sum_{v \in V} x_v' \\
\text{s.t.} & \sum_{v\in \pi \cap V} x_v' & \ge & h_i' & \forall i\in [p],\pi\in\Pi_i\\
& \sum_{i=1}^p h_i' & \geq & 1 & \\
 & x_v', h_i' & \ge & 0 & \forall  v\in V, \forall i \in [p]
\end{array}$
\tabularnewline
\hline
\end{tabular}
\end{center}

By the discussion in Section 2.2 of \cite{chuzhoy2009polynomial}, the integrality gap between $\widehat{\textsc{SVCut}}$ and \textsc{SVCut} is at most the integrality gap between $\widehat{\textsc{SCut}}$ and \textsc{SCut}. By lower bounding the integrality gap of $\widehat{\textsc{SVCut}}$,  we will immediately obtain lower bounds for the directed sparsest cut flow-cut gap.

We will lower bound the integrality gap for $\widehat{\textsc{SVCut}}$ using our construction from the proof of Theorem \ref{thm:fcg} and a standard argument from \cite{chuzhoy2009polynomial}. 
Let $G = (V \cup V', E)$ be the graph on $N := \beta(n, n, \infty)$ non-terminal nodes defined in Section \ref{sec:multicut_const}, and let $P \subseteq V' \times V'$ be the corresponding set of demand pairs of size $|P| = n$. Observe the following solution to the sparsest vertex cut LP. For every $(s_i, t_i) \in P$, let $h_i' := 1/n$. For every non-terminal vertex $v \in V$, let $x_v' := 1/(nd')$. This is a feasible solution to $\widehat{\textsc{SVCut}}$ of size $N/(nd')$ by Claim \ref{claim:long_paths_in_G}. 

Now assume that $\widehat{\textsc{SVCut}}$ has integrality gap less than $g(N)$ for some function $g$, and fix a sufficiently small $\varepsilon > 0$. We will show that there is an (integral) vertex cut $\mathcal{S}$ of $G$ of size $|\mathcal{S}|= O(N / d')g(N)$ that disconnects more than  a $(1 - \varepsilon)$-fraction of the demand pairs $P$ in $G $. Recall that by Theorem \ref{thm:fcg}, there is an $\varepsilon > 0$ such that $|\mathcal{S}| = \Omega(N)$ if $\mathcal{S}$ disconnects more than  a $(1 - \varepsilon)$-fraction of demand pairs $P$ in $G$. Then we will conclude that $g(N) = \Omega(d')$. 

Fix a sufficiently small $\varepsilon > 0$. Our construction of $\mathcal{S}$ will proceed in rounds, where in each round we will add nodes in $G$ to $\mathcal{S}$ and disconnect some demand pairs in $G \setminus \mathcal{S}$. 
We will repeat our procedure until $\mathcal{S}$ disconnects more than a $(1 - \varepsilon)$-fraction of the demand pairs $P$. Let $G_1 := G$ and $P_1 := P$. In round 1, $G_1, P_1$ has a feasible solution to $\widehat{\textsc{SVCut}}$ of size $\varphi := N/(nd')$, so there is an integral solution to \textsc{SVCut} of size $\varphi \cdot g(N)$.  This means there is a set  $\mathcal{S}_1 \subseteq V$ in $G$ of size $|\mathcal{S}_1| = \varphi \cdot g(N) \cdot p_1$  that disconnects $p_1 \geq 1$ pairs $P_{\mathcal{S}_1}\subseteq P_1$. Add $\mathcal{S}_1$ to $\mathcal{S}$, and let $G_2 := G \setminus \mathcal{S}_1$ and $P_2 := P_1 \setminus P_{\mathcal{S}_1}$.

In round $i$, we are given a graph $G_i \subseteq G$ and a set of demand pairs $P_i \subseteq P$ that are connected in $G_i$. We halt when $|P_i| <  \varepsilon|P|$, so we may assume that $|P_i| \geq \varepsilon|P|$. Then if we let $h_j' :=  \varepsilon^{-1}n^{-1}$ for every $(s_j, t_j) \in P_i$ and let $x_v' := (\varepsilon nd')^{-1}$, then this is a feasible solution to $\widehat{\textsc{SVCut}}$ for $G_i, P_i$ of size at most $ N / (\varepsilon nd')$ by Claim \ref{claim:long_paths_in_G}. Then there is an integral vertex cut $\mathcal{S}_i$ of size at most $\varepsilon^{-1} \varphi \cdot g(N) \cdot p_i$ that disconnects $p_i \geq 1$ pairs $P_{\mathcal{S}_i} \subseteq P_i$.  Add $\mathcal{S}_i$ to $\mathcal{S}$, and let $G_{i+1} := G_i \setminus \mathcal{S}_{i}$ and $P_{i+1} := P_i \setminus P_{\mathcal{S}_i}$. When our procedure ends in round $k$, we will have a set $\mathcal{S}$ of size at most $\varepsilon^{-1} \varphi \cdot g(N) \cdot n = O(N/d')g(N)$ that disconnects at least $|P| - |P_k| > (1 - \varepsilon )|P|$ pairs in $P$. 
Then  $|\mathcal{S}| = \Omega(N)$ by Lemma \ref{lem:big_cut_fail} and the subsequent discussion, so we conclude that $g(N) = \Omega(d')$, as desired. The theorem is immediate from the above discussion. 

\subsection{Claim \ref{claim:vdsf_lp}}
\label{app:vdsf_lp} 
Fix an $n$-node unweighted directed graph $G = (V \cup V', E)$ with non-terminal nodes $V$ and terminal nodes $V'$, and a set of demand pairs $P \subseteq V' \times V'$ of size $p$.
Inputs $G, P$ correspond to an instance of the Vertex Directed Steiner Forest problem. Given $G, P$, we will construct an instance $G_1, P_1$ of the Directed Steiner Forest problem. We will then show that the integrality gap of $\widehat{\textsc{VDSF}}$ on $G, P$ is at most the integrality gap of $\widehat{\textsc{DSF}}$ on $G_1, P_1$, which will imply Claim \ref{claim:vdsf_lp}.

We  construct the weighted directed graph $G_1 = (V_1, E_1, w)$ as follows.  Initialize $G_1$ as $G_1 := (V \cup V', E)$. For each vertex $v \in V$, we replace $v$ in $G_1$ with a special directed edge $(v^+, v^-)$, which we assign weight $w_{(v^+, v^-)} = 1$. For each original edge $(u, v) \in E \cap (V\times V)$, we replace it with edge $(u^-, v^+)$, which we assign weight $w_{(u^-, v^+)} = 0$. For each original edge $(u, v) \in E \cap (V \times V')$, we replace it with $(u^-, v)$ and assign it weight $w_{(u^-, v)} = 0$; and for each original edge $(u, v) \in E \cap (V' \times V)$, we replace it with $(u, v^+)$ and assign it weight $w_{(u, v^+)} = 0$. 
This concludes our construction of $G_1$. We let $P_1 := P$. Then $G_1, P_1$ will be our corresponding instance of Directed Steiner Forest.

Let $\{x_v\}_{v \in V}$ be a feasible solution to $\widehat{\textsc{VDSF}}$  on $G, P$. Now for every $v \in V$, assign capacity $x_v$ to edge $(v^+, v^-)$ in $G_1$. For all other edges in $G_1$, assign capacity 1. Observe that the resulting solution to $\widehat{\textsc{DSF}}$ on inputs $G_1, P_1$  is feasible and has capacity 
$$\sum_{e \in E(G_1)}w_ex_e =  \sum_{v \in V} x_v \leq \widehat{\textsc{VDSF}}(G, P).$$  Then the size of the optimal solution to 
$\widehat{\textsc{DSF}}(G_1, P_1)$ is at most
$$\widehat{\textsc{DSF}}(G_1, P_1) \leq 
\widehat{\textsc{VDSF}}(G, P).$$

Now consider a feasible solution to Directed Steiner Forest on inputs $G_1, P_1$. This optimal solution corresponds to a subgraph $H_1$ of $G_1$. Now we define a corresponding feasible  solution to Vertex Directed Steiner Forest as follows. Let $H$ be the induced subgraph $$H = G[\{ v \in V \mid  (v^+, v^-) \in E(H_1)  \} \cup V'].$$   Observe that $H$ is a feasible  solution for Vertex Directed Steiner Forest on inputs $G, P$ and has  size at most $$|V(H) \cap V| =  |\{ v \in V \mid  (v^+, v^-) \in E(H_1)  \}| =  \sum_{e \in E(H_1)} w_e \leq \textsc{DSF}(G_1, P_1).$$ Then the size of the optimal solution to Vertex Directed Steiner forest on $G, P$ is at most
$$
\textsc{VDSF}(G, P) \leq \textsc{DSF}(G_1, P_1).
$$
Now suppose that for some integer $k$, inputs $G, P$ satisfy
$$
k \leq \frac{\textsc{VDSF}(G, P)}{\widehat{\textsc{VDSF}}(G, P)}.$$
Then
$$k \leq \frac{\textsc{VDSF}(G, P)}{\widehat{\textsc{VDSF}}(G, P)} \leq \frac{\textsc{DSF}(G_1, P_1)}{\widehat{\textsc{DSF}}(G_1, P_1)} = \dsfg(|V(G_1)|, p) \leq \dsfg(2n, p).$$


\subsection{Lemma \ref{lemma:sr_cleaning_lemma}}
\label{app:sr_cleaning_lemma}

We now prove the source-restricted cleaning lemma (Lemma \ref{lemma:sr_cleaning_lemma}). 
The proof will require the following lemma:

\begin{lemma}
\label{lem:scwise_cleaning_1}
    Suppose $S$ is a path system with $n$ nodes, $p$ paths, bridge girth $b$, average node degree $d$, and average path length $\ell$. Then there exists a path system $S'$ that has:
    \begin{itemize}
        \item $n' = \Theta(n)$ nodes,
        \item $p' = \Theta(p)$ paths,
        \item size $\|S'\| = \Theta(\|S\|)$,
        \item bridge girth $\geq b$,
        \item average degree $d' = \Theta(d)$, and all nodes have degree $\Theta(d')$,
        \item average length $\ell' = \Theta(\ell)$, and all paths have length $\Theta(\ell')$, 
        \item $S'$ is source-restricted with respect to a set $X$ of size $|X| = \Theta(p/d)$.
    \end{itemize}
\end{lemma}
\begin{proof}
    By the cleaning lemma (Lemma \ref{lem:cleaning}), we may assume that $S = (V, \Pi)$ is approximately degree-regular and length-regular.\footnote{This initial application of the cleaning lemma is not technically needed to make the following analysis work, but it simplifies the analysis a bit.} In particular, we may assume that $c_1\ell \leq |\pi| \leq c_2 \ell$ for $\pi \in \Pi$, and $\deg(v) \leq c_3d$ for $v \in V$,  where $c_1, c_2, c_3$ are positive universal constants. Additionally, we will assume that $\ell$ is greater than a sufficiently large constant; when $\ell \leq c$ for a constant $c >0$, we can simply take $X := V$ and make our path system source-restricted with respect to $X$ by shortening all paths in $\Pi$ until they contain only one node. Now perform the following sequence of operations on $S$:
\begin{enumerate}
    \item Uniformly at random, sample a subset $X\subseteq V$ of size $|X| = \frac{p}{d}$. Delete all paths from $S$ that do not contain a node in $X$.
    \item  For each path $\pi \in \Pi$, let $x_{\pi} \in \pi \cap X $ be the first node in $\pi$ that is also in $X$. Delete all nodes from $\pi$ preceding $x_{\pi}$; that is, all nodes $y$ such that $y <_{\pi} x_{\pi}$.
    \item For each path $\pi \in \Pi$, delete all occurrences of nodes in $X \setminus \{x_{\pi}\}$ from $\pi$.  
    \item  While there exists a node $v$ of degree $\deg(v) <  \lambda d$ or a path $\pi$ of length $|\pi| < \lambda \ell$, where $\lambda > 0$ is a sufficiently small constant, delete that node or path from $S$. If a node $v \in X$ is deleted, then delete all paths in $\Pi$ that contain $v$.
\end{enumerate}
Let $S' := (V', \Pi')$ be the resulting path system. We will now prove that with nonzero probability, $\|S'\| = \Theta(\|S\|)$. All other properties of $S'$ claimed in the lemma are immediate or follow from arguments identical to those of Lemma \ref{lem:true_cleaning}.

Fix a path $\pi$ in $\Pi$, and let $\pi_1$ be the prefix of $\pi$ corresponding to the first $ c_1/2 \cdot \ell $ nodes. Observe that $\pi_1$ contains a node in $X$ with probability 
$$
\Pr[ \pi_1 \cap X \neq \emptyset ] = 1 - \left(1 - \frac{p/d}{n}\right)^{|\pi_1|} \geq 1 - e^{-\frac{c_1/2 \cdot \ell p}{dn}} = 1 - e^{-c_1/2}.
$$
Now suppose that $\pi_1 \cap X \neq \emptyset$, so that $x_{\pi} \in \pi_1$. Then $\pi$ survives step 1 of our procedure. Let $\pi_2$ be a subpath of $\pi$ such that $\pi = \pi_1 \circ \pi_2$. Since $x_{\pi} \in \pi_1$, we  are guaranteed that $\pi_2$ is a subpath of a path surviving after the second step of our procedure; moreover, $|\pi_2| \geq c_1/2 \cdot \ell$, since $|\pi| \geq c_1 \ell$.  Let $S_2$ be the  path system obtained after performing the first two steps of our procedure on $S$. Then the expected size of $\|S_2\|$ is
$$
\mathbb{E}[\|S_2\|] \geq \sum_{\pi \in \Pi} \Pr[\pi_1 \cap X \neq \emptyset] \cdot |\pi_2| \geq (1-e^{-c_1/2})p \cdot c_1/2 \cdot \ell = c_1/2 \cdot (1-e^{-c_1/2}) \|S\|.
$$
Consequently, we may assume that $S_2$ satisfies $\|S_2\| \geq c_1/2 \cdot (1-e^{-c_1/2}) \|S\|$. Now we just need to bound the amount that $\|S_2\|$ decreases in steps 3 and 4 of our procedure. 
Note that the total decrease of $\|S_2\|$ in step 3 is at most $$\sum_{x \in X} \deg(x) \leq |X|\cdot c_3d = c_3p.$$ 
As stated earlier, we may assume $\ell$ is greater than a sufficiently large constant. If we assume $\ell > \frac{4c_3}{c_1/2 \cdot(1-e^{-c_1/2})}$, then the total decrease of $\|S_2\|$ in step 3 is at most
$$
c_3p \leq  \left( \frac{c_1/2 \cdot(1-e^{-c_1/2})}{4} \cdot \ell \right) \cdot  p \leq \frac{c_1/2 \cdot(1-e^{-c_1/2})}{4} \cdot \|S\| \leq \|S_2\|/4.
$$

In step 4, if we delete a node $v \in V \setminus X$ of degree $\deg(v) < \lambda d$, then the  decrease in $\|S_2\|$ is at most $\lambda d$. Else if we delete a path $\pi$ of length $|\pi| < \lambda \ell$, then the decrease in $\|S_2\|$ is at most $\lambda \ell$. Finally, if we delete a node $x \in X$, then the decrease in $\|S_2\|$ is at most $\lambda d \cdot c_2 \ell$, since we delete at most $\lambda d$ paths each of length at most $c_2 \ell$. Then the total decrease in $\|S_2\|$ in step 4 is at most
$$
|V| \cdot \lambda d + |\Pi| \cdot \lambda \ell + |X| \cdot \lambda d \cdot c_2 \ell \leq  \lambda(2 + c_2 )\|S\|.
$$
If we choose our constant $\lambda$ to be $\lambda := \frac{c_1/2 \cdot (1-e^{-c_1/2})}{4(2+c_2)}$, then this decrease is at most $\|S_2\|/4$. We conclude that $\|S'\| \geq \|S_2\|/2 = \Theta(\|S\|)$, as desired.
\end{proof}
Now the proof of Lemma \ref{lemma:sr_cleaning_lemma} follows by using  Lemma \ref{lem:scwise_cleaning_1} and Lemma \ref{lem:paramadjust} in an argument identical to the proof of the cleaning lemma (Lemma \ref{lem:cleaning}) in Appendix \ref{app:cleaning}.

%% file: implicitbounds.tex
\section{Implicit Bounds on $\beta$ in Prior Work}

\subsection{Upper Bounds for $k=2$ \label{app:twobounds}}

Here, we repeat some arguments from prior work that implicitly show upper bounds on the value of $\beta$, translated into language directly about $\beta$.
We note that $\beta(n, p, 2) = \beta^*(n, p, 2)$, since $2$-bridges are not sensitive to ordering, and so we typically prove only the upper bounds on $\beta(n, p, 2)$.

\begin{theorem} [\cite{CE06}] \label{thm:cebp2}
$\beta(n, p, 2) = \beta^*(n, p, 2) = O(np^{1/2} + p).$
\end{theorem}
\begin{proof}
Let $S = (V, \Pi)$ be a path system with $n$ nodes, $p$ paths, and bridge girth $>2$.
By the Cleaning Lemma (Lemma \ref{lem:cleaning}), we may assume without loss of generality that all paths have length $\Theta(\ell)$, where $\ell$ is the average length in $S$.
We may also assume that $\ell$ is at least a sufficiently large constant, as otherwise the bound $O(p)$ is immediate.

There are $O(n^2)$ ordered pairs of distinct nodes in $S$.
Since $S$ does not have $2$-bridges, for each such ordered pair $(x, y)$, there is at most one path $\pi \in \Pi$ with $x <_{\pi} y$.
On the other hand, each path $\pi$ contains ${|\pi| \choose 2} = \Theta(\ell^2)$ such node pairs (note: this equality uses that $\ell$ is a large constant, and so $|\pi| \ge 2$).
We therefore have:
\begin{align*}
p \ell^2 &= O\left(n^2 \right)\\
\|S\|^2 = p^2 \ell^2 &= O\left( n^2 p \right)\\
\|S\| &= O\left(np^{1/2}\right). \tag*{\qedhere}
\end{align*}
\end{proof}

\begin{theorem} [\cite{Bodwin21}] \label{thm:bodbnp2}
$\beta(n, p, 2) = \beta^*(n, p, 2) = O\left( n^{2/3} p + n\right).$
\end{theorem}
\begin{proof}
Let $S = (V, \Pi)$ be a path system with $n$ nodes, $p$ paths, and bridge girth $>2$.
By the Cleaning Lemma (Lemma \ref{lem:cleaning}), we may assume without loss of generality that all nodes have degree $\Theta(d)$, where $d$ is the average degree in $S$.
We may also assume that $d$ is at least a sufficiently large constant, as otherwise the bound $O(n)$ is immediate.

We first claim that, for any triple of distinct paths $\pi_1, \pi_2, \pi_3 \in \Pi$, there exists at most one node $v$ in $\pi_1 \cap \pi_2 \cap \pi_3$.
To see this, suppose for contradiction that there are distinct nodes $u, v \in (\pi_1 \cap \pi_2 \cap \pi_3)$.
Notice that there must be two paths that use $u, v$ in the same order; e.g., without loss of generality, we have $u <_{\pi_1} v$ and also $u <_{\pi_2} v$.
But this implies that $\pi_1, \pi_2$ form a $2$-bridge, giving contradiction.

Meanwhile, consider an arbitrary node $v$.
There are ${\deg(v) \choose 3} = \Theta(d^3)$ triples of paths that intersect at $v$ (note: this equality uses that $d$ is a large enough constant, and so $\deg(v) \ge 3$).
We therefore have
\begin{align*}
nd^3 &= O\left(p^3\right)\\
\|S\|^3 = n^3 d^3 &= O\left(n^2 p^3\right)\\
\|S\| &= O\left(n^{2/3} p\right). \tag*{\qedhere}
\end{align*}
\end{proof}

\subsection{Bounds for $k=3$ \label{app:threebounds}}

\begin{theorem} [Tweaked Folklore Argument] \label{thm:threebound}
$\beta(n, p, 3)=O\left((np)^{2/3} + n + p\right)$
\end{theorem}
\begin{proof}
This argument can be viewed as a slightly more careful version of the standard $O((np)^{2/3} + n + p)$ upper bound on $\gamgam(n,p,6)$.
While technically slightly different, it follows the same rhythms and overall does not contain a significant new idea.

Let $S = (V, \Pi)$ be a path system with $n$ nodes, $p$ paths, and bridge girth $>3$.
By the Cleaning Lemma (Lemma \ref{lem:cleaning}), we may assume without loss of generality that all nodes have degree $\Theta(d)$ and all paths have length $\Theta(\ell)$, where $d, \ell$ are respectively the average degree and length in $S$.
We may also assume that both $d, \ell$ are sufficiently large constants, as otherwise the bound $O(n + p)$ is immediate.

Choose an arbitrary path $\pi \in \Pi$, which we will call the \emph{main path}.
We have that $\pi$ intersects $\Omega(\ell)$ nodes, and each of these nodes have $\Omega(d)$ paths of length $\ell$ each.
For a path $q$ that intersects $\pi$, let us say that the \emph{downstream part} of $q$ is the suffix following the first point at which $q, \pi$ intersect, and the \emph{upstream part} of $q$ is the prefix preceding the last point at which $q, \pi$ intersect (so if $q, \pi$ intersect at several nodes, which is conceivable so long as they use those nodes in opposite order, the downstream/upstream parts of $q$ overlap).

Let $U, D$ be the set of upstream, downstream parts of paths (respectively) that intersect $\pi$.
We claim that the subpaths in $U$ are pairwise node-disjoint from each other, and also the paths in $D$ are pairwise node-disjoint from each other.
To see that the paths in $D$ are pairwise node-disjoint, suppose for contradiction that there are paths $q_1, q_2 \in D$ that intersect at a node $v$.
Also suppose that the full paths containing $q_1, q_2$ intersect $\pi$ at nodes $x, y$, respectively.
Then we notice that these paths form a $3$-bridge with $\pi$, on the nodes $\{v, x, y\}$, reaching contradiction.
A similar argument works to show node-disjointness of paths in $U$.

This node-disjointness implies that $\|U\| + \|D\| = \Omega(d\ell^2)$, and so without loss of generality we may assume $\|D\| = \Omega(d\ell^2)$.
Since the paths in $D$ are node-disjoint we have $d \ell^2 = O(n)$, and so
$$(nd)(p^2 \ell^2) = O(n^2 p^2).$$
Since $nd = p\ell = \|S\|$, this implies
$$\|S\| = O((np)^{2/3})$$
as claimed.
\end{proof}

\begin{figure}[htbp]
  \begin{center}
    \includegraphics[width=3.0in]{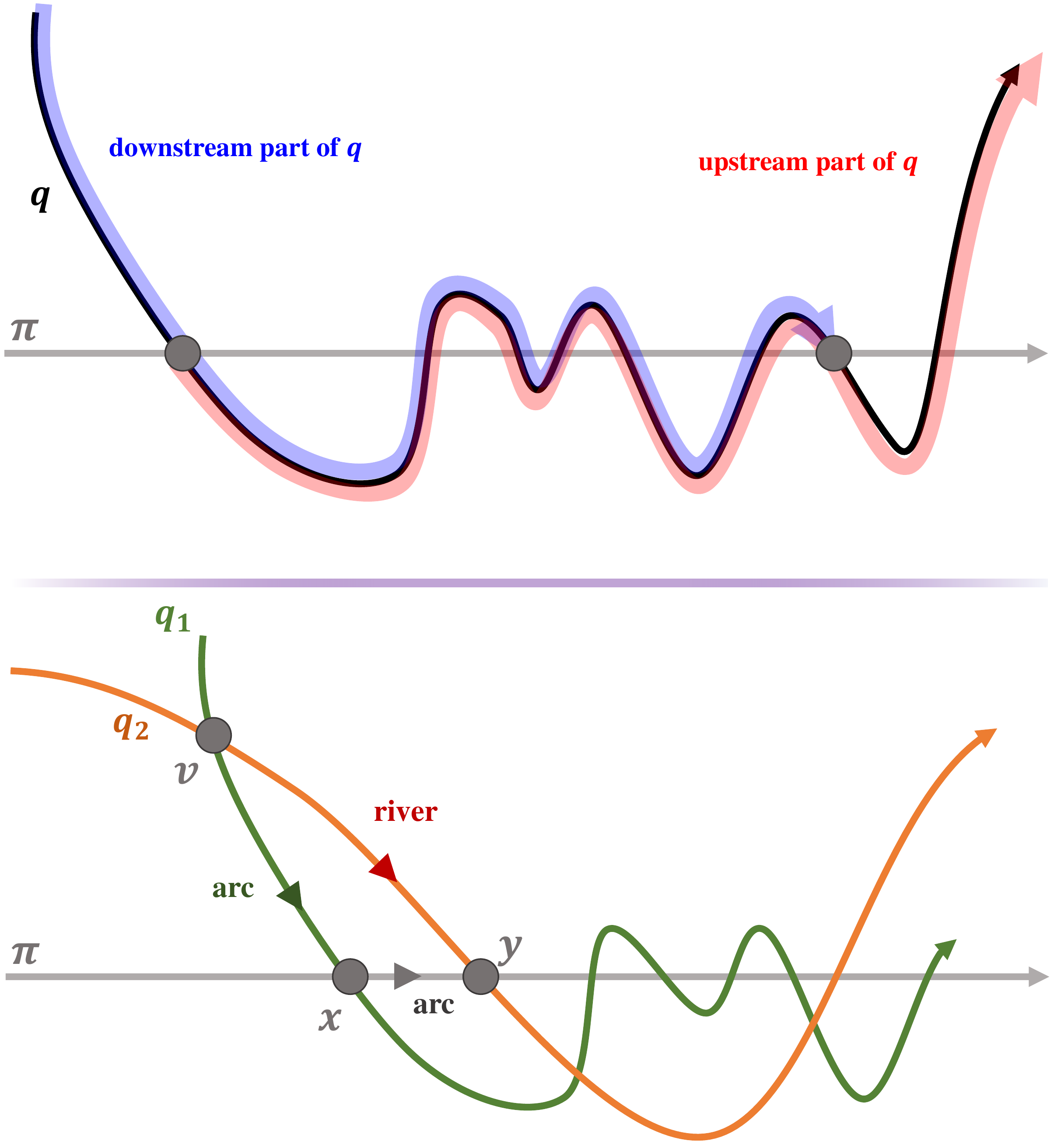}
    \end{center}
    \caption{The upper bound for $\beta(n,p,3)$. In the top is an illustration of the upstream and the downstream parts of a path $\pi$, and in the bottom the $3$-bridge on the paths $q_1,q_2$, and $\pi$.}
    \label{Figs:3bridge}
\end{figure}

    
    


    


\begin{theorem} [Based on \cite{de1997dense}] \label{thm:b3rs}
$\beta(n, p, 3) = O\left(\frac{n^2}{\rs(n)} + p\right)$.
\end{theorem}
\begin{proof}
Let $S = (V, \Pi)$ be a path system with $n$ nodes, $p$ paths, and bridge girth $>3$.
We may assume that the average path length $\ell$ is at least a large enough constant, as otherwise the bound of $O(p)$ is immediate.
We associate $S$ to an auxiliary graph $G$ as follows:
\begin{itemize}
\item Split each path $\pi \in \Pi$ into as many node-disjoint subpaths as possible of length exactly $3$ each.  (We may discard one or two nodes at the end of the path.)
Note that, since $\ell$ is a large enough constant, we change the size of $\|S\|$ by at most a constant factor over this splitting process.

\item Take a uniform-random equitable tripartition $V = V_1 \cup V_2 \cup V_3$.
For each path $\pi=(x, y, z) \in \Pi$, keep $\pi$ iff $x \in V_1, y \in V_2, z \in V_3$; otherwise delete $\pi$ from $\Pi$.
Each path survives with constant probability, and so in expectation we again change the size of $\|S\|$ by at most a constant factor.

\item Let $G$ be the bipartite graph between vertex sets $V_2, V_3$, where we include an edge $(v_2, v_3)$ iff there exists a path $\pi \in \Pi$ with $v_2 <_{\pi} v_3$.
Note that there is one edge in $G$ per path in $\Pi$, and thus $|E(G)| = \Theta(\|S\|)$, so it suffices to bound $|E(G)|$.

\item For each node $v_1 \in V_1$, define an edge subset $M[v_1] \subseteq E(G)$ as all edges $(v_2, v_3)$ where there exists a path $(v_1, v_2, v_3) \in \Pi$.
\end{itemize}

In order to bound $|E(G)|$, we will show that each edge subset $M[v_1]$ is an induced matching.
To see this, suppose for contradiction that there are distinct edges $(u_2, u_3), (v_2, v_3) \in M[v_1]$, and also an edge $(v_2, u_3) \in E(G)$.
Suppose this other edge is caused by a path $(v'_1, v_2, u_3)$.
Then we notice that the three paths $(v_1, u_2, u_3), (v_1, v_2, v_3), (v'_1, v_2, u_3)$ form a $3$-bridge, with the first path as the river, giving contradiction.

Thus each set $M[v_1]$ is an induced matching, and so $E(G)$ may be partitioned into $|V_1| < n$ induced matchings.
Thus, by definition of $\rs(n)$ we have $|E(G)| = O\left( n^2 / \rs(n) \right)$, completing the proof.
\end{proof}

\begin{corollary}
$\beta(n, p, 3) = O\left( \frac{n^2}{2^{O(\log^* n)}} + n \right)$.
\end{corollary}
\begin{proof}
Follows from the previous theorem, and by plugging in the state-of-the-art bounds on $\rs(n)$ from \cite{Fox11, MS19}.
\end{proof}

\begin{theorem}\label{thm:points3}
$\beta(n, p, 3) = \Theta((np)^{2/3})$ when $p \in \{n^{4/5}, n^{7/8}, n, n^{8/7}, n^{5/4}\}$.
\end{theorem}
\begin{proof}
The upper bound follows from Theorem \ref{thm:threebound}.
For the lower bound, for each of the given values of $p$, it is known \cite{van12} that
$$\gamgam(n, p, 6) = \Omega\left((np)^{2/3}\right);$$
that is, there are constructions of bipartite graphs with $n$ nodes on one side, $p$ nodes on the other side, and girth $>6$.
We may convert any such graph to a path system $S = (V, \Pi)$ by taking $V$ as the $n$ nodes one one side, taking $\Pi$ as the $p$ nodes on the other side, and including a node $v$ in a path $\pi$ iff $(v, \pi)$ is an edge in the graph.
The order of the nodes in each path can be chosen arbitrarily.
Notice that a $2$-bridge in $S$ corresponds to a $4$-cycle in the graph, and a $3$-bridge in $S$ corresponds to a $6$-cycle in the graph.
Since neither such cycle exists, $S$ has bridge girth $>3$, and its size is $\|S\| = \Omega((np)^{2/3})$, completing the proof.
\end{proof}